\documentclass[letterpaper,11pt]{article}
\usepackage[utf8]{inputenc}

\pdfoutput=1

\usepackage{amsmath, amsthm, amssymb}
\usepackage{algpseudocode,algorithm}
\usepackage{algorithmicx}
\usepackage{mathtools}
\usepackage[numbers,sort]{natbib}
\usepackage{comment} % enables comment environment
\usepackage[suppress]{color-edits} % author comments with colors
\usepackage{tcolorbox}
\usepackage{xfrac}
\usepackage{hyperref}
\usepackage{multirow}
\usepackage{caption}
\usepackage{bm}
\usepackage{newfloat}
\usepackage{enumitem}
\usepackage{wrapfig}
\usepackage{setspace}
\usepackage{xfrac}
\usepackage{tocloft}
\let\LaTeXStandardContentsName\contentsname
\renewcommand{\contentsname}{\phantom{\LaTeXStandardContentsName}}%

\usepackage[T1]{fontenc}
\usepackage{tgtermes}

\usepackage[margin=1in]{geometry}

\allowdisplaybreaks

\title{Optimal Strategies of Blotto Games: Beyond Convexity\thanks{A portion of this work was completed while some of the authors were visiting Simons Institute for Theory of Computing.}}

\author{
\texorpdfstring{\hspace*{-8pt}}{}%
\begin{tabular}{c} Soheil Behnezhad\\ University of Maryland\\ \\
MohammadTaghi Hajiaghayi\\University of Maryland\end{tabular} \and
\begin{tabular}{c} Avrim Blum\\TTI-Chicago\\\\
Christos H. Papadimitriou\\ Columbia University \end{tabular}\and
\begin{tabular}{c} Mahsa Derakhshan \\ University of Maryland \\ \\Saeed Seddighin \\ University of Maryland\end{tabular}
\texorpdfstring{\hspace*{-8pt}}{}%
}

\date{}

\DeclareMathOperator{\poly}{poly}

\addauthor{sb}{blue}    % sb for Soheil
\addauthor{md}{red}  % md for Mahsa
\addauthor{ss}{purple} %ss for Saeed

\newtheorem{theorem}{Theorem}[section]
\newtheorem{lemma}[theorem]{Lemma}

\newtheorem{corollary}[theorem]{Corollary}

\newtheorem{definition}[theorem]{Definition}
\newtheorem{claim}[theorem]{Claim}

\newtheorem{observation}[theorem]{Observation}

\definecolor{mygreen}{RGB}{20,100,60}

\hypersetup{
     colorlinks=true,
     citecolor= mygreen,
     linkcolor= black
}

\newcommand{\etal}[0]{\textit{et al.}}

\renewcommand{\paragraph}[1]{\vspace{0.2cm} \noindent \textbf{#1}}

\algnewcommand{\IIf}[1]{\State\algorithmicif\ #1\ \algorithmicthen}
\algnewcommand{\EndIIf}{\unskip\ \algorithmicend\ \algorithmicif}

\newcommand{\restatethm}[2]{\vspace{0.35cm}\noindent \textbf{Theorem~#1} (restated). {\em #2}\vspace{0.15cm}}

\newcommand{\resultsrestatethm}[3]{\vspace{0.15cm}\mybox{\noindent \textbf{Theorem~#1} (#2). {\em #3}}\vspace{0.10cm}}

\newcommand{\resultsrestatethms}[3]{\vspace{0.15cm}\mybox{\noindent \textbf{Theorems~#1} (#2). {\em #3}}\vspace{0.10cm}}

\newcommand{\resultsrestatethmnorestated}[2]{\vspace{0.15cm}\mybox{\noindent \textbf{Theorems~#1}. {\em #2}}\vspace{0.10cm}}

\DeclareRobustCommand{\mybox}[2][gray!20]{%
\begin{tcolorbox}[%% Adjust the following parameters at will.
        left=0pt,
        right=0pt,
        top=0pt,
        bottom=0pt,
        colback=#1,
        colframe=#1,
        width=\dimexpr\textwidth\relax, 
        enlarge left by=0mm,
        boxsep=5pt,
        arc=0pt,outer arc=0pt,
        ]
        #2
\end{tcolorbox}
}

\DeclareMathOperator*{\argmin}{arg\,min}
\DeclareMathOperator*{\argmax}{arg\,max}
\renewcommand{\b}[1]{\ensuremath{\bm{\mathrm{#1}}}}

\newcommand{\utilitya}[2]{\ensuremath{\mathsf{u}_1(#1, #2)}}
\newcommand{\utilityb}[2]{\ensuremath{\mathsf{u}_2(#1, #2)}}

\newcommand{\pureseta}[0]{\ensuremath{\mathcal{S}_1}}
\newcommand{\puresetb}[0]{\ensuremath{\mathcal{S}_2}}

\newcommand{\puresetaheavy}[0]{\ensuremath{\mathcal{S}^h_1}}
\newcommand{\puresetbheavy}[0]{\ensuremath{\mathcal{S}^h_2}}

\DeclareFloatingEnvironment[fileext=frm,name=LP]{lp}

\newcommand{\maxmin}[2]{\ensuremath{(#1,#2)\text{-}\textsf{maximin}}}
\newcommand{\maxminc}[3]{\maxmin{#1}{#2} \cmixedstrategy{#3}}
\newcommand{\cmixedstrategy}[1]{\ensuremath{#1\textsf{-strategy}}}
\newcommand{\cmixedstrategies}[1]{\ensuremath{#1\textsf{-strategies}}}

\newcommand{\A}{\ensuremath{L_1}}
\newcommand{\B}{\ensuremath{L_2}}
\newcommand{\AB}{\ensuremath{L_{12}}}
\newcommand{\wmax}{\mathsf{w_{max}}}
\newcommand{\sign}[1]{\ensuremath{\sigma_{#1}}}
\newcommand{\signtwo}[2]{\ensuremath{\mathsf{\sigma}(#1, #2)}}
\newcommand{\cbinstance}[3]{\ensuremath{\mathcal{B}(#1, #2, #3)}}

\newcommand{\profile}[1]{\ensuremath{\rho(#1)}}
\newcommand{\winningsubset}[2]{\ensuremath{\mathsf{W}_{#1}(#2)}}

\newcommand{\nonlosingset}[1]{\ensuremath{\mathcal{N}(#1)}}

\begin{document}

\maketitle

%\vspace{-0.5cm}
\begin{abstract}
The {\em Colonel Blotto} game, first introduced by Borel in 1921, is a well-studied game theory classic. Two colonels each have a pool of {\em troops} that they divide simultaneously among a set of {\em battlefields}. The winner of each battlefield is the colonel who puts more troops in it and the overall utility of each colonel is the sum of weights of the battlefields that s/he wins. Over the past century, the Colonel Blotto game has found applications in many different forms of competition from advertisements to politics to sports.

Two main objectives have been proposed for this game in the literature: (i) maximizing the guaranteed expected payoff, and (ii) maximizing the probability of obtaining a minimum payoff $u$. The former corresponds to the conventional utility maximization and the latter concerns scenarios such as elections where the candidates' goal is to maximize the probability of getting at least half of the votes (rather than the expected number of votes). In this paper, we consider both of these objectives and show how it is possible to obtain (almost) optimal solutions that have few strategies in their support.

One of the main technical challenges in obtaining bounded support strategies for the Colonel Blotto game is that the solution space becomes non-convex. This prevents us from using convex programming techniques in finding optimal strategies which are essentially the main tools that are used in the literature. However, we show through a set of structural results that the solution space can, interestingly, be partitioned into polynomially many disjoint convex polytopes that can be considered independently. Coupled with a number of other combinatorial observations, this leads to polynomial time approximation schemes for both of the aforementioned objectives.

We also provide the first complexity result for finding the maximin of Blotto-like games: we show that computing the maximin of a  generalization of the Colonel Blotto game that we call {\sc General Colonel Blotto} is exponential time-complete. 
\end{abstract}

%\setcounter{tocdepth}{2}
%\vspace{-1cm}
%\setlength{\cftbeforesecskip}{-1.5pt}
%\tableofcontents
%\thispagestyle{empty}

\clearpage
\setcounter{page}{1}

\section{Introduction}
The {\em Colonel Blotto} game, originally introduced by Borel in 1921 \cite{B21, B53, F53a, F53b}, is among the first mathematically formulated strategic situations and has become one of the game theory classics over the past century. In this game, two players --- which we refer to as player 1 and player 2 --- are fighting over a set of {\em battlefields}. Each player has a number of {\em troops} to distribute across the battlefields and the winner of each of the battlefields is the player who puts more troops in it. The overall utility of each player is the sum of weights of the battlefields that s/he wins. Although originally introduced to consider a war situation (as the terminology implies,) the Colonel Blotto game has found applications in many different forms of competition from advertisements to politics to sports \cite{myerson1993incentives, merolla2005play, laslier2002distributive, kovenock2010conflicts, kovenock2012coalitional, chowdhury2009experimental}.

The goal in the Colonel Blotto game is to find optimal (i.e., {\em maximin}) strategies of the players. A mixed strategy is maximin if it maximizes the guaranteed {\em expected} utility against any strategy of the opponent. For certain applications of the Colonel Blotto game --- such as the United States presidential elections, which we elaborate on in the next paragraph --- maximizing the expected utility is not desirable. In such scenarios, we need to maximize the probability of obtaining a minimum utility $u$. This has been captured by the concept of \maxmin{u}{p} strategies in the literature \cite{behnezhad2018battlefields}. A strategy is \maxmin{u}{p} if it guarantees receiving a utility of at least $u$ with probability at least $p$. In this paper we consider both of these objective functions.

Perhaps the most notable application of the Colonel Blotto game is in the U.S. presidential election race. Every state has a number of electoral votes and the winning candidate is the one who receives the majority of these votes. In most\footnote{All states except Maine and Nebraska.} of the states, a {\em winner-take-all} rule determines the electoral votes. That is, the candidate who gets the popular vote in a state receives all the electoral votes of that state. Thus, the candidates have to strategically decide on how they spend their resources such as funds, staff, etc., in different states. This can be modeled by a Colonel Blotto instance in which each battlefield corresponds to a state and the colonels' troops correspond to the candidates' resources. Here, maximizing the expected utility (i.e., the expected number of electoral votes) does not necessarily maximize the likelihood of winning the race. Instead, the actual goal is to maximize the probability of obtaining more than half (i.e., at least 270) of the electoral votes. As such, we need to find a \maxmin{270}{p} strategy with maximum possible $p$.

%Nash equilibrium has been the most prominent solution concept for non-cooperative games, including Colonel Blotto. Although from an algorithmic perspective, computing a Nash equilibrium is generally PPAD-hard \cite{daskalakis2009complexity, daskalakis2005three, chen2006computing, chen2006settling, goldberg2006reducibility}, many natural games of interest admit polynomial time algorithms. For example, one can find the Nash equilibrium of  any constant-sum normal form\footnote{A game is in {\em normal form} if it is represented in the input by a matrix that indicates the payoff of each of the players for every combination of their pure strategies. This means that the players can have only polynomially (in the input size) many  strategies in normal-form games.} game in polynomial time via standard linear programming techniques. It is much more challenging, or in some cases even NP-hard, to compute the Nash equilibrium of even constant-sum games that can be represented in a more succinct way. In the Colonel Blotto game for example, a player with $n$ troops has up to $\binom{n+k-1}{k-1}$ pure strategies over $k$ battlefields which is exponentially larger than the input size.

Maximin strategies, or equivalently Nash equilibrium for constant-sum games such as Colonel Blotto, are often criticized for that they may be too complicated (see e.g., \cite{lipton2003playing, simon1982models, rubinstein1998modeling}). That is, even if we are able to find such solution in polynomial time, we may not be able to deploy it since the equilibria can have a large support\footnote{The {\em support} of a mixed strategy is the set of all pure strategies to which a non-zero probability is associated.}.  In the case of the Blotto game, the potential size of the support is enormous, while every possible pure strategy in the support requires a prior (often costly) setup. %For instance, having very few potential strategies in the support is essential in the presidential election example since each possible campaigning strategy requires a prior setup that is both time consuming and expensive. 
Therefore it is tempting to find a near-optimal strategy that uses very few pure strategies, and is near optimum against the opponent's best response --- the main goal of the present paper.

However, limiting the support size ofter renders the game intractable. For instance, the decision problem of existence of a Nash equilibrium  when the support size is bounded by a given number is NP-hard even in two-player zero-sum games \cite{gilboa1989nash}, while this problem is unlikely to be fixed parameter tractable when the problem is parametrized by the support size ~\cite{estivill2009computing}. These results imply that in order to find optimal strategies with bounded support, we need to use structural properties of the game at hand, even when the players have polynomially many strategies. It becomes even more challenging for succinct games such as Colonel Blotto wherein the strategy space itself is exponentially large.

Recent studies have made significant progress in understanding the optimal response problem in Colonel Blotto when the support size is unbounded \cite{behnezhad2018battlefields, ahmadinejad2016duels,behnezhad2017faster,roberson2006colonel,gross1950continuous,hart2008discrete,shubik1981systems,hortala2012pure,roberson2012non}. Note that the unbounded case is also challenging to solve due to the exponential number of pure strategies of the players in the Colonel Blotto game. In spite of that, for maximizing the expected utility, one can obtain optimal strategies in polynomial time by algorithms of Ahmadinejad \etal~\cite{ahmadinejad2016duels} and Behnezhad \etal~\cite{behnezhad2017faster}. The case of \maxmin{u}{p} strategies is generally harder to solve. However, Behnezhad \etal~\cite{behnezhad2018battlefields} show that it is possible to obtain a logarithmic factor approximation in polynomial time.

All of the results mentioned above rely crucially on the fact that the support size is unbounded. The challenges in obtaining bounded support strategies turn out to be entirely different. On one hand, for the choice of each pure strategy in the support we still have exponentially\footnote{Or even an unbounded number of strategies for the continuous variant of Colonel Blotto where the troops are capable of fractional partitioning.} many possibilities. On the other hand, we show that bounding the support size makes the solution space non-convex. The latter prevents us from using convex programming techniques in finding optimal strategies --- which are essentially the main tools that are used in the literature for solving succinct games in polynomial time \cite{immorlica2011dueling, ahmadinejad2016duels, behnezhad2018battlefields, behnezhad2017faster}. However, we show through a set of structural results that the solution space can, interestingly, be partitioned into polynomially many disjoint convex polytopes, allowing us to consider each of these sub-polytopes independently. Combined with a number of other techniques that are highlighted in Section~\ref{sec:results}, this leads to polynomial time approximation schemes (PTASs) for both the expectation case and the case of \maxmin{u}{p} strategies. %See Section~\ref{sec:results} for an overview of our main results and techniques.

We also provide the first complexity result for finding the maximin of Blotto-like games: we show that computing the maximin of a  generalization of the Colonel Blotto game that we call {\sc General Colonel Blotto}  --- roughly, the utility is a general function of the two allocations --- is exponential time-complete. 

%Note that when we bound the support size of mixed strategies, a Nash equilibrium may not exist, and in fact the decision problem of existence of a Nash equilibrium in such scenarios becomes NP-hard as it is the case for most decision problems regarding Nash equilibria \sbcomment{cite}. As such, we use the closely related notion of a {\em maximin} strategy which maximizes the guaranteed  payoff of a player can be used.

%In this paper, we consider two variants of the Colonel Blotto game, namely the {\em continuous} variant where the resources are capable of continuous partition, and the {\em discrete} variant. We also consider the two main objectives that have been considered for the Colonel Blotto game in the literature: (i) maximizing the expected payoff, and (ii) maximizing the probability of obtaining a specific minimum utility. We show how it is possible to obtain an (almost) optimal solution in all these scenarios when the support size is bounded by a given small number. Our results and techniques are summarized in Section~\ref{sec:results}.

%From an algorithmic perspective, it is generally PPAD-hard to find Nash equilibrium, even for 2-player normal form games \cite{}. For zero-sum games,\footnote{Or equivalently, constant-sum games.} however, if the players have only polynomially many pure strategies, we can find Nash equilibrium in polynomial time using standard linear programming techniques. 

\section{Preliminaries}
%\subsection{Notations}
Throughout the paper, for any integer $n$, we use $[n]$ to denote the set $\{1, 2, \ldots, n\}$. We denote the vectors by bold fonts and for every vector $\b{v}$, use $v_i$ to denote its $i$th entry.

\paragraph{The Colonel Blotto game.} The Colonel Blotto game is played between two players which we refer to as player 1 and player 2. Any instance of the game can be represented by a tuple \cbinstance{n}{m}{\b{w}} where $n$ and $m$ respectively denote the number of {\em troops} of player 1 and player 2, and $\b{w} = (w_1, \ldots, w_k)$ is a vector of length $k$ of positive integers denoting the {\em weight} of $k$ {\em battlefields} on which the game is being played.

 A pure strategy of each of the players is a partitioning of his troops over the battlefields. That is, any pure strategy of player 1 (resp. player 2) can be represented by a vector $\b{x} = (x_1, \ldots, x_k)$ of length $k$ of non-negative numbers where $\sum_{i\in[k]} x_i \leq n$ (resp. $\sum_{i\in[k]} x_i \leq m$). In the {\em discrete} variant of the game, the number of troops that are assigned to each of the battlefields must be non-negative integers. In contrast, in the {\em continuous} variant, any assignment with non-negative real values is considered valid. Throughout the paper, we denote respectively by \pureseta{} and \puresetb{} the set of all valid pure strategies of player 1 and player 2.

 Let $\b{x}$ and $\b{y}$ be the pure strategies that are played respectively by player 1 and player 2. Player 1 wins battlefield $i$ if $x_i > y_i$ and loses it otherwise. The winner of battlefield $i$ gets a partial utility of $w_i$ on that battlefield, and the overall utility of each player is the sum of his partial utilities on all the battlefields. More precisely, the utilities of players 1 and 2, which we respectively denote by $\utilitya{\b{x}}{\b{y}}$ and $\utilityb{\b{x}}{\b{y}}$, are as follows:
%\begin{equation}
%	\utilitya{\b{x}}{\b{y}} = \sum_{i \in \mathcal{W}_1} w_i \text{ where } \mathcal{W}_1 = \{ i : i\in[k], x_i > y_i\}, \qquad \utilityb{\b{x}}{\b{y}} = (\sum_{i\in [k]} w_i) - \utilitya{\b{x}}{\b{y}}.
%\end{equation}
\begin{equation*}
	\utilitya{\b{x}}{\b{y}} = \sum_{i \in [k] : x_i > y_i} w_i, \hspace{2cm} \utilityb{\b{x}}{\b{y}} = \sum_{i \in [k] : x_i \leq y_i} w_i.
\end{equation*}
Note that in the definition above, we break the ties in favor of player 2, i.e., when both players put the same number of troops on a battlefield, we assume the winner is player 2. 

Also, we define the \textit{uniform} Colonel Blotto game to be a special case of Colonel Blotto in which all of the battlefields have the same weights, i.e., $w_1 = w_2 = \ldots = w_k = 1$.

%\sscomment{Remove the solution concept section and bring the whole content here. We do not define a new notion of approximation or something new so we don't need to dedicate an entire section to define the solution concept.}

\paragraph{The objective.} 
The guaranteed expected utility of a (possibly mixed) strategy $\b{x}$ of player 1 is $u$, if against any (possibly mixed) strategy $\b{y}$ of player 2, we have $\mathbb{E}_{\b{x}' \sim \b{x}, \b{y}' \sim \b{y}}[\utilitya{\b{x}'}{\b{y}'} ] \geq u.$
A strategy is {\em maximin} if it maximizes this guaranteed expected utility.

In this paper, we are interested in bounded {\em support} strategies. The support of a mixed strategy is the set of pure strategies to which it assigns a non-zero probability. We call a strategy \cmixedstrategy{c} if its support has size at most $c$. A maximin \cmixedstrategy{c} is a \cmixedstrategy{c} that has the highest guaranteed expected utility among all \cmixedstrategies{c}. Observe that we do not restrict the adversary to play a \cmixedstrategy{c} here. In fact, given the mixed strategy of a player, it is well-known that the best response of the opponent can be assumed to be a pure strategy w.l.o.g. This means that an opponent that can respond by only one pure strategy is as powerful as an opponent that can play mixed strategies as far as maximin strategies are concerned.

Another standard objective for many natural applications of the Colonel Blotto game is to maximize the probability of obtaining a utility of at least $u$. This has been captured in the literature by the notion of \maxmin{u}{p} strategies.
\begin{definition}
	For any two-player game, a (possibly mixed) strategy $\b{x}$ of player 1 is called a \maxmin{u}{p} strategy, if for any (possibly mixed) strategy $\b{y}$ of player 2, $$\Pr_{\b{x}' \sim \b{x}, \b{y}' \sim \b{y}} [\mathsf{u}(\b{x}', \b{y}') \geq u] \geq p.$$
\end{definition}
\noindent In such scenarios, for a given minimum utility $u$, our goal is to compute (or approximate) a \maxmin{u}{p} with maximum possible $p$. Similar to the expected case, we are interested in bounded support strategies. That is, given $u$ and $c$, our goal is to compute (or approximate) a \maxmin{u}{p} \cmixedstrategy{c} with maximum possible $p$. Again, we do not restrict the adversary's strategy.

%In this paper we consider both expected maximin and \maxmin{u}{p} strategies when the support size is bounded by a given constant $c$. More precisely, we attempt to answer the following question: {\em Given a Colonel Blotto instance which is guaranteed to admit a \maxmin{u}{p} with a support of size up to $c$ for player 1, compute (or approximate) such a strategy in polynomial time.}

\section{Our Results \& Techniques}\label{sec:results}
Throughout this paper, we discuss the optimal strategies of the Colonel Blotto game when the support size is small. That is, when player~1 can only randomize over at most $c$ pure strategies and wishes to maximize his utility. Our main interest is in finding (almost) optimal \maxmin{u}{p}  strategies as discussed above. Nonetheless, we show in Section~\ref{sec:exp} that our results carry over to the conventional definition of the maximin strategies and can be used when the goal is to maximize the expected utility.

When randomization is not allowed, any \maxmin{u}{p} \cmixedstrategy{1} with $p > 0$ is also \maxmin{u}{1}. Thus, the problem boils down to finding a pure strategy with the maximum guaranteed payoff. However, when randomization over two pure strategies is allowed, we may play two pure strategies with different probabilities $q$ and $1-q$ and thus the problem of finding a \maxmin{u}{p} \cmixedstrategy{2} with $p \neq 1$ does not necessarily reduce to the case where the goal is to find a single pure strategy. As an example, when $n = m = 2$ and we have two battlefields with equal weights ($w_1 = w_2 = 1$), a \maxmin{1}{\sfrac{1}{2}} \cmixedstrategy{2} can be obtained by selecting a battlefield at random and placing two troops on the selected battlefield. However, in this example (or in any example with $m=n$), no \cmixedstrategy{2} is \maxmin{1}{p} for $p>\sfrac{1}{2}$ since the opponent can just select our higher-probability strategy and copy it, ensuring we get utility 0 with probability at least $\sfrac{1}{2}$. More generally, for $c=2$, even if $n\leq m$, a simple observation shows that the only interesting cases are when $p = 1$ (which reduces to the $c =1$ case) or when $p = \sfrac{1}{2}$. The idea is that when $p > \sfrac{1}{2}$ holds, existence of a \maxmin{u}{p} \cmixedstrategy{2} for an arbitrary $u$ implies that of a \maxmin{u}{1} \cmixedstrategy{2}. Similarly, if a \cmixedstrategy{2} is \maxmin{u}{p} for some $u$ and $0 < p < \sfrac{1}{2}$, one can modify the same strategy to make it \maxmin{u}{\sfrac{1}{2}}. Therefore, for $c = 2$, the problem is (computationally) challenging only when a \maxmin{u}{\sfrac{1}{2}} \cmixedstrategy{2} is desired. Thus, for $c = 2$, the problem essentially reduces to finding two pure strategies $\b{x},\b{x}'$ such that no strategy of the opponent can prevent {\em both} $\b{x}$ and $\b{x}'$ from getting a utility of at least $u$. A similar, but more in-depth analysis gives us the same structure for $c = 3$. That is, in this case, we may look for a \maxmin{u}{\sfrac{1}{3}} or a \maxmin{u}{\sfrac{2}{3}} \cmixedstrategy{3}.  This implies that in an optimal solution, we look for three strategies $\b{x},\b{x}',\b{x}''$ such that no strategy of the opponent prevents two of or all of (depending on whether $p = \sfrac{1}{3}$ or $p=\sfrac{2}{3}$) $\b{x}, \b{x}'$ and $\b{x}''$ from getting a utility of at least $u$.

It is surprising to see that this structure breaks when considering more than three pure strategies ($c \geq 4$).  For instance, consider an instance of the Colonel Blotto game with 4 battlefields ($k=4$) in which the players have $4$ and $6$ troops, respectively ($n = 4$, $m = 6$). Let the weights of the first 3 battlefields be $5$ and the weight of the last battlefield be $10$. In this example, the goal of player 1 is to obtain a utility of at least $10$ with the highest probability. One can verify with an exhaustive search that if player 1 were to randomize over 4 pure strategies with equal probability, he could guarantee a utility of at least 10 with probability no more than $\sfrac{1}{4}$.\footnote{We have verified this claim with a computer program.} However, we present in Table~\ref{table:nonuinformprob}, a  \maxmin{10}{\sfrac{2}{5}} \cmixedstrategy{4} for player 1 that plays 4 pure strategies with non-uniform probabilities.

\newcommand{\thead}[2]{\parbox{2cm}{\setstretch{0.7} \centering \vspace{0.1cm} \textbf{#1} \\ #2 \vspace{0.1cm}}}

\newcommand{\strattable}[2]{\textbf{Strategy #1}, played with probability #2.}

\begin{table}[h]
\centering
\begin{tabular}{|c|c|c|c|c|}
\hline
& \thead{Battlefield 1}{$w_1 = 5$} & \thead{Battlefield 2}{$w_2 = 5$} & \thead{Battlefield 3}{$w_3=5$} & \thead{Battlefield 4}{$w_4=10$} \\ \hline
\strattable{1}{\sfrac{2}{5}} & 0                        & 0                        & 0                      & 4                       \\ \hline
\strattable{2}{\sfrac{1}{5}} & 1                        & 1                        & 2                      & 0                       \\ \hline
\strattable{3}{\sfrac{1}{5}}  & 1                        & 2                        & 1                      & 0                       \\ \hline
\strattable{4}{\sfrac{1}{5}}  & 2                        & 1                        & 1                      & 0                       \\ \hline
\end{tabular}
\caption{A \maxmin{10}{\sfrac{2}{5}} \cmixedstrategy{4} for player 1 on instance \cbinstance{4}{6}{(5, 5, 5, 10)}.}
\label{table:nonuinformprob}
\end{table}

In order to design an algorithm for finding $\maxmin{u}{p}$ strategies that does not lose on $p$, it is essential to understand how player 1  randomizes over his strategies in an optimal solution. We begin in Section~\ref{sec:probdist} by showing that when randomization is allowed only on a small number of pure strategies, the number of different ways to distribute the probability over the pure strategies in an optimal solution is limited. That is, one can list a  number of probability distributions and be sure that at least one of such probability distributions leads to an optimal solution. This structural property is general and applies to any game so long as \maxmin{u}{p} strategies are concerned.
	
\resultsrestatethm{\ref{thm:cstrategyprobabilities}}{restated informally}{When randomization is only allowed on a constant number of pure strategies, the number of possible probability distributions for an optimal \maxmin{u}{p}  strategies is limited by a constant.}

Theorem~\ref{thm:cstrategyprobabilities} is proven via a combinatorial analysis of the optimal solutions. On one hand, we leverage the optimality of the solution to argue that $p$ cannot be improved. On the other hand, we use the maximality of $p$ to derive relations between the probabilities that the pure strategies are played with. Finally, we use these relations to narrow down the probabilities a to a small set.

Indeed, a consequence of Theorem~\ref{thm:cstrategyprobabilities} is that when we fix a utility $u$ and wish to find a \maxmin{u}{p} \cmixedstrategy{c}  with highest $p$, $p$ can only take a constant number of values. This observation makes the problem substantially easier as we can iterate over all possible probability profiles and solve the problem separately for each profile. Thus, from now on, we assume that we fix a probability distribution over the pure strategies, and the goal is to select $c$ strategies to be played according to the fixed probabilities. We assume that (i) $u$ and $p$ are given in the input, (ii) we are guaranteed that a \maxmin{u}{p} \cmixedstrategy{c}  exists, and (iii) the goal is to compute/approximate such a strategy.

\paragraph{Continuous Colonel Blotto.} For the continuous case, we can represent each pure strategy with a vector of size $k$ of real values indicating the number of troops that is placed on each battlefield. Thus, a \cmixedstrategy{c} can be represented by $c$ vectors of length $k$ (given that we have fixed the probability distribution over the pure strategies). A simple observation (also pointed out by Behnezhad \textit{et al.}~\cite{behnezhad2018battlefields}) is that when the goal is to find a pure maximin strategy, the solution is essentially a convex set with respect to the representations. That is, if $\b{x}^1$ and $\b{x}^2$ are both \maxmin{u}{p}, any strategy whose representation is a convex combination of the representations of $\b{x}^1$ and $\b{x}^2$  is also  \maxmin{u}{p}. The pure maximin problem coincides with our setting when $c = 1$. However, it is surprising to see that for $c > 2$, the solution may not be convex, even though we use the same approach to represent the strategies.

As an example, imagine we have two battlefields with equal weights (say 1) and $n = m = 2$. Indeed a \maxmin{1}{\sfrac{1}{2}} strategy can be obtained for player 1 by randomizing over $(0,2)$  and $(2,0)$ uniformly. Similarly, randomizing over $(2,0)$ and $(0,2)$ (the order is changed) gives us the same guarantee. However, a convex combination of the two strategies plays $(1,1)$ deterministically and loses both battlefields against the strategy $(1,1)$ of player 2. The situation may be even worse as one can construct a delicate instance whose solution set is the union of up to $2^{\Omega(k)}$ convex regions no two of which make a convex set when merged. 

A key observation that enables us to compute/approximate the solution is the \textit{partial convexity} of the solution. That is, we show that although the solution set is not necessarily convex, one can identify regions of the space where for each region, the solution is convex. To be more precise, denote by $\b{x}^1,\b{x}^2,\ldots,\b{x}^c$ the representations of the strategies. In this representation, $x^i_j$ denotes the number of troops that $i$'th pure strategy places on the $j$'th battlefield. This gives us a $ck$ dimensional problem space $[0,n]^{ck}$. Divide this space by $k\binom{c}{2}$ hyperplanes each formulated as $x^{i}_j = x^{i'}_j$ for some $i,i' \in [c]$ and $j \in [k]$. Partial convexity implies that the solution space in each region is convex, and as a result, gives us an exponential time solution for the problem in the continuous setting. Roughly speaking, since the solution is convex in each region, we can iterate over all regions and solve the problem separately using a linear program for each region. However, we have exponentially many regions and therefore the running time of this approach is exponentially large.

\begin{figure}[h]
	\centering
	\includegraphics[scale=1]{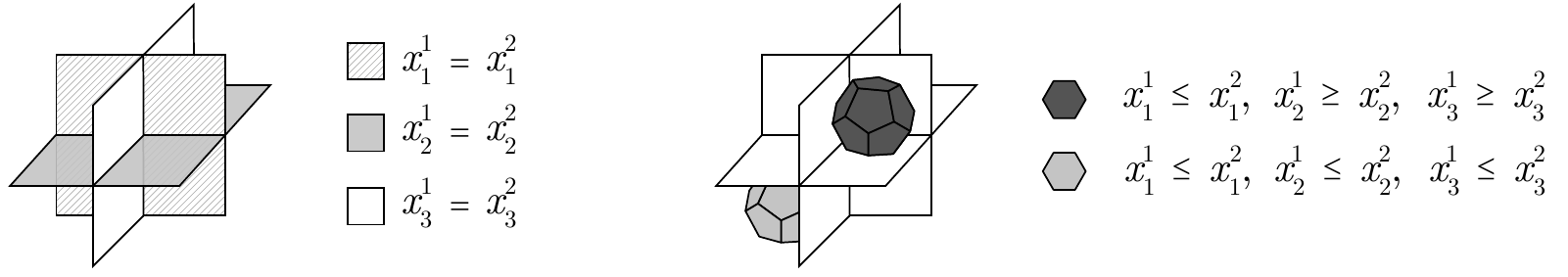}
	\caption{An example of how the solution space is decomposed for the case when $k=3$ and $c=2$. The figure on the left illustrates 3 hyperplanes $x^1_j = x^2_j$. These hyperplanes divide the solution space into 8 regions and for all the points in each region, as illustrated in the figure on right, $x^1_j$ and $x^2_j$ for $j \in [3]$ compare in the same way.}
\end{figure}

The above algorithm can be modified to run in polynomial time in the uniform setting. The high-level idea is that when the battlefields have equal weights, there is a strong symmetry between the solutions of the regions. Based on this, we show that in the uniform setting, one only needs to search a polynomial number of regions for a solution. This idea along with the partial convexity gives us a polynomial time solution for the uniform setting.

\resultsrestatethm{\ref{theorem:one}}{restated informally -- proven in Section~\ref{sec:contc2} for $c=2$ and extended in Section~\ref{sec:contcg} to $c > 2$}{There exists a polynomial time solution for finding a \maxmin{u}{p} \cmixedstrategy{c} for Colonel Blotto in the continuous case when all the battlefields have the same weight.}

Indeed, the uniform setting is a very special case since we are essentially indifferent to the battlefields. When we incorporate the weights of the battlefields, we no longer expect the solutions to be symmetric with respect to the battlefields. However, one may still observe a weak notion of symmetry between the regions. Recall that we denote the weights of the battlefields with $w_1, w_2, \ldots, w_k$. Let us lose a factor of $1+\epsilon$ in the utility and assume w.l.o.g. that each weight $w_i$ is equal to $(1+\epsilon)^j$ for some integer $j \geq 0$. Assuming that the maximum weight is polynomially bounded (we only make this assumption for the sake of simplicity and our solution does not depend on this constraint), the number of different battlefield weights is logarithmically small. Thus, we expect many battlefield to have equal weights which as a result makes the solution regions more symmetric. However, this idea alone gives us a quasi-polynomial time algorithm  for searching the regions as we may have a logarithmic number of different battlefield weights. To reduce the running time to polynomial, we need to further prune the regions of the solution to polynomially many. Indeed, we show that it suffices to search over a polynomial number of regions if we allow an approximate solution. Via this observation, we can design a polynomial time algorithm that approximates the solution within a factor $1+\epsilon$. This settles the problem for the continuous setting.

\resultsrestatethm{\ref{thm:nonuniform-c}}{restated informally -- proven in Section~\ref{sec:contc2} for $c=2$ and extended in Section~\ref{sec:contcg} to $c > 2$}{\maxmin{u}{p} \cmixedstrategy{c} strategies of the Colonel Blotto admit a PTAS in the continuous setting.}

\paragraph{Discrete Colonel Blotto.} For the discrete setting, we take a rather different approach. The main reason is that in this setting, even if we are guaranteed that the solution is convex, we cannot use LP's to compute/approximate a solution. Behnezhad \textit{et al.}~\cite{behnezhad2018battlefields} give a 2 approximation algorithm that finds a pure maximin strategy for the discrete setting. Indeed, this solution can be used to get a 2 approximate solution for the case of $c=1$. We both extend their algorithm to work for $c \geq 2$ and also devise a \textit{heavy-light decomposition} to improve the approximation factor to $1+\epsilon$. Both our extension and  decomposition techniques are novel.

Behnezhad \textit{et al.}~\cite{behnezhad2018battlefields} introduce the notion of a \textit{weak adversary}. Roughly speaking, they define a relaxed best response for player 2 that does not maximize the utility of player 2, but instead, approximately maximizes his utility. We call a player that plays a relaxed best response, a weak adversary. By proposing a greedy algorithm for the weak adversary, Behnezhad \textit{et al.}~\cite{behnezhad2018battlefields} show that the payoff of the weak adversary and the actual adversary differ by the value of at most one battlefield. That is, if the weights of the battlefields are bounded by $\wmax$, then the difference between the utility of an adversary and that of a weak adversary is always bounded by $\wmax$. Next, they show that a dynamic program can find a pure strategy of player 1 that performs best against a relaxed adversary and turn this algorithm into a 2 approximate solution for the problem, by considering two cases $2 \wmax \geq u $ and $2 \wmax < u$. 

Indeed losing an additive error of $\wmax$ may hurt the approximation factor a lot when the desired utility $u$ is not much more than $\wmax$. Thus, in order to improve the approximation factor, one needs to design a separate algorithm for the heavy battlefields. To this end, we introduce our heavy-light decomposition. We define a threshold $\tau \approx \epsilon u$ and call a battlefield $i$ heavy if $w_i > \tau$ and light otherwise. In addition to this, we assume w.l.o.g. that $\wmax \leq u$ since otherwise one can set a cap of $u$ on the weights without changing the solution. Therefore, the maximum weight and the minimum weight of heavy battlefields are within a multiplicative factor of $1/\epsilon$. Next, by incurring an additional $1+\epsilon$ multiplicative factor to the approximation guarantee, we round down the weight of each battlefield to the nearest $(1+\epsilon)^i$. We show that this leaves us with a constant number of different weights for the heavy battlefields. Next, we state that since the number of different weights for the heavy battlefields is constant, the total number of (reasonable) pure strategies of player 1 over these battlefields can be reduced down to a polynomial. Indeed this also holds for player 2, but for the sake of our solution, we should further bound the number of responses of player 2 over these battlefields. We show in fact, that the number of (reasonable) responses of player 2 over the heavy battlefields against a strategy of player 1 is bounded by a constant! For light battlefields, we use the idea of a weak adversary. However, in order to find a solution that considers both heavy and light battlefields, we need to define multiple weak adversaries each with regard to a response of player 2 on the heavy battlefields. 

Let us clarify the challenge of mixing the two solutions via an example. Suppose we have 4 battlefields with weights $10, 8, 7$ and $5$ and the players' troops are as follows: $n = 5$ and $m = 2$. One can verify that in this case, the following pure strategy of player 1 guarantees a payoff of $15$ for him.

\begin{table}[h]
	\centering
	\begin{tabular}{|c|c|c|c|c|}
		\hline
		& \thead{Battlefield 1}{$w_1 = 10$} & \thead{Battlefield 2}{$w_2 = 8$} & \thead{Battlefield 3}{$w_3=7$} & \thead{Battlefield 4}{$w_4=5$} \\ \hline
		A \maxmin{15}{1} \cmixedstrategy{1} for player 1.  & 2                        & 2                        & 1                      & 0                       \\ \hline
	\end{tabular}
	\caption{A \maxmin{15}{1} \cmixedstrategy{1} for player 1 on instance \cbinstance{5}{2}{(10, 8, 7, 5)}.}
	\label{table:pure-example}
\end{table}
 
 In fact, $15$ is the highest utility player 1 can get with a single pure strategy as no other pure strategy of player 1 can guarantee a payoff more than $15$ for him. Now, assume that we select the first two battlefields with weights $10$ and $8$ as heavy battlefields and the rest of the battlefields as the light ones. One may think that by taking a maximin approach for the heavy battlefields and solving the problem separately for the light battlefields, we can obtain a correct solution. The above example shows that this is not the case. We show in what follows, that the maximin approach reports a payoff of $17$ for player 1 which is more than the actual solution. Fix the strategy of player 1 on the heavy battlefields to be placing 2 troops on battlefield 1 and 1 troop on battlefield 2.  As such, the only reasonable responses of player 2 on these battlefields are as shown in the following table.

\begin{table}[h]
	\centering
	\begin{tabular}{|c|c|c|c|}
		\hline
		& \thead{Battlefield 1}{$w_1 = 10$} & \thead{Battlefield 2}{$w_2 = 8$} & troops left for player 2 \\ \hline
		response 1 & 2                        & 0           & 0                                 \\ \hline
		response 2 & 0                        & 1            & 1                                \\ \hline
		response 3 & 0                        & 0             & 2                               \\ \hline
	\end{tabular}
	\caption{The responses of player 2 on the heavy battlefields.}
	\label{table:pure-example}
\end{table}

\noindent Response 3 already gives player 1 a payoff of $18$ which is more than $17$. Also, response 1 of player 2 leaves him with no troops for the light battlefields and thus he loses both light battlefields against strategy $(1,1)$ of player 1 on the light battlefields. Therefore, this gives player 1 a payoff of $20$. Also, if player 2 plays response 2 on the heavy battlefields, player 1 can win the light battlefield $w_3$ by putting two troops on it. Indeed player 2 has only one troop left and there is no way for him to win this battlefield. Thus, in this case, the payoff of player 1 would be $17$. Since we take the maximum solution over all strategies of player 1 for the heavy battlefields, our final utility would be at least 17.

What the above analysis shows is that, if we take a maximin approach on the heavy battlefields and then solve (or approximate) the problem for the light battlefields, we may incorrectly  report a higher payoff for player 1. Roughly speaking, this error happens since in this approach, we allow player 1 to have different actions over the light battlefields against different responses of player 2. To resolve this issue, we design a dynamic program that takes into account all responses of player 2 simultaneously. Indeed, to make sure the program can be solved in polynomial time, we need to narrow down the number of responses of player 2 to a constant. Our heavy-light decomposition along with our structural properties of the optimal strategies enables us to reach this goal. This gives us a non-trivial dynamic program that can approximate the solution within a factor $1+\epsilon$ for the case of a single pure strategy.

\resultsrestatethm{\ref{thm:discretepure}}{restated informally}{The problem of finding a \maxmin{u}{1} strategy for player 1 in discrete Colonel Blotto admits a PTAS.}

To extend the result to the case of $c \geq 2$, we need to design a weak adversary that plays a relaxed best response against more than 1 strategy of player 1. For $c=1$, the greedy algorithm follows from the well-known greedy solution of knapsack. However, when $c \geq 2$ the best-response problem does not necessarily reduce to knapsack and therefore our greedy solution is much more intricate. Roughly speaking, we design a non-trivial procedure for player 2 that gets $c$ thresholds as input, and based on these thresholds, decides about the response for each battlefield locally. This local decision making is a key property that we later exploit in our dynamic program to find an optimal strategy against a weak adversary. This in addition to the heavy-light decomposition technique gives us a PTAS for \maxmin{u}{p} \cmixedstrategy{c} strategies of Colonel Blotto in the discrete setting.

\resultsrestatethm{\ref{theorem:discretemulti}}{restated informally}{The problem of finding a \maxmin{u}{p} \cmixedstrategy{c} for player 1 in discrete Colonel Blotto admits a PTAS.}

\paragraph{Further results.} We show in Section~\ref{sec:exp}, that our techniques also imply PTASs for the Colonel Blotto game when instead of a \maxmin{u}{p} \cmixedstrategy{c} we wish to find a maximin \cmixedstrategy{c}. For the continuous case, similar to the case of \maxmin{u}{p} strategies, we divide the solution space into polynomially many convex subregions and prove that among them a $(1+\epsilon)$-approximate solution is guaranteed to exist. The main difference with the case of \maxmin{u}{p} strategies is in the LP formulation of the problem, but the general approach is essentially the same. For the discrete variant of Colonel Blotto, we also follow a similar approach as in the case of \maxmin{u}{p} strategies. In more details, we partition the battlefields into heavy and light subsets and define a weaker adversary that is adapted to approximately best respond against maximin strategies. We find it surprising and possibly of independent interest that essentially the same approach (though with minor changes) can be applied to these two variants of Colonel Blotto. Prior algorithms proposed for these two variants were fundamentally different \cite{ahmadinejad2016duels, behnezhad2017faster, behnezhad2018battlefields}.

\resultsrestatethms{\ref{thm:expcont} and \ref{thm:expdisc}}{restated informally}{The problem of finding a maximin \cmixedstrategy{c} for player 1 in both discrete and continuous variants of Colonel Blotto admits a PTAS.}

Finally, recall that motivation for approximate algorithms comes from intractability.  In view of all the recent sophisticated algorithmic approaches to solving, approximately, various special cases of the Colonel Blotto game, it is worth asking, what is the computational complexity of the full fledged problem of computing a maximin strategy of the Colonel Blotto game? (Notice that, since the full strategy is too long to return, we should formulate the problem in terms of something succinct, for example {\em one} component of the maximin.)  In Section~\ref{sec:complexity} we present the first complexity results in this area, establishing that an interesting variant of the problem that we call {\sc General Colonel Blotto} --- roughly, the utility is a general function of the two allocations, instead of the probability of winning more than a certain goal of total battlefield weight --- is {\em complete for exponential time.}  %The special case in which the function is the sum, over all battlefields, of a function of the troop allocation of each player, is also complete for exponential time, albeit {\em under polynomial space reductions} (we conjecture that this restriction can be removed).  
The precise complexity of the two versions of the original game (computing a maximin of the probability of winning a majority, or of the expectation of the total weights of battlefields won) is left as an open question here.  We conjecture that both problems are also exponential time-complete. %--- which would be interesting, as we are not aware of exponential time-complete optimization problems which have a PTAS.

\resultsrestatethmnorestated{\ref{thm:complexity}}{{{\sc General Colonel Blotto} is exponential time-complete.}
}

\subsection{Paper Outline}
In Section~\ref{sec:probdist}, we start by discussing possible probability distributions over the pure strategies in an optimal \maxmin{u}{p} \cmixedstrategy{c}. Then, focusing on \maxmin{u}{p} strategies, we show in Section~\ref{sec:continuous} that the continuous Colonel Blotto game admits a PTAS. In Section~\ref{sec:integral}, we show that the discrete variant of the game also admits a PTAS. We further show how it is possible to adapt these results to the case of maximin strategies in Section~\ref{sec:exp}. Finally, in Section~\ref{sec:complexity}, we describe our complexity results.

\section{Probability Distribution Over the Support}\label{sec:probdist}
In this section we discuss the structure of probabilities that are assigned to the pure strategies in the support of an optimal \maxmin{u}{p} strategy. It is worth mentioning that our results in this section apply to any two player game so long as the goal is to compute a \maxmin{u}{p} strategy.

Recall that given a minimum utility $u$, and a given upper bound $c$ on the support size, our goal is to compute a \maxmin{u}{p} \cmixedstrategy{c} for maximum possible $p$; we call this an optimal \maxmin{u}{p} \cmixedstrategy{c} or simply an optimal strategy when it is clear from context. Naively, there are uncountably many ways to assign probabilities to the $c$ pure strategies in the support. However, in this section we show that there are a small number of (efficiently constructible) possibilities for the probabilities among which an optimal solution is guaranteed to exist.

\begin{definition}
	For a mixed strategy $\b{x}$, define the profile \profile{\b{x}} of $\b{x}$ to be the multiset of probabilities associated to the pure strategies in the support of $\b{x}$. 
\end{definition}

The main theorem of this section is the following.

\newcommand{\thmcstrategyprobabilities}[0]{
	For every constant $c > 0$, there exists an algorithm to construct a set $P_c$ of $O(1)$ profiles in time $O(1)$, such that existence of an optimal \maxmin{u}{p} \cmixedstrategy{c} \b{x} with $\profile{\b{x}} \in P_c$ is guaranteed.
}
\begin{theorem}\label{thm:cstrategyprobabilities}
	\thmcstrategyprobabilities{}
\end{theorem}

As a corollary of Theorem~\ref{thm:cstrategyprobabilities}, in order to find an optimal \maxmin{u}{p} \cmixedstrategy{c}, one can iterate over the profiles in $P_c$, optimize the  strategy based on each profile, and report the best solution.

Let us start with the case when $c=1$. Clearly, in this case, we have no choice but to play our pure strategy with probability 1, therefore it suffices to have profile $\{1\}$ in set $P_1$, i.e., $P_1 = \{\{1\} \}$. For the case of $c=2$, we play two pure strategies with probabilities $q$ and $q' = 1-q$. Our claim for this case, is that it suffices to have $P_2 = \{ \{1\}, \{\sfrac{1}{2}, \sfrac{1}{2} \} \}$.\footnote{Note that the elements in $P_c$ are multisets, hence multiple occurrences of the same element is allowed in them.} Note that if the profile of a \cmixedstrategy{2} \b{x} is $\{1\}$ it is also a \cmixedstrategy{1} since it plays only 1 pure strategy with non-zero probability.

To prove that $P_2 = \{ \{1\}, \{\sfrac{1}{2}, \sfrac{1}{2} \} \}$ is sufficient, we first show that if a \cmixedstrategy{2} is \maxmin{u}{p} for some $p > \sfrac{1}{2}$, then there also exists a \maxmin{u}{1} \cmixedstrategy{1}, thus, $\{1\} \in P_2$ suffices.

\begin{observation}\label{obs:pmorethanhalf}
	If player 1 has a \maxmin{u}{p} \cmixedstrategy{2} with $p > \sfrac{1}{2}$, then he also has a \maxmin{u}{1} \cmixedstrategy{1}.
\end{observation}
\begin{proof}
	Let $\b{x}$ and $\b{x}'$ be the two strategies in the support of a given \maxmin{u}{p} \cmixedstrategy{2} for some $p > \sfrac{1}{2}$ and assume w.l.o.g., that strategy $\b{x}$ is played with probability at least $\sfrac{1}{2}$ (this should be true for at least one of the strategies in the support of any \cmixedstrategy{2}). Since $p > \sfrac{1}{2}$, by definition of a \maxmin{u}{p} strategy, $\b{x}$ should obtain a payoff of at least $u$ against any strategy of player 2. This means that pure strategy $\b{x}$ itself, is a $\maxmin{u}{1}$ \cmixedstrategy{1}. 
\end{proof}

On the other hand, the following observation shows that for an optimal \maxmin{u}{p} \cmixedstrategy{2}, we have $p \geq \sfrac{1}{2}$.

%if a \cmixedstrategy{2} is \maxmin{u}{p} for some $0 < p < \sfrac{1}{2}$, we can simply play the two strategies in its support with equal probability $\sfrac{1}{2}$ to obtain a better \maxmin{u}{\sfrac{1}{2}} strategy.

\begin{observation}\label{obs:plessthanhalf}
	If player 1 has a \maxmin{u}{p} \cmixedstrategy{2} \b{x} for some $0 < p < \sfrac{1}{2}$, then by playing each of the two pure strategies of \b{x} with probability $\sfrac{1}{2}$, he gets a better \maxmin{u}{\sfrac{1}{2}} \cmixedstrategy{2}. 
\end{observation}
\begin{proof}
	Let $\b{x}^1$ and $\b{x}^2$ be the two pure strategies in the support of \b{x}. Since $p > 0$, against any pure strategy $\b{y}$ of the opponent, at least one of $\b{x}^1$ or $\b{x}^2$ obtain a utility of at least $u$. Thus, if they are both played with probability $\sfrac{1}{2}$, the resulting strategy is \maxmin{u}{\sfrac{1}{2}}.
\end{proof}

Combining the two observations above, we can conclude that for an optimal \maxmin{u}{p} \cmixedstrategy{2}, either $p=1$ or $p = \sfrac{1}{2}$. For the former case Observation~\ref{obs:pmorethanhalf} implies that $\{1\} \in P_2$ is sufficient and for the latter $\{\sfrac{1}{2}, \sfrac{1}{2}\} \in P_2$ is sufficient by Observation~\ref{obs:plessthanhalf}. We remark that it is necessary for any choice of $P_2$ to have both $\{1\}$ and $\{\sfrac{1}{2}, \sfrac{1}{2}\}$. Therefore our choice of $P_2 = \{ \{1 \}, \{\sfrac{1}{2}, \sfrac{1}{2}\} \}$ is both sufficient and necessary.

We do not attempt to prove it here, as it is not required for the proof of our main theorem in this section, however, it can also be shown that for the case of \cmixedstrategies{3}, it suffices to have $$P_3 = \{ \{1\}, \{\sfrac{1}{2}, \sfrac{1}{2}\}, \{\sfrac{1}{3}, \sfrac{1}{3}, \sfrac{1}{3}\}\}.$$ Unfortunately, the case of $c \geq 4$ does not follow the same pattern. That is,  e.g. when $c=4$, it is not sufficient to have $P_4 = \big\{ \{1\}, \{\sfrac{1}{2}, \sfrac{1}{2}\}, \{\sfrac{1}{3}, \sfrac{1}{3}, \sfrac{1}{3}\}, \{\sfrac{1}{4}, \sfrac{1}{4}, \sfrac{1}{4}, \sfrac{1}{4}\}\big\}$. An example was given for this in Table~\ref{table:nonuinformprob} in the context of the Colonel Blotto game where the optimal \cmixedstrategy{4} has profile $\{ \sfrac{2}{5}, \sfrac{1}{5}, \sfrac{1}{5}, \sfrac{1}{5} \}$. Nevertheless, we show that size of $P_c$, for any constant $c$, can be bounded by a constant.

Consider a \maxmin{u}{p} \cmixedstrategy{c} \b{x} for player 1 in a two player game $\mathcal{G}$ and assume w.l.o.g. that the strategies in its support are numbered from 1 to $c$. A subset $W \subseteq [c]$ is a {\em $u$-winning-subset} of $\b{x}$, if player 2 has a strategy $\b{y}$ such that the $i$th strategy in the support of $\b{x}$ gets a payoff of at least $u$ against $\b{y}$ if and only if $i \in W$. We denote by \winningsubset{u}{\b{x}} the set of all $u$-winning-subsets of $\b{x}$. The following claim implies that it suffices to have the set \winningsubset{u}{\b{x}} of any  strategy \b{x} to decide how to assign probabilities to the strategies in its support.

\begin{lemma}\label{lem:constructprofile}
	Given the set \winningsubset{u}{\b{x}} of a \maxmin{u}{p} \cmixedstrategy{c} \b{x}, one can, in time $O(1)$, assign probabilities $\rho_1, \ldots, \rho_c$ to the pure strategies in the support of \b{x} in a way that the modified strategy $\b{x}'$ is \maxmin{u}{p'} for some $p' \geq p$.
\end{lemma}
\begin{proof}
Let $\mathcal{W}$ be the given set. Clearly we have $|\mathcal{W}| \leq O(1)$ since it is a subset of the power set of $[c]$ and $c$ is a constant. The following linear program finds $\rho_1, \ldots, \rho_c$ in time $O(1)$ since it has $O(1)$ variables and $O(1)$ constraints. We show that by playing the $i$th pure strategy in the support of \b{x} with probability $\rho_i$, we obtain a strategy $\b{x}'$ that is as good as \b{x}.%We claim that if $\mathcal{W}$ is the set of all $u$-winning-subsets of $\b{x}$, then $\{\rho_1, \ldots, \rho_c\}$ is the profile of either $\b{x}$.
\begin{equation}\label{lp:prob}
\begin{array}{ll@{}ll}
\text{maximize}  & p' &  &\\
\text{subject to}& \rho_i \geq 0  \qquad & & \forall i\in[c] \\
&                  \sum \rho_{i} = 1         & &\\
&		p' \leq \sum_{i \in W} \rho_i & & \forall W \in \mathcal{W}
\end{array}\end{equation}
Let us denote by $q_i$ the actual probability with which the $i$th strategy in the support of $\b{x}$ is played. Observe that the probability $p$ for which strategy \b{x} guarantees receiving a payoff of at least $u$ is exactly $\min_{W \in \winningsubset{u}{\b{x}}} \sum_{i\in W} q_i$. This comes from the fact that by definition if $W$ is a  $u$-winning-subset of $\b{x}$, then player 2 has a strategy \b{y} against which player 1 receives a payoff of $u$ by its $i$th pure strategy iff $i \in W$. This means the probability of guaranteeing a utility of $u$ against $\b{y}$ is equal to $\sum_{i\in W}q_i$. On the other hand, by definition of \maxmin{u}{p}, we need to guarantee a payoff of at least $u$ with probability at least $p$ against {\em any} strategy of the opponent. Therefore we have $p = \min_{W \in \winningsubset{u}{\b{x}}} \sum_{i\in W} q_i$. Now, observe that LP~\ref{lp:prob} finds probabilities $\rho_i$ in a way that precisely maximizes $\min_{W \in \winningsubset{u}{\b{x}}} \sum_{i\in W} \rho_i$. Therefore, if we play the $i$th pure strategy in the support of $\b{x}$ with probability $\rho_i$, we obtain a strategy that is as good as \b{x}.
\end{proof}

\begin{claim}\label{claim:sexists}
	One can obtain a set $S$ of size $O(1)$ in time $O(1)$ such that for some optimal \maxmin{u}{p} \cmixedstrategy{c} {\b{x}}, we have $\winningsubset{u}{\b{x}} \in S$.
\end{claim}
\begin{proof}
	Let $S$ be the power set of the power set of $C = \{1, \ldots, c\}$. Note that $|S| = 2^{2^c} = O(1)$; thus, we can simply construct $S$ in time $O(1)$. On the other hand, take an optimal \maxmin{u}{p} \cmixedstrategy{c} \b{x}. By definition each $u$-winning-subset of $\b{x}$ is a subset of $C$ and, thus, $\winningsubset{u}{\b{x}}$ is a set of some subsets of $C$, thus, $\winningsubset{u}{\b{x}} \in S$ as desired.
\end{proof}

Indeed the two claims above are sufficient to prove Theorem~\ref{thm:cstrategyprobabilities}.

\restatethm{\ref{thm:cstrategyprobabilities}}{\thmcstrategyprobabilities{}}
\begin{proof}
	Consider the set $S$ provided by Claim~\ref{claim:sexists}. For every set $\mathcal{W} \in S$, construct a profile $\rho$ using the algorithm of Lemma~\ref{lem:constructprofile} and put it in a set $P$. Observe that $|P| = |S| = O(1)$. Therefore it only remains to prove that the profile of at least one optimal strategy is in $P$.
	
	Note that by Claim~\ref{claim:sexists}, at least one set $\mathcal{W} \in S$ is the set of all $u$-winning-subsets of an optimal strategy. By Lemma~\ref{lem:constructprofile}, the constructed profile for $\mathcal{W}$ is the profile of a strategy that is also optimal. Therefore among the profiles in $P$, we have the profile of at least one optimal strategy which concludes the proof.
\end{proof}

\section{Continuous Colonel Blotto}\label{sec:continuous}

In this section we consider the continuous variant of Colonel Blotto. We start with the case where our goal is to find a \maxmin{u}{p} \cmixedstrategy{1} and show how we can generalize it to \cmixedstrategies{2} and, further, to any bounded number of strategies.

%\subsection{The Case of 1-Strategies}

%We consider \cmixedstrategies{1} in this section. 
For the particular case of $c=1$, our goal is to find a single pure strategy that is \maxmin{u}{p}. Indeed, since we are playing a single pure strategy with no randomization, the probability $p$ must be 1. One can think of the solution of this case as a vector of non-negative real values that sum up to $n$ and formulate the problem in the following way.
\begin{equation}\label{program:one}
\begin{array}{ll@{}ll}
	\text{find}  & \b{x}\\
	\text{subject to}& x_i \geq 0   & &\forall i \in [k] \\
	&                  \sum x_{i}  \leq n         & &\\
	& \utilitya{\b{x}}{\b{y}} \geq u & & \forall \b{y} \in \puresetb{}
\end{array}\end{equation}
To analyze Program~\ref{program:one}, we need to better understand the payoff constraints. To this end, we state another interpretation of the payoff constraints in Observation~\ref{observation:two}.

\begin{observation}[proven in ~\cite{behnezhad2018battlefields}]\label{observation:two}
	A pure strategy $\b{x}$ of player 1 guarantees a payoff of at least $u$ against any pure strategy of player 2 if and only if
	$\sum_{i \in S} x_i > m$ for any set $S$ of battlefields with $\sum_{i \notin S} w_i < u$.
\end{observation} 

\begin{proof}
	($\Rightarrow$): Suppose that a strategy $\b{x}$ does not get a payoff of $u$ against a strategy $\b{y}$ of player 2. This means that there exists a set $S$ of battlefields that player $\b{x}$ loses. The total payoff of $\b{x}$ is below $u$, and therefore $\sum_{i \notin S} w_i < u$. In addition to this $\sum_{i \in S} x_i \leq m$ holds since player 2 needs to match player 1's troops in all battlefields of $S$.
	
	($\Leftarrow$): It is trivial to show that if there exists such a violating set $S$, then $\b{x}$ cannot promise a payoff of $u$ against any strategy of player 2. The reason is that since $\sum_{i \in S} x_i \leq m$ player 2 can match the troops of player 1 in set $S$ and that suffices to prevent player 1 from winning a payoff of $u$.
\end{proof}

Via Observation~\ref{observation:two}, we can turn Program~\ref{program:one} into LP~\ref{lp:one} where we have a constraint for every possible subset $S$ of battlefields with $\sum_{i \notin S} w_i < u$. Although the number of these subsets can be exponentially large, Behnezhad \etal~\cite{behnezhad2018battlefields} show that one can find a violating constraint of LP~\ref{lp:one} in polynomial time and thus find a feasible solution using the ellipsoid method. 

\begin{equation}\label{lp:one}
\begin{array}{ll@{}ll}
\text{find}  & \b{x} &  &\\
\text{subject to}& x_i \geq 0   & &\forall i \in [k] \\
&                  \sum x_{i}  \leq n         & &\\
&                  \sum_{i \in S} x_{i} > m  & &\text{ for every subset $S$ of battlefields with }\sum_{i \notin S} w_i < u.
\end{array}\end{equation}

The key idea that enables us to solve this variant is that we are playing only one pure strategy. Even the case of having two pure strategies in the support is much more challenging. To illustrate the challenges and ideas, we next focus on how to obtain a \cmixedstrategy{2} and later generalize it to \cmixedstrategies{c}.

%One direction to generalize this idea is to design an algorithm that finds \maxmin{u}{p} strategies when the size of the strategy profile is bounded by a constant number $c$. 
%\begin{center}
%	\noindent\framebox{\begin{minipage}{6.50in}
%			Given a Blotto game which is guaranteed to admit a \maxmin{u}{p} strategy for player 1 by randomizing over at most $c$ pure strategies, compute (or approximate) such a strategy in polynomial time.
%	\end{minipage}}
%\end{center}

\subsection{The Case of 2-Strategies}\label{sec:contc2}

Recall by Observation~\ref{obs:pmorethanhalf} of Section~\ref{sec:probdist} that if a \cmixedstrategy{2} is \maxmin{u}{p} for some $p > \sfrac{1}{2}$, then there also exists a pure \maxmin{u}{1} strategy that, by aforementioned techniques, can be found in polynomial time. It was further shown in Observation~\ref{obs:plessthanhalf} that if a \cmixedstrategy{2} is \maxmin{u}{p} for some $p < \sfrac{1}{2}$, we can simply play the two strategies in its support with equal probability $\sfrac{1}{2}$ to obtain a better \maxmin{u}{\sfrac{1}{2}} strategy. Combining these two observations, we can assume w.l.o.g., that in the case of $c=2$, we wish to find two pure strategies $\b{x}$ and $\b{x}'$, and play them with equal probabilities such that any strategy of player 2 gives us a payoff of at least $u$ for at least one of $\b{x}$ or $\b{x}'$. A mathematical formulation of the problem is given below.

\begin{equation}\label{program:two}
\begin{array}{ll@{}ll}
\text{find}  & \b{x}  \text{ and } \b{x}' &  &\\
\text{subject to}& x_i \geq 0 \text{ and } x'_i \geq 0  & &\forall i\in[k] \\
&                  \sum x_{i} \leq n         & &\\
&                  \sum x'_{i} \leq n         & &\\
&		\text{either } \utilitya{\b{x}}{\b{y}} \geq u \text{ or } \utilitya{\b{x'}}{\b{y}} \geq u & & \forall \b{y} \in \puresetb{}
\end{array}\end{equation}

Observe that the fourth constraint of the above program is not linear. In the following, we show that even the polytope that is described by this program is essentially nonconvex.

\begin{observation}\label{obs:nonconvex}
   Program~\ref{program:two} is not convex.
\end{observation}
\begin{proof}
   Suppose $n = m = 2$, $w_1 = w_2 = 1$ and the goal is to find a \maxmin{1}{\sfrac{1}{2}} strategy by randomizing over two pure strategies. A possible solution is to play $\b{x} = (2,0)$ with probability $1/2$ and play $\b{x}' = (0, 2)$ with probability $1/2$ which guarantees a payoff of $1$ with probability $1/2$. An alternative way to achieve this goal is to set $\b{x} = (2, 0)$ and $\b{x}' = (0, 2)$ which is the same strategy except that $\b{x}$ and $\b{x}'$ are exchanged. However, the linear combination of the two strategies results in $\b{x} = (1,1)$ and $\b{x}' = (1,1)$ which always loses both battlefields against $\b{y} = (1, 1)$.
   \begin{figure}[h]
   	\centering
   	\includegraphics[scale=0.7]{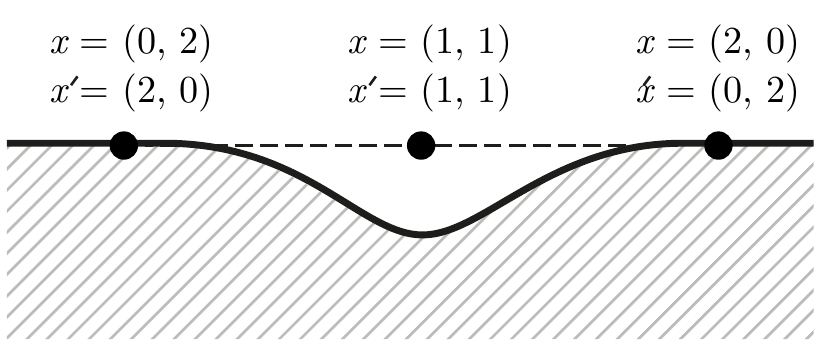}
   	\caption{The linear combination of two feasible strategies in the example of Observation~\ref{obs:nonconvex} is not feasible, therefore, the feasibility region, which is  represented by the hatched area is not convex.}
   	\label{fig:convex}
   \end{figure}
   This is illustrated in Figure~\ref{fig:convex}.
\end{proof}

A more careful analysis shows that the feasible region of Program~\ref{program:two} may be the union of up to $2^k$ convex polytopes which makes it particularly difficult to find a desired solution. In what follows, we present algorithms to overcome this challenge for both the uniform and nonuniform settings.

\subsubsection{Uniform Setting}\label{section:uniformc=2}
Recall that in order to find a pure \maxmin{u}{1} strategy, we proved the linearity of Program~\ref{program:one} by characterizing the optimal solution. Similar to that, we show necessary and sufficient conditions for the solution of Program~\ref{program:two}. To this end, we define {\em critical tuples} as follows.

\begin{definition}
	Let $\A$, $\B$, and $\AB$ be three disjoint subsets of battlefields. We call the tuple $\langle \A, \B, \AB\rangle$ a critical tuple, if critical any of $\A \cup \AB$ or $\B \cup \AB$ suffices to prevent player 1 from getting a payoff of $u$. In other words, $\langle \A,\B,\AB \rangle$ is a critical tuple, if and only if $$\sum_{i \notin \A \cup \AB} w_i < u, \qquad \text{and,} \qquad \sum_{i \notin \B \cup \AB} w_i < u.$$
\end{definition}

Via this definition, we can now describe the feasible solutions of Program~\ref{program:two} as follows.

\begin{observation}\label{observation:three}
	Two pure strategies $\b{x}$ and $\b{x}'$ of player 1 meet the constraints of Program~\ref{program:two} if and only if for any critical tuple $\langle \A, \B, \AB\rangle$ we have
	$$\sum_{i \in \A} x_i + \sum_{i \in \B} x'_i + \sum_{i \in \AB} \max\{x_i, x'_i\} > m. $$
\end{observation}
\begin{proof}
	The proof is similar to that of Observation~\ref{observation:two}. Note that $\b{x}$ and $\b{x}'$ violate a payoff constraint of Program~\ref{program:two} if they both get a payoff less than $u$ against a pure strategy $\b{y}$ of player 2. In this case we define three sets $\A$, $\B$, and $\AB$ as 
	\begin{equation*}
		\A = \{i : x_i \leq y_i \text{ and } x'_i > y_i\}, \quad
		\B = \{i : x_i > y_i \text{ and } x'_i \leq y_i\}, \quad
		\AB = \{i : x_i \leq y_i \text{ and } x'_i \leq y_i\}.
	\end{equation*}
	Observe that $\A$, $\B$, and $\AB$ make a critical tuple since both $\b{x}$ and $\b{x}'$ get a payoff less than $u$ against $y$. Since $\A$, $\B$, and $\AB$ are disjoint and $\sum y_i = m$ we have 	$\sum_{i \in \A} x_i + \sum_{i \in \B} x'_i + \sum_{i \in \AB} \max\{x_i,x'_i\} \leq m $. A similar argument implies that if this condition does not hold for any critical tuple, then $\b{x}$ and $\b{x}'$ meet the conditions of Program~\ref{program:two}.
\end{proof}

Based on Observation~\ref{observation:three} we rewrite Program~\ref{program:two} in the following way.
\begin{equation}\label{program:three}
\begin{array}{ll@{}ll}
\text{find}  & \b{x}  \text{ and } \b{x}' &  &\\
\text{subject to}& x_i \geq 0 \text{ and } x'_i \geq 0  & &\forall i\in[k]\\
&                  \sum x_{i}  \leq n         & &\\
&                  \sum x'_{i}  \leq n         & &\\
&					z_i = \max\{x_i, x'_i\} & & \forall i \in [k]\\
& \sum_{i \in \A} x_i + \sum_{i \in \B} x'_i + \sum_{i \in \AB} z_i > m  & & \text{ for every critical tuple }\langle \A, \B, \AB\rangle
\end{array}\end{equation}
Indeed Program~\ref{program:three} is not convex since $z_i = \max\{x_i, x'_i\}$ is not a linear constraint. The naive approach to get around this issue is to consider $2^k$ possibilities for the assignment of $z_i$'s. More precisely, if we knew in an optimal strategy for which $i$'s we have $x_i > x'_i$ and for which $x_i \leq x'_i$, we could turn Program~\ref{program:three} into a linear program by replacing each $z_i$ with either $x_i$ or $x'_i$. This observation gives us an exponential time solution to find a \maxmin{u}{p} strategy by trying all $2^k$ combinations. However, for the uniform case, we can further improve the running time to a polynomial. The overall idea is that when we are indifferent between the battlefields, we do not necessarily need to know for which subset of battlefields $\b{x}$ puts more troops that $\b{x}'$. It suffices to be given the count!

In the uniform setting, let $\hat{\b{x}}$ and $\hat{\b{x}}'$ be the actual solution. Count the number of battlefields on which $\hat{\b{x}}$ puts more troops than $\hat{\b{x}}'$ and call this number $\alpha$. Therefore, on $k - \alpha$ battlefields, $\b{x}'$ puts at least as many troops as $\b{x}$. Since the battlefields are identical, we can rearrange the order of the battlefields to make sure $\b{x}$ puts more troops than $\b{x}'$ in the first $\alpha$ battlefields. If $\alpha$ is given to us, we can formulate the problem as follows.

\begin{equation}\label{program:four}
\begin{array}{ll@{}ll}
\text{find}  & \b{x} \text{ and } \b{x}' &  &\\
\text{subject to}& x_i \geq 0 \text{ and } x'_i \geq 0  & &\forall i \in [k]\\
&                  \sum x_{i}  \leq n         & &\\
&                  \sum x'_{i}  \leq n         & &\\
&					x_i \geq x'_i, z_i = x_i & & \forall i: 1 \leq i \leq \alpha\\
&					x_i \leq x'_i, z_i = x'_i & & \forall i: \alpha < i \leq k\\
& \sum_{i \in \A} x_i + \sum_{i \in \B} x'_i + \sum_{i \in \AB} z_i > m  & & \text{for every critical tuple }\A, \B, \AB
\end{array}\end{equation}

% \begin{figure}
%	\centering
%	\includegraphics[scale=0.7]{figs/decompose}
%	\caption{The feasible region of Program~\ref{program:three} may not be convex. In each convex polytope and for each battlefield $i$,  $x_i \tiny{\overset{?}{\le}} x'_i$ is the same for all solution points.}
%\end{figure}

Program~\ref{program:four} is clearly an LP. We show in Theorem~\ref{theorem:one} that LP~\ref{program:four} can indeed be solved in polynomial time. This settles the problem when $\alpha$ is given.  Note that $\alpha$ takes an integer value between $0$ and $k$ and thus we can iterate over all possibilities and solve the problem in polytime.

\begin{theorem}\label{theorem:one}	Given that an instance of continuous Colonel Blotto in the uniform setting admits a \maxmin{u}{\sfrac{1}{2}} \cmixedstrategy{2} for player 1, there exists an algorithm to find one such solution in polynomial time.
\end{theorem}

\begin{proof}
As discussed earlier, the solution boils down to solving LP~\ref{program:four}. Here we show that LP~\ref{program:four} can be solved in polynomial time using the ellipsoid method. For that, we need a separating oracle that for any given assignment to the variables decides in polynomial time whether any constraint is violated and if so, reports one. LP~\ref{program:four} has polynomially many constraints, except the constraints of form
\begin{equation*}
	\sum_{i \in \A} x_i + \sum_{i \in \B} x'_i + \sum_{i \in \AB} z_i > m  \qquad \text{for every critical tuple }\A, \B, \AB,
\end{equation*}
since there may be exponentially many critical tuples. Therefore the only challenge is whether any constraint of this form is violated. That is for given strategies $\b{x}$ and $\b{x}'$ and with $z_i = \max \{x_i, x'_i\}$, we need to design an algorithm that finds a critical tuple $\langle \A, \B, \AB\rangle$ (if any) for which we have
\begin{equation}
	\sum_{i \in \A} x_i + \sum_{i \in \B} x'_i + \sum_{i \in \AB} \max\{x_i, x'_i\} \leq m.
\end{equation}
For this, note that as implied by Observation~\ref{observation:three}, it suffices to be able to find a pure strategy $\b{y}$ of player 2 such that \begin{equation}\label{eq:bestresponse}
		\utilitya{\b{x}}{\b{y}} < u, \qquad \text{and,} \qquad \utilitya{\b{x}'}{\b{y}} < u.
	\end{equation}
This is in some sense, equivalent to player 2's best-response which we show can be solved in polynomial time via a dynamic program. Define $D(j, m', \upsilon, \upsilon')$ to be 1 if and only if player 2 can use up to $m'$ troops in the first $j$ battlefields in a way that prevents $\b{x}$ (resp. $\b{x}'$) from obtaining a payoff of at least $\upsilon$ (resp. $\upsilon'$) in those battlefields. More precisely, $D(j, m', \upsilon, \upsilon')$ is 1 if and only if there exists a strategy $\b{y}$ for player 2 such that
\begin{equation*}
	\sum_{i=1}^{j} y_i \leq m', \qquad \sum_{i: i \in [j], x_i > y_i} w_i \leq \upsilon, \qquad \text{and,} \qquad \sum_{i: i \in [j], x'_i > y_i} w_i \leq \upsilon'.
\end{equation*}
Clearly, if we are able to solve $D(j, m', \upsilon, \upsilon')$ for all possible inputs, then it suffices to check the value of $D(k, m, u, u)$ to see whether we can find a strategy $\b{y}$ satisfying (\ref{eq:bestresponse}). Indeed, we can update the dynamic program in the following way:
\begin{equation*}
	D(j, m', \upsilon, \upsilon') = \max_{y_j \in \{0, \ldots, m'\}} D\Big(j-1, m'-y_j, \upsilon - g(y_j, j), \upsilon' - g'(y_j, j)\Big),
\end{equation*}
where,
\begin{equation*}
	g(y_j, j) = \begin{cases}
	w_j , & \text{if } x_j > y_j, \\
	0 & \text{otherwise},
	\end{cases}
	\qquad\qquad \text{and,} \qquad\qquad
	g'(y_j, j) = \begin{cases}
	w_j , & \text{if } x'_j > y_j, \\
	0 & \text{otherwise.}
	\end{cases}
\end{equation*}
As for the base case, we set $D(0, 0, 0, 0)=1$. The correctness of the dynamic program is easy to confirm, since we basically check all possibilities for the number of troops that the second player can put on the $j$th battlefield and update the requirements on the previous battlefields accordingly. A minor issue, here, is that this only confirms whether a strategy $\b{y}$ exists that satisfies (\ref{eq:bestresponse}) or not and does not give the actual strategy. However, one can simply obtain the actual strategy by slightly modifying the dynamic program to also store the strategy itself. 

To summarize, we gave a polynomial time separating oracle for LP~\ref{program:four} that gives a polynomial time algorithm to solve it which leads to a \maxmin{u}{\sfrac{1}{2}} \cmixedstrategy{2} for player 1 in polynomial.
\end{proof}

\subsubsection{General Weights}\label{section:nonuniformc=2}
Theorem~\ref{theorem:one} shows that the problem of computing a \maxmin{u}{\sfrac{1}{2}} \cmixedstrategy{2} is computationally tractable when the weights are uniform. In this section, we study the general (i.e., non-uniform) setting and show that it is possible to obtain an (almost) optimal solution for this problem in polynomial time.

Recall that we assume there exists a \maxmin{u}{\sfrac{1}{2}} \cmixedstrategy{2} and the goal is to either compute or approximate such a strategy. Fix the pure strategies of the solution to be $\hat{\b{x}}$ and $\hat{\b{x}}'$. Similar to Section~\ref{section:uniformc=2}, if we knew for which $i$'s $\hat{x}_i \geq \hat{x}'_i$ holds and for which $i$'s it is the opposite, we could model the problem as a linear program and find a solution in polynomial time. Since $w_i$'s are not necessarily the same, unlike Section~\ref{section:uniformc=2}, we need to try an exponential number of combinations to make a correct decision. To alleviate this problem, we show a generalized variant of the above argument. Define the status of battlefield $i$ to be compatible with either $\leq$ or $\geq$ (or both in case of equality) based on the comparison of $\hat{x}_i$ and $\hat{x}'_i$. Assume we make a guess for the status of battlefields, which is incorrect for a set $S$ of battlefields but correct for the rest of them. This means that for every battlefield $i$ in set $S$, if $\hat{{x}}_i > \hat{{x}}'_i$, we assume ${x}_i \leq {x}'_i$ and vice versa. We show that if the total weight of the battlefields in $S$ is small, there exists an almost optimal solution for the problem whose status is compatible with our guess.

For simplicity, we represent a guess for the status of the battlefields with a vector $\b{g} \in \{\leq,\geq\}^k$ of length $k$ in which every entry is either `$\leq$' or
`$\geq$'. A solution $(\b{x},\b{x}')$ is compatible with this guess if $g_i$ correctly compares $x_i$ to $x'_i$.
\begin{lemma}\label{lemma:approx}
Let $(\hat{\b{x}},\hat{\b{x}}')$ be an optimal \maxmin{u}{p} solution of the problem and $\b{g}$ be a guess for the status of the battlefields. Let $S$ be the set of battlefields for which $\b{g}$ makes an incorrect comparison between $\hat{\b{x}}$ and $\hat{\b{x}}'$ on these battlefields. If $\sum_{i \in S} w_i = \alpha$ then there exists a \maxmin{u-\alpha}{p} strategy that is compatible with $\b{g}$.
\end{lemma}
\begin{proof}
  We construct a pair of strategies $(\b{x},\b{x}')$ based on $\hat{\b{x}}$ and $\hat{\b{x}}'$. For every battlefield $i \in S$, we set $x_i = x'_i = 0$ and for every battlefield $i \notin S$ we set $x_i = \hat{x}_i$ and $x'_i = \hat{x}'_i$. If a strategy of player 2 prevents both $\b{x}$ and $\b{x}'$ from getting a utility of $u-\alpha$, then the same strategy prevents both $\hat{\b{x}}$ and $\hat{\b{x}}'$ from getting a payoff of $u$. Therefore, $(\b{x},\b{x}')$ is \maxmin{u-\alpha}{p}. Since for every battlefield outside set $S$ we have $x_i = x'_i = 0$, then both $\leq$ and $\geq$ correctly compare the corresponding values for such battlefields. Therefore, $(\b{x},\b{x}')$ is compatible with $\b{g}$.
\end{proof}

A simple interpretation of Lemma~\ref{lemma:approx} is that if we make a guess that differs from a correct guess in a subset of battlefields with a total weight of $\alpha$, we can use this guess to find a solution with an additive error of at most $\alpha$ in the utility. Based on this idea, we present a polynomial time algorithm that for any arbitrarily small constant $\epsilon < 1$ computes a \maxmin{(1-\epsilon)u}{\sfrac{1}{2}} \cmixedstrategy{2}.

\paragraph{$\delta$-Uniform weights.} One of the crucial steps of our algorithms, is updating battlefield weights. This step, is indeed used in multiple other places of the paper as well. For a parameter $\delta$, we define a {\em $\delta$-uniform} variant of the game to be an instance on which the weight of each of the battlefields is {\em rounded down} to be in set $\mathcal{W} = \{ 1, (1+\delta)^1, (1+\delta)^2, \ldots \}$. That is, for any $i \in [k]$, we set the updated weight of battlefield $i$ to be $w'_i = \max \{ w : w \in \mathcal{W}, w \leq w_i \}$. The following observation implies that we can safely assume  the game is played on the updated weights without losing a considerable payoff.

\begin{observation}\label{obs:bfweightrounding}
	For any $u'$, any \maxmin{u'}{p} strategy of the game instance \cbinstance{n}{m}{\b{w'}} with the updated weights is a \maxmin{(1-\delta)u'}{p} strategy of the original instance \cbinstance{n}{m}{\b{w}}. Similarly, any \maxmin{u'}{p} strategy of the original instance \cbinstance{n}{m}{\b{w}} is a \maxmin{(1-\delta)u'}{p} for instance \cbinstance{n}{m}{\b{w}'}.
\end{observation}
\begin{proof}
	Let $\b{x} = (x_1, \ldots, x_k)$ be any pure strategy in the support of the \maxmin{u'}{p} strategy of \cbinstance{n}{m}{\b{w'}}. Consider any arbitrary strategy $\b{y}$ of player 2, it suffices to show that $\b{x}$ gets a payoff of at least $(1- \delta)u'$ against \b{y} in the original instance \cbinstance{n}{m}{\b{w}}. Note that for any $i \in [k]$ we have $w'_i \geq w_i/(1+\delta) \geq (1-\delta) w_i$ by the way that we round down the weights; therefore, we have that
	\begin{equation*}
		\sum_{i: x_i > y_i} w'_i \geq \sum_{i: x_i > y_i} (1-\delta) w_i \geq (1-\delta) \sum_{i: x_i > y_i} w_i \geq (1-\delta) u',
	\end{equation*}
	completing the proof for the first part.
	
	Similarly, since the weight of each battlefield is multiplied by a factor of at most $\sfrac{1}{(1+\delta)}$, any $\maxmin{u'}{p}$ strategy for the original instance is a \maxmin{u'/(1+\delta)}{p} or simply a \maxmin{(1-\delta)u'}{p} strategy for instance \cbinstance{n}{m}{\b{w}'}.
\end{proof}

\paragraph{The algorithm in a nutshell.} 
We first update the weight of every battlefield $i$ to be $\min\{u, w_i\}$. This, in fact, does not change the game instance for player 1 since his only objective is to guarantee a payoff of at least $u$. Now, for $\delta = \epsilon^3/10$ which is a relatively smaller error threshold than $\epsilon$, we consider the $\delta$-uniform variant of the game. In the $\delta$-uniform instance, since the weights change exponentially in $1+\delta$, we have at most $O(1/\delta \cdot \log u)$ distinct weights . We put the battlefields with the same weight into a bucket and denote the sizes of the buckets by $k_1, k_2, \ldots, k_b$ where $b$ is the number of buckets.
Recall from Section~\ref{section:uniformc=2} that if a set of battlefields have the same weight, then we are indifferent between these battlefields and thus the only information relevant to these battlefields is on how many of them $\hat{\b{x}}$ puts more troops than $\hat{\b{x}}'$. Therefore one way to make a correct guess is to try all $\prod (k_i +1)$ possibilities for all of the buckets. Unfortunately, $\prod (k_i +1)$ is not polynomial since the number of buckets is not constant. In order to reduce the number of possibilities to a polynomial, we make a number of observations.

First, since the weights decrease exponentially between the buckets, the number of distinct weights that are larger than $\delta u/k$ (and smaller than $u$ as described above,) is at most $\log_{1+\delta}k/\delta = O_\delta(\log k)$. Observe that we can safely ignore (i.e., make a wrong guess for) all the buckets with weight less than $\delta u /k$ since sum of the weights of all battlefields in such buckets is at most $\delta u$ and by Lemma~\ref{lemma:approx} it causes us to lose an additive error of at most $\delta u$. Although this reduces the number of buckets down to $O_\delta(\log k)$, it is still more than we can afford to try all $\prod (k_i +1)$ possibilities.
	
Second, instead of trying $k_j +1$ possibilities for bucket $j$, we reduce it down to only $O(1/\delta)$ options. More precisely, let $t_j$ be the number of battlefields $i$ in bucket $j$ such that $\hat{x}_i \geq \hat{x}'_i$. For any bucket $j$ with $k_j > 1/\delta$, if we only consider $t_j$ to be in set $\{0,\lfloor \delta k_j \rfloor, \lfloor 2\delta k_j \rfloor, \lfloor 3\delta k_j \rfloor, \ldots, k_j\}$ one of the realizations of $t_j$ makes at most $\delta k_j + 1$ incorrect guesses for bucket $j$. We use this later to argue that we do not lose a significant payoff by considering only $O(1/\delta)$ possibilities per bucket. As a result, we reduce the size of the cartesian product of all possibilities over all the buckets down to a polynomial.
	
Third, we show that if $n \geq (1+\epsilon) m/2$, we can safely assume that losing the value of at most $\lfloor \delta k_j\rfloor$ battlefields of buckets with more than $1/\delta$ battlefields does not hurt the payoff significantly. In other words, when $n \geq (1+\epsilon) m/2$, the optimal utility $u$ is much larger than the total sum of the payoff we lose by only trying $1/\delta$ possibilities for every bucket $j$ (proven in Lemma~\ref{lemma:notlosemuch}).

Based on the above ideas, we outline our PTAS as follows: (i) We first set a cap of $u$ for the weight of the battlefields. (ii) Let $\delta = \epsilon^3/10$. Next, we round down the weight of the battlefields to be powers of $(1+\delta)$.   (iii) We put the battlefields with the same weights in the same bucket and remove the buckets whose battlefield weights are smaller than $\delta u/k$. (iv) Finally, we try $O(1/\delta)$ possibilities for the status of the battlefields within each bucket and check its feasibility using LP~\ref{program:four}.\footnote{As a minor technical detail, since our goal is to guarantee a payoff of at least $(1-\epsilon)u$ instead of $u$, we need to update the definition of losing tuples accordingly for this case.} We return the first feasible solution that we find. The formal algorithm is given as Algorithm~\ref{alg:continuous2general}.

\begin{algorithm}
	\caption{Algorithm to find a $\maxmin{u}{p}$ $\cmixedstrategy{2}$ for general weights.}
	\label{alg:continuous2general}
	\begin{algorithmic}[1]
		\Statex \textbf{Input:} A payoff $u$ for which existence of a \maxmin{u}{\sfrac{1}{2}} \cmixedstrategy{2} is guaranteed.
		\Statex \textbf{Output:} Two pure strategies $\b{x}$ and $\b{x}'$ that form a \maxmin{u}{\sfrac{1}{2}} when played with equal prob. $1/2$.
		\State For every battlefield $i$, update $w_i$ to be $\min\{u, w_i\}$.
		\State For $\delta = \epsilon^3/10$, we further update the battlefield weights and consider its $\delta$-uniform variant.\label{line:two}
%		\State $G \gets \emptyset$ \Comment{Will contain the set of all possible guesses for how $x_i$'s compare to $x'_i$'s.}
		\State Ignore every battlefield $i$ with weight less than $\delta u/k$ (i.e., naively guess $x_i \geq x'_i$).\label{line:three}
		\State Put all battlefields of the same weight into the same bucket and denote by $k_j$ the number of battlefields in bucket $j$.
		\State For each bucket $j$ with $k_j \leq 1/\delta$, let $G_j = \{\leq, \geq\}^{k_j}$ be the set of all possible guesses for it.
		\State For each bucket $j$ with $k_j > 1/\delta$, let $G_j$ be the set of all guesses where for any $d \in \{0,\lfloor \delta k_j \rfloor, \lfloor 2\delta k_j \rfloor, \lfloor 3\delta k_j \rfloor, \ldots, k_j\}$ we have $x_i \geq x'_i$ for any $i \leq d$ and $x_i \leq x'_i$ for any $i > d$.\label{line:six}
		\State Let $G = G_1 \times \ldots \times G_b$ be the cartesian product of the partial guesses of the buckets.
		\State For any guess $\b{g} \in G$, construct an instance of LP~\ref{program:four} and return the first found feasible solution $(\b{x}, \b{x}')$.
	\end{algorithmic}	
\end{algorithm}

Before we present a formal proof, we state an auxiliary lemma  to show a lower bound on the value of $u$ when $n \geq (1+\epsilon)m/2$.
\begin{lemma} \label{lemma:notlosemuch}
	Let $\alpha$ be the total sum of the weights of the battlefields whose buckets have a size of at least $1/\delta$. If $n > (1+\epsilon)m / 2$ then there exists a \maxmin{\epsilon \alpha/8}{\sfrac{1}{2}} strategy for player~1 that randomizes over two pure strategies.
\end{lemma}
\begin{proof}	
	Let $\mathcal{B} = \{B_1, \ldots, B_{b}\}$ be the set of all buckets with at least $1/\delta$ battlefields. We slightly abuse the notation and respectively denote by $w_i$ and $k_i$ the weight and the number of battlefields in $B_i$. This means we have $\sum_{i=1}^{b}k_i w_i = \alpha$. We construct a \cmixedstrategy{2} $(\b{x}, \b{x}')$ where both $\b{x}$ and $\b{x}'$ are played with equal probability $\sfrac{1}{2}$ and claim that it is \maxmin{\epsilon \alpha /8}{\sfrac{1}{2}}.
%	To that end, we first assign to each bucket $B_i$ a {\em budget} $d_i := \frac{k_iw_i}{\alpha}n$ \sbcomment{Might wanna assign the budget after getting rid of battlefields to make them even}. Note that
%	\begin{equation}\label{eq:budgets}
%		\sum_{i=1}^{b}\frac{k_i w_i}{\alpha}n = \frac{n}{\alpha} \sum_{i=1}^{b} k_i w_i = \frac{n}{\alpha} \cdot \alpha = n,
%	\end{equation}
%	thus, the total budget assigned to the buckets is exactly $n$. The goal is to put exactly $d_i$ troops over the battlefields in bucket $B_i$ in both strategies $\b{x}$ and $\b{x}'$ (this means that we do not put any troops in the battlefields that are not in $\mathcal{B}$). 
	To that end, for any bucket $B_i \in \mathcal{B}$ with an odd number of battlefields, we ignore one battlefield (i.e., we put zero troops in it in both $\b{x}$ and $\b{x}'$) to consider only an even number of battlefields for each bucket. Denote by $\alpha'$ the total weight of the remaining battlefields, i.e., the battlefields in some $B_i \in \mathcal{B}$ that are not ignored. It is easy to confirm that
	\begin{equation}\label{eq:alphaplarge}
		\alpha' \geq (1-\delta)\alpha
	\end{equation}
	since from each of the buckets in $\mathcal{B}$, at most one battlefield is ignored, which is only a $\delta$ fraction of the battlefields in that bucket since all buckets in $\mathcal{B}$ are assumed to have at least $1/\delta$ battlefields. Now, since only an even number of battlefields remain in each  bucket $B_i$, we can partition them into two disjoint subsets $B^{(1)}_i$ and $B^{(2)}_i$ of equal size. Strategies $\b{x}$ and $\b{x}'$ are then constructed as follows:
	\begin{itemize}
		\item In strategy $\b{x}$, for any $i \in [b]$, we put exactly $2w_i n / \alpha'$ troops in each battlefield in $B^{(1)}_i$. We put zero troops in all other battlefields.
		\item In strategy $\b{x}'$, for any $i \in [b]$, we put exactly $2w_i n / \alpha'$ troops in each battlefield in $B^{(2)}_i$. We put zero troops in all other battlefields.
	\end{itemize}
	Observe that the total number of troops that we use in each of the strategies $\b{x}$ and $\b{x}'$ is exactly $n$ as required. To see this, in strategy $\b{x}$ for instance, the number of troops that are used is
	\begin{equation*}
		\sum_{i\in[b]} \frac{2w_in}{\alpha'} |B^{(1)}_i|= \frac{2n \sum_{i\in[b]}{w_i |B^{(1)}_i|}}{\alpha'} =  \frac{2n \sum_{i\in[b]}{w_i |B^{(1)}_i|}}{\sum_{i\in[b]} w_i \cdot 2|B^{(1)}_i|} = n.
	\end{equation*}
	It only remains to prove that this strategy is \maxmin{\epsilon \alpha/8}{\sfrac{1}{2}}. Assume for the sake of contradiction that player 2 has a strategy $\b{y}$ that prevents both $\b{x}$ and $\b{x}'$ from achieving a payoff of at least $\epsilon \alpha/8$. We can assume w.l.o.g., that for every battlefield $i$, we have $y_i \in \{ 0, x_i, x'_i\}$, thus, on all battlefields that are ignored (i.e., $x_i = x'_i = 0$), we have $y_i = 0$. Further, note that because of the special construction of strategies $\b{x}$ and $\b{x}'$, in each battlefield, at most one of $\b{x}$ or $\b{x}'$ put non-zero troops; therefore, on every battlefield $i$ where $y_i > 0$, either we have $x_i > 0$ or $x'_i > 0$. Define
	\begin{equation*}
		m_1 = \sum_{i \in [k]: x_i > 0} y_i, \qquad\qquad \text{and,} \qquad\qquad m_2 = \sum_{i \in [k]: x'_i > 0} y_i.
	\end{equation*}
	Note that $m_1 + m_2 \leq m$ since it cannot be the case that both $x_i$ and $x'_i$ are non-zero at the same time as described above. Therefore at least one of $m_1$ or $m_2$ is not more than $m/2$. Assume w.l.o.g., that $m_1 \leq m/2$. One can think of $m_1$ as the number of troops that are spent by player 2 to prevent strategy $\b{x}$ from getting a payoff of at least $\epsilon \alpha/8$. To obtain the contradiction, we prove that player 2 cannot use only $m/2$ troops to prevent $\b{x}$ from obtaining a payoff of $\epsilon \alpha/8$. Recall that we denote the total sum of battlefields on which we put a non-zero number of troops either in $\b{x}$ or $\b{x}'$ by $\alpha'$. Strategy $\b{x}$ puts non-zero troops in half of these battlefields, and therefore sum of their weights is at least $\alpha'/2$. To prevent $\b{x}$ from getting a payoff of at least $\epsilon \alpha/8$,  player 2 can lose a weight of less than $\epsilon \alpha/8$ on these battlefields. Let $\mathcal{Y}$ be the subset of battlefields on which $\b{y}$ wins $\b{x}$ and let $w(\mathcal{Y})$ be the total weight of all these battlefields. We need to have 
	\begin{equation}\label{eq:wylarge}
			w(\mathcal{Y}) > \alpha'/2 - \epsilon \alpha/8.
	\end{equation}
	Since the number of troops that is put on the battlefields in $\b{x}$ is proportional to the battlefield weights on which $\b{x}$ puts non-zero troops, we have
	\begin{equation*}
		\sum_{i \in \mathcal{Y}} x_i \geq \Big(\frac{w(\mathcal{Y})}{\alpha'/2}\Big) n \overset{\text{By (\ref{eq:wylarge})}}{\ge} \Big(\frac{\alpha'/2 - \epsilon \alpha/8}{\alpha'/2}\Big) n \geq \Big( 1- \frac{\epsilon \alpha}{4\alpha'} \Big)n \overset{\text{By (\ref{eq:alphaplarge})}}{\ge} \Big(1-\frac{\epsilon}{4(1-\delta)}\Big) n \overset{\text{Since $\delta = \epsilon^3/10$}}{\ge} \frac{n}{1+\epsilon}.
	\end{equation*}
	Therefore, to be able to match the troops of $\b{x}$ in every battlefield in $\mathcal{Y}$, using only $m/2$ troops, we have
	\begin{equation}
		m/2 \geq \sum_{i\in\mathcal{Y}} x_i \geq \frac{n}{1+\epsilon}.
	\end{equation}
	The last inequality contradicts the assumption of the lemma that $n > (1+\epsilon)m/2$. Therefore, there exists no such strategy $\b{y}$ for player 2. That means, the constructed strategy is indeed a $\maxmin{\epsilon \alpha/8}{\sfrac{1}{2}}$ \cmixedstrategy{2}.
\end{proof}

\begin{theorem}\label{theorem:nonuniform-two}
	Let $\epsilon > 0$ be an arbitrarily small constant and suppose that we have an instance of continuous Colonel Blotto in which $n \geq (1+\epsilon)m/2$. Given that player 1 has a \maxmin{u}{\sfrac{1}{2}} \cmixedstrategy{2}, there exists an algorithm that finds a \maxmin{(1-\epsilon)u}{\sfrac{1}{2}} \cmixedstrategy{2} for him in polynomial time.
\end{theorem}

\newcommand{\leqc}{\ensuremath{\text{`$\leq$'}}}
\newcommand{\geqc}{\ensuremath{\text{`$\geq$'}}}
	 
\begin{proof}
%	\sbcomment{\maxmin{u}{p} to \maxmin{u}{\sfrac{1}{2}}.}
	The algorithm to achieve this strategy is given as Algorithm~\ref{alg:continuous2general}. We first analyze the approximation factor of Algorithm~\ref{alg:continuous2general} and then prove that it runs in polynomial time.
	
	\paragraph{Approximation factor.} Let $\b{w}'$ denote the updated battlefield weights by the end of Line~\ref{line:two} of Algorithm~\ref{alg:continuous2general}. First note that setting a cap of $u$ for the battlefield weights in the first line of algorithm does not change the game instance at all since the only goal of player 1 is to guarantee a payoff of at least $u$. Second, we know by Observation~\ref{obs:bfweightrounding} that for any $u'$, any $\maxmin{u'}{p}$ strategy for the $\delta$-uniform variant is a $\maxmin{(1-\delta)u'}{p}$ strategy for the original (i.e., not $\delta$-uniform) instance. Roughly speaking, since $\delta$ is relatively smaller than $\epsilon$, we still get a \maxmin{(1-\epsilon)u}{\sfrac{1}{2}} strategy for the original instance if we achieve a good approximation on the $\delta$-uniform variant. 
	
	Consider an optimal \maxmin{u}{\sfrac{1}{2}} \cmixedstrategy{2} for player 1 on the original instance (which recall is guaranteed to exist) and assume that it randomizes over two pure strategies $\hat{\b{x}}$ and $\hat{\b{x}}'$. By Observation~\ref{obs:bfweightrounding}, this is a \maxmin{(1-\delta)u}{\sfrac{1}{2}} strategy for the $\delta$-uniform variant. Let us denote by vector $\b{c} \in \{\leq, \geq\}^k$ the comparison between the entries of $\hat{\b{x}}$ and $\hat{\b{x}}'$. That is, $c_i$ is \geqc{} if and only if $\hat{x}_i \geq \hat{x}'_i$ and it is \leqc{} otherwise. Our goal is to argue that there exists a guess $\b{g} \in G$ that is {\em sufficiently close} to $\b{c}$ --- where by sufficiently close we mean sum of weights of all battlefields with $g_i \not= c_i$ is very small. We will later combine this with Lemma~\ref{lemma:approx} to obtain the desired guarantee.
	
	We first assume w.l.o.g. that for any two battlefields $i$ and $j$ with $i \leq j$ that have the same weight, if $c_i \not= c_j$, then $c_i = \leqc$ and $c_j = \geqc$ (otherwise we swap these two battlefields without changing the payoff guaranteed by the strategy). We have two relaxations over the guesses in Lines~\ref{line:three} and \ref{line:six} of Algorithm~\ref{alg:continuous2general}.
	
	First, in Line~\ref{line:three} of Algorithm~\ref{alg:continuous2general}, sum of weights of all light battlefields with weight at most $\delta u/k$ is not more than $\delta u$ since $k$ is the total number of battlefields. Thus, even if $g_i \not= c_i$ on these battlefields, their total sum is less than $\delta u$.
	
	Second, in Line~\ref{line:six} of Algorithm~\ref{alg:continuous2general}, let us denote by $\alpha$ the sum of weights of all battlefields whose bucket contains more than $1/\delta$ battlefields. Observe that we check almost all possibilities of guesses on these buckets, except on $\delta$ fraction of their battlefields. More precisely, the total sum of weights of such battlefields on which our guess is wrong is at most $\delta \alpha$. It only remains to argue that $\delta \alpha$ cannot be very large. Note that by Lemma~\ref{lemma:notlosemuch}, we can obtain a simple \maxmin{\epsilon \alpha/8}{\sfrac{1}{2}} strategy since the condition of $n \geq (1+\epsilon)m/2$ is also satisfied here. Thus, we can assume $u \geq \epsilon \alpha/8$ or simply $\alpha \leq 8u/\epsilon$ (otherwise instead of Algorithm~\ref{alg:continuous2general} we return the strategy of Lemma~\ref{lemma:notlosemuch}). Therefore the total weight of battlefields of this type, on which we guess wrong is no more than $$\delta \alpha \leq \frac{\delta 8 u}{\epsilon} \overset{\text{Since $\delta = \epsilon^3/10$}}{\leq} \frac{4}{5}\epsilon^2 u.$$

	Combining these two, we show that there exists a guess $\b{g} \in G$ for which sum of battlefields with $g_i \not= c_i$ is at most $$\delta u + \frac{4}{5} \epsilon^2 u = \frac{1}{10} \epsilon^3u + \frac{4}{5} \epsilon^2 u < \epsilon^2 u.$$ By Lemma~\ref{lemma:approx}, this implies a \maxmin{(1-\delta)u - \epsilon^2 u}{\sfrac{1}{2}} \cmixedstrategy{2} for the $\delta$-uniform variant; and by Observation~\ref{obs:bfweightrounding}, guarantees a utility of at least
	\begin{equation*}
		(1-\delta-\epsilon^2)(1-\delta)u = (1-\frac{\epsilon^3}{10} - \epsilon^2)(1-\frac{\epsilon^3}{10}) u \geq (1-\epsilon) u
	\end{equation*}
	for the original instance with probability at least $\sfrac{1}{2}$, which in other words, gives a \maxmin{(1-\epsilon)u}{\sfrac{1}{2}} \cmixedstrategy{2} as desired.
	
	\paragraph{Running time.} It is easy to confirm that the running time of Algorithm~\ref{alg:continuous2general} is $\poly(|G|)$. Thus, it suffices to show that the total number of guesses in $G$ is polynomial. Observe that for any $i \in [b]$, the total number of partial guesses for bucket $i$ is $O(1)$ (though dependent on $\delta$). On the other hand, the total number of buckets as mentioned before is at most $O(\log k)$ (we hide the dependence on $\delta$) therefore $|G| \leq O(1)^{O(\log k)}$ which is polynomial.
\end{proof}

Notice that when $n \leq m/2$, there is no chance for player 1 to get a nonzero utility by randomizing over two pure strategies since player 2 can always play $y_i = \max\{x_i, x'_i\}$ troops on every every battlefield and win all of them.

\subsection{Generalization to the Case of $c$-Strategies for $c > 2$}\label{sec:contcg}

In this section we generalize our results for the continuous variant to the case of multiple (i.e., more than 2) strategies in the support. That is, for a given $u$, we seek to find a \maxmin{u}{p} \cmixedstrategy{c} $\b{x}$ for maximum possible $p$. Throughout the section, we denote the support of \b{x} by $\b{x}^1, \ldots, \b{x}^c$.

We showed that the only computationally challenging problem for the case of \cmixedstrategies{2} is when our goal is to find \maxmin{u}{\sfrac{1}{2}} strategies. For that we only needed to give two pure strategies $\b{x}^1$ and $\b{x}^2$ and make sure that against every pure strategy of the opponent, at least one of these strategies obtains a utility of at least $u$. This structure becomes more complicated when we allow more than 2 strategies in the support.  Consider, for instance, the case of \cmixedstrategies{3} and suppose for simplicity that we are promised that the three pure strategies in the support of an optimal strategy are each played with probability \sfrac{1}{3}. For this example, the computationally challenging cases are obtaining either a \maxmin{u}{\sfrac{2}{3}} strategy or a \maxmin{u}{\sfrac{1}{3}} strategy. For the former case, we need to make sure that against every strategy of the opponent at least 2 of the strategies in the support obtain a utility of at least $u$. For the latter it suffices for 1 of the strategies to obtain a utility of at least $u$. The idea is to first attempt to find a \maxmin{u}{\sfrac{2}{3}}. It could be the case that no such strategy exists; if so, we then attempt to find a \maxmin{u}{\sfrac{1}{3}} strategy.

To generalize this to more than 3 strategies, we use the notion of {\em non-losing sets}.

\begin{definition}[Non-losing sets]
	Consider a \maxmin{u}{p} \cmixedstrategy{c} \b{x} with support $\b{x}^1, \ldots, \b{x}^c$. We define a set $N \in [c]$ to be a {\em non-losing set} of $\b{x}$ if for every strategy $\b{y}$ of player 2 there exists some $i \in N$ for which $\utilitya{\b{x}^i}{\b{y}} \geq u$. A set $N \in [c]$ is a {\em minimal non-losing} set of $\b{x}$ if it is a non-losing set of $\b{x}$ and there is no strict subset $N' \subsetneq N$ of $N$ that is a non-losing set of \b{x}. We denote the set of all minimal non-losing sets of a strategy \b{x} by \nonlosingset{\b{x}}.
\end{definition}

\newcommand{\Lguess}[0]{\ensuremath{\mathcal{L}}}
\newcommand{\Nguess}[0]{\ensuremath{\mathcal{N}}}

Observe that for a \maxmin{u}{\sfrac{1}{2}} \cmixedstrategy{2}, its only minimal non-losing set is $\{1, 2\}$. Moreover, for our example of a \maxmin{u}{p} \cmixedstrategy{3}, if $p=\sfrac{2}{3}$, the minimal non-losing sets are $\{1, 2\}, \{1, 3\}, \{2, 3\}$. In the same example, if $p = \sfrac{1}{3}$, the only minimal non-losing set is $\{1, 2, 3\}$. 

The general structure of our algorithm is to first guess the set of all minimal non-losing sets $\Nguess$ and see whether it is possible to {\em satisfy} it by constructing the pure strategies $\b{x}^1, \ldots, \b{x}^c$ in such a way that we get $\nonlosingset{\b{x}} = \Nguess$. Recall that there are only a constant number of possibilities for the choice of $\mathcal{N}$ since it is a subset of the power set of $[c]$ and $c$ is constant. We also know by Lemma~\ref{lem:constructprofile} that having \nonlosingset{\b{x}} is sufficient to decide what is the best way to assign probabilities to strategies $\b{x}^1, \ldots, \b{x}^c$ as it uniquely determines the winning subsets. Therefore we can find the optimal \cmixedstrategy{c} if for a given choice of $\mathcal{N}$ we can decide in polynomial time whether it is satisfiable or not.

Now, given that our guess \Nguess{} is fixed, we need to find strategies $\b{x}^1, \ldots, \b{x}^c$ in such a way that every $N \in \Nguess$ is indeed a nonlosing set. This gives the following formulation of the problem.
\begin{equation}\label{prog:cg1}
\begin{array}{ll@{}ll}
\text{find}  & \b{x}^1, \ldots, \b{x}^c &  &\\
\text{subject to}& x^j_i \geq 0  & &\forall i, j: i\in[k], j\in[c] \\
&                  \sum_{i\in[k]} x^j_{i} \leq n         & &\forall j \in [c]\\
&	\text{for some } i\in N \text{ we have } \utilitya{\b{x}^i}{\b{y}} \geq u & & \forall \b{y} \in \puresetb{}, \forall N \in \Nguess
\end{array}\end{equation}
Program~\ref{prog:cg1} is in fact the generalization of Program~\ref{program:two} to the case of \cmixedstrategies{c}. Clearly, the last constraint in its current form is not linear. We showed how Program~\ref{program:two} can be decomposed into a set of convex polytopes and how each of them can be solved in polynomial time via the definition of critical tuples. We follow a similar approach and give a generalized definition of critical tuples.

\begin{definition}[Critical tuples]
	Consider a tuple $\b{W} = (W_1, \ldots, W_k)$ where each $W_i$ is a subset of $[c]$. We call $\b{W}$ a {\em critical tuple} if and only if for some $N \in \Nguess$ we have
	\begin{equation*}
		\sum_{i: j \in W_i } w_i < u \qquad \qquad \forall j\in N.
	\end{equation*}
\end{definition}

Consider a tuple $\b{W} = (W_1, \ldots, W_k)$. Fix a strategy $\b{y}$ of player 2, and assume $\b{x}^1, \ldots, \b{x}^c$ are constructed in such a way that $x^j_i > y_i$ iff $j \in W_i$. In other words, $W_i$ is the set of indices of the strategies that win the $i$th battlefield against $\b{y}$. Now $\b{W}$ is a critical tuple iff $\b{x}^1, \ldots, \b{x}^c$ do not satisfy \Nguess{}. More precisely, $\b{W}$ is a critical tuple iff all of the strategies in one of the minimal non-losing sets of \Nguess{} lose (i.e., get a payoff of less than $u$) against $\b{y}$.

Before describing how we can use critical tuples in rewriting Program~\ref{prog:cg1}, we need to know precisely how for each $i\in[k]$, the values of $x^1_i, x^2_i, \ldots, x^c_i$ compare to each other. To that end, we define a {\em configuration} $\b{G}$ to be a vector of $k$ matrices $G_1, \ldots, G_k$ which we call {\em partial configurations}, where for any $i \in [k]$, and for any $j_1, j_2 \in [c]$, the value of $G_i(j_1, j_2)$ is $\leqc$ if $x^{j_1}_i \leq x^{j_2}_i$ and it is $\geqc$ otherwise.

Clearly if we fix the configuration \b{G} of strategies a priori, it is possible to ensure the found strategies comply with it via $O(kc^2)$ linear constraints. It suffices to have one constraint for every $G_i(j_1, j_2)$. The linear program below shows that if the configuration is fixed, we can even rewrite Program~\ref{prog:cg1} as a linear program. For configuration $\b{G}$ and critical tuple $\b{W}$, define $z_i(\b{G}, \b{W}) := \argmax_{j: j \not\in W_i}{x^j_i}$. Note that it is crucial that $z_i(\b{G}, \b{W})$ is solely a function of $\b{G}$ and $\b{W}$ (and not the actual strategies $\b{x}^1, \ldots, \b{x}^c$) so long as strategies $\b{x}^1, \ldots, \b{x}^c$ comply with \b{G}.
\begin{equation}\label{lp:cg2}
\begin{array}{ll@{}ll}
\text{find}  & \b{x}^1, \ldots, \b{x}^c &  &\\
\text{subject to}& x^j_i \geq 0  & &\forall i, j: i\in[k], j\in[c] \\
&                  \sum_{i\in[k]} x^j_{i} \leq n         & &\forall j \in [c]\\
&  \text{ensure that }\b{x}^1, \ldots, \b{x}^c \text{ comply with } \b{G}\\
&	\sum_{i \in [k]} x^{z_i(\b{G}, \b{W})}_i > m & & \text{for every critical tuple $\b{W}=(W_1, \ldots, W_k)$}
\end{array}\end{equation}
The intuition behind the last constraint is that for player 2 to be able to enforce any critical tuple to happen, he needs to have more than $m$ troops. This is a sufficient and necessary condition to ensure that all of the non-losing sets in \Nguess{} are satisfied. Observe that if our guess for  \Nguess{} is wrong, for every configuration \b{G}, LP~\ref{lp:cg2} will be infeasible.

The takeaway from LP~\ref{lp:cg2}, is that if we fix the configuration \b{G}, the solution space becomes convex and can be described via linear constraints. Although there may be exponentially many critical tuples \b{W} and, thus, exponentially many constraints in LP~\ref{lp:cg2}, one can design an appropriate separating oracle and use ellipsoid method to solve it in polynomial time using a dynamic programming approach similar to the one used for solving LP~\ref{program:four}. Therefore one algorithm to solve the problem is to iterate over all possible configurations, solve LP~\ref{lp:cg2} for each and report the best solution. This gives us an exponential time algorithm with running time $(c!)^{k}$. Recall that we followed a rather similar approach for the case of \cmixedstrategies{2} and showed that the solution space can be decomposed to $2^k$ convex polytopes. To overcome this when $c=2$, we showed how it is possible to only consider polynomially many such polytopes as far as the problem is concerned. We follow the same approach here.

Let us start with uniform setting where all the battlefield weights are the same. Since the players are both indifferent to the battlefields in the uniform setting, instead of individually fixing the partial configuration of each battlefield, it only suffices to know the number of battlefields having a particular partial configuration. Note that each partial configuration, e.g., $G_1$, is determined uniquely if we are given the sorted order of $x^1_1, x^2_1, \ldots, x^c_1$. Therefore, there are a total number of $c! = O(1)$ possibilities for each partial configuration. This means, if we only count the number of battlefields having each partial configuration, we reduce the total number of considered configurations down to a polynomial ($O(k^{c!})$ to be more precise). This gives a polynomial time algorithm for the uniform case.

Generalization to the case where the weights are not are equal follows from similar ideas described in Section~\ref{section:nonuniformc=2}. That is, we can consider the $\delta$-uniform variant of the game for a relatively smaller error threshold than $\epsilon$ and group battlefields into buckets with each bucket containing the battlefields of the same updated weight. Then in each bucket, similar to the uniform case, it suffices to only count the number of battlefields of each configuration. Recall, however, that the crucial property for this idea to work was to show that for each bucket, it suffices to check only a constant number of different possibilities. Naively, the number of possibilities for each bucket is $k_i^{O(1)}$. Therefore, if $k_i \leq 1/\delta$ it is bounded by $O(1)$ as desired. However, if $k_i > 1/\delta$, the idea, similar to Section~\ref{section:nonuniformc=2}, is to discretize the number of battlefields having one particular partial configuration of $G_j$ to be in set $\{0, \delta k_i, 2\delta k_i, \ldots, k_i\}$. This way, a similar argument as in Lemma~\ref{lemma:notlosemuch} shows that making a mistake on only a $\delta$ fraction of buckets with at least $1/\delta$ battlefields is negligible so long as our goal is to guarantee a utility of $(1-\epsilon)u  $, resulting in the following theorem.

\begin{theorem}\label{thm:nonuniform-c}
	Let $\epsilon > 0$ be an arbitrarily small constant. Given that player 1 in an instance of continuous Colonel Blotto has a \maxmin{u}{p} \cmixedstrategy{c} for a constant $c$, and given that $n \geq (1+\epsilon)m/\lfloor (1-p)c+1 \rfloor$, there exists an algorithm that finds a \maxmin{(1-\epsilon)u}{p} \cmixedstrategy{c} for him in polynomial time.
\end{theorem}

We remark that if $n \leq m/\lfloor (1-p)c+1 \rfloor$, then player 1 has no \maxmin{u}{p} strategy.

\section{Discrete Colonel Blotto}\label{sec:integral}
In this section we focus on the discrete variant of Colonel Blotto. Though similar in spirit to the continuous variant, discrete Colonel Blotto is an inherently different game (especially from a computational perspective) and requires different techniques.

\subsection{The Case of One Strategy} \label{sec:one-pure}
Similar to the continuous variant, we start with the case where the support size is bounded by 1. We showed how for the continuous case, it is possible to obtain an optimal \maxmin{u}{1} strategy by solving LP~\ref{lp:one}. Observe, however, that variable $x_i$ in LP~\ref{lp:one} denotes the number of troops in battlefield $i$; this implies that LP~\ref{lp:one} relies crucially on the fact that a fractional number of troops in a battlefield is allowed. As such, the same idea cannot be applied to the discrete case since the integer variant of LP~\ref{lp:one} is not necessarily solvable in polynomial time.

Prior algorithm of \cite{behnezhad2018battlefields} gives a 2-approximation for this problem, i.e., given that there exists a \maxmin{u}{1} strategy for player 1, they give a \maxmin{u/2}{1} strategy in polynomial time. We improve this result by obtaining an almost optimal solution in polynomial time.

\newcommand{\thmdiscretepure}{Given that player 1 has a \maxmin{u}{1} strategy, there exists a polynomial time algorithm that obtains a \maxmin{(1-\epsilon)u}{1} strategy of player 1 for any arbitrarily small constant $\epsilon > 0$.}

\begin{theorem}\label{thm:discretepure}
	\thmdiscretepure{}
\end{theorem}

\paragraph{The algorithm in a nutshell.} The algorithm that we use to prove Theorem~\ref{thm:discretepure} is composed of three main steps. In step 1, we round down the battlefield weights to be powers of $(1+\epsilon/2)$, i.e., we consider the $(\epsilon/2)$-uniform variant of the game. The goal here is to reduce the number of distinct battlefield weights. As it was previously shown, by optimizing over this updated instance of the game, player 1 does not lose a considerable payoff on the original instance.

In step 2, we partition the battlefields into two subsets: {\em light} battlefields that have a weight of at most $(\epsilon/2) u$ and {\em heavy} battlefields that each has a weight of more than $(\epsilon/2) u$. Roughly speaking, the goal is to separate battlefields that have a high impact on the outcome of the game from the lower weight ones. We show that as a result of step 1, we can give a new representation for strategies of player 1 that reduces the total number of {\em partial strategies} of player 1 on heavy battlefields down to a polynomial. We further show that for any strategy of player 1 on heavy battlefields, player 2 has only a constant number of {\em valid} responses as far as the optimal solution is concerned. Importantly, bounding the number of responses of player 2 by a constant has a crucial role in solving the problem in polynomial time --- we will elaborate more on this in the next paragraph.

In step 3, we propose a {\em weaker adversary} than player 2. Roughly speaking, we assume that for any given strategy of player 1, the weaker adversary responds greedily on light battlefields (though we do not limit him on heavy battlefields). We show that optimizing player 1's strategy against this weaker adversary guarantees an acceptable payoff against the actual adversary (i.e., player 2). In brief, the main advantage of optimizing player 1's strategy against the weaker adversary is in that it allows us to exploit the more predictable greedy response of the opponent. Recall, however, that we do not limit the weaker adversary on heavy battlefields and, therefore, his response to a strategy that we give for player 1 may still come from a somewhat unpredictable function. However, step 2 guarantees that for every strategy of player 1, it suffices to consider only a constant number of responses of the weaker adversary. This allows us to use a dynamic program that has, roughly, the same number of dimensions as the number of strategies of the weaker adversary and solve the problem in polynomial time.

\paragraph{Basic structural properties.} Recall that a \maxmin{u}{1} strategy guarantees a payoff of $u$ against {\em any} strategy of player 2. Therefore, if $u > 0$, we need to have $n > m$; or otherwise, no matter what pure strategy the first player chooses, the second player can match the number of troops of the first player in all battlefields and win them all. Now, assuming that $n > m$, if there exists one battlefield with weight at least $u$, the first player can simply put all his troops in that battlefield and guarantee a payoff of at least $u$. Therefore we assume throughout the rest of the section that $n > m$ and that the maximum battlefield weight is less than $u$.

\paragraph{Step 1: Updating the battlefield weights.} We first {\em round down} the weight of each of the battlefields to be in set $\mathcal{W} = \{ 1, (1+\epsilon/2)^1, (1+\epsilon/2)^2, \ldots \}$ and denote the modified weights vector by $\b{w}'$. That is, we consider the $(\epsilon/2)$-uniform variant of the game (see Section~\ref{sec:continuous} and Observation~\ref{obs:bfweightrounding} for the formal definition).

\begin{corollary}[of Observation~\ref{obs:bfweightrounding}]\label{cor:bfweightrounding}
	For any $u'$, any \maxmin{u'}{1} strategy of the game instance \cbinstance{n}{m}{\b{w'}} with the updated weights is a \maxmin{(1-\epsilon/2)u'}{1} strategy of the original instance \cbinstance{n}{m}{\b{w}}.
\end{corollary}

Corollary~\ref{cor:bfweightrounding} implies that to obtain a $\maxmin{(1-\epsilon)u}{1}$ strategy on the original instance \cbinstance{n}{m}{\b{w}}, it suffices to find a \maxmin{(1-\epsilon/2)u}{1} strategy on instance \cbinstance{n}{m}{\b{w}'}. Therefore, from now on, we only focus on the instance with the updated weights and, for simplicity of notations, denote the updated weight of battlefield $i$ by $w_i$.

Updating the battlefield weights in the aforementioned way results in reducing the total number of distinct  weights down to $O(1/\epsilon \cdot \log u)$ (recall that the maximum weight among original weights was assumed to be at most $u$). We later show how this can be used to represent pure strategies of player 1 in a different way that results in a significantly fewer number of strategies.

\paragraph{Step 2: Partitioning the battlefields into heavy and light subsets.} We set a threshold $\tau = \epsilon u/2$ and partition the battlefields into two subsets of {\em heavy} battlefields with weights of at least $\tau$ and {\em light} battlefields with weights of less than $\tau$. Roughly speaking, for light battlefields, we can afford to be the weight of one battlefield away from the optimal strategy and remain $(1-\epsilon/2)$-competitive since the maximum weight among them is bounded by $\epsilon u / 2$. For the heavy battlefields, however, we need to be more careful as even one battlefield might have a huge impact on the outcome of the game. The idea is to significantly reduce the number of strategies of the players that have to be considered on heavy battlefields so that we can consider them all. Let us denote by $k_d$ the number of distinct battlefield weights; further, we denote by $k^h_d$ the number of distinct battlefield weights that are larger than $\tau$ (i.e., are heavy). The following observation bounds $k^h_d$ by a constant.

\begin{observation}\label{obs:constantheavybfw}
	The number of distinct heavy battlefield weights, $k^h_d$, is bounded by a constant.
\end{observation}
\begin{proof}
	Recall that we assume no battlefield has weight more than $u$. Moreover, all heavy battlefields have weight at least $\epsilon u/2$. Since we updated the weight of all battlefields to be in set $\{ 1, (1+\epsilon/2)^1, (1+\epsilon/2)^2, \ldots \}$, it leaves only $\log_{1+\epsilon/2} (2/\epsilon) = \frac{\log (2/\epsilon)}{\log (1+\epsilon/2)} \leq 4/\epsilon \cdot \log (2/\epsilon)$ values in range $[\epsilon u/2, u]$ which is a constant number since $\epsilon$ is a constant.
\end{proof}

We show by Claim~\ref{cl:samewsametroops} that there always exists an optimal strategy that puts roughly the same number of troops in all battlefields of the same weight. Intuitively, the players are indifferent to battlefields of the same weights and therefore the first player has no incentive to put significantly more troops on one of them.

\begin{claim}\label{cl:samewsametroops}
	If player 1 has a \maxmin{u}{1} strategy, he also has a \maxmin{u}{1} strategy $\b{x}$ where for any two battlefields $i$ and $j$ with $w_i = w_j$, we have $|x_i - x_j| \leq 1$.
\end{claim}
\begin{proof}
	Consider two battlefields $j$ and $j'$ with the same weights and assume that for \maxmin{u}{1} strategy $\b{\hat{x}}$, we have $\hat{x}_{j} \geq \hat{x}_{j'} + 2$. We construct another strategy $\b{x}'$ in the following way: $x'_j = \hat{x}_{j}-1$, $x'_{j'} = \hat{x}_{j'} + 1$, and $x'_i = \hat{x}_i$ for all other battlefields. It is easy to confirm that $\b{x}'$ is also a \maxmin{u}{1} strategy. Applying the same function to $\b{x}'$, we can inductively update the strategy to finally achieve a strategy \maxmin{u}{1} strategy $\b{x}$ for which we have $|x_i - x_j| \geq 1$ for all battlefields of the same weight.
\end{proof}

Now, the idea is to represent the pure strategies of player 1 differently. That is, instead of specifying the number of troops that are put on each battlefield, we represent each valid pure strategy of player 1 by the number of troops that are put on each battlefield weight. By Claim~\ref{cl:samewsametroops}, this uniquely determines an optimal strategy without loss of generality. Therefore the dimension of each pure strategies of player 1 is now changed from $k$ to $O(1/\epsilon \cdot \log u)$. Note that this is not sufficient to solve the problem in polynomial time since the number of pure strategies of player 1 can still be up to $n^{O(1/\epsilon \cdot \log u)}$. However, we show how this can bound the number of pure strategies of player 1 on heavy battlefields.

We define a {\em partial strategy} of player 1 on heavy battlefields to be a strategy in the new representation that only specifies how player 1 plays on heavy battlefields. Formally, each partial strategy of player 1 on heavy battlefields can be represented by a vector $\b{x}^h$ of length $k^h_d$, with non-negative entries in $[n]$ that sum up to at most $n$. Let us denote the set of all such strategies by \puresetaheavy{}. The following observation bounds the total number of strategies in \puresetaheavy{}.

\begin{observation}\label{obs:puresetaheavypoly}
	$|\puresetaheavy{}| \leq n^{\poly(1/\epsilon)}$.
\end{observation}
\begin{proof}
	Since by Observation~\ref{obs:constantheavybfw}, $k^h_d \leq \poly(1/\epsilon)$, each strategy in \puresetaheavy{} is a vector of length at most $\poly(1/\epsilon)$. This, combined with the fact that each entry is an integer between 0 and $n$, leads to the desired bound.
\end{proof}

A similar analysis can show that the strategies of player 2 can also be bounded by a polynomial on heavy battlefields; however, for technical reasons that we will elaborate on later, it is crucial to further bound the number of pure strategies of player 2 on heavy battlefields by a constant. The key idea here, is that to prevent player 1 from achieving a payoff of $u$, player 2 can only lose in at most $2/\epsilon$ heavy battlefields. Similar to that of player 1, we represent each partial strategy of player 2 on heavy battlefields by a vector of length $k^h_d$. However, instead of specifying the number of troops that player 2  puts on each battlefield weight, we only specify how many battlefields of each weight player 2 loses in. Therefore, formally, each partial strategy of player 2 on heavy battlefields can be represented by a vector of length $k^h_d$ where each entry is a non-negative integer and all entries sum up to at most $2/\epsilon$. Denoting the set of all these partial strategies by $\puresetbheavy{}$, we can bound its size to be a constant.

\begin{observation}\label{obs:puresetbheavyconstant}
	$|\puresetbheavy{}| \leq O(1)$.
\end{observation}
\begin{proof}
	By Observation~\ref{obs:constantheavybfw}, the vectors in $\puresetbheavy{}$ have $k^h_d \leq \poly(1/\epsilon)$ dimensions and each entry is a number in $[0, 2/\epsilon]$ therefore there are at most $\poly(1/\epsilon)^{2/\epsilon+1}$ such vectors which is $O(1)$ since $\epsilon$ is a constant.
\end{proof}

\paragraph{Step 3: Solving the problem against a weaker adversary.} Given a strategy $\b{x}$ of player 1, the second player's best response is a strategy that maximizes his payoff. That is, player 2 seeks to find a strategy $\b{y} \in \puresetb$ such that $\utilityb{\b{x}}{\b{y}}$, which is the weight of the battlefields in which $\b{y}$ puts more troops than $\b{x}$, is maximized. This is precisely equivalent to the following knapsack problem: The knapsack has capacity $m$ (the number of troops of player 2) and for any $i \in [k]$, there exists an item with weight $x_i$ and value $w_i$. Indeed one of the main challenges in finding an optimal \maxmin{u}{1} strategy for player 1 is in that the second player's best response problem is a rather complicated function that makes it hard for the first player to optimize his strategy against. To resolve this issue, we optimize player 1's strategy against a {\em weaker adversary} than player 2. That is, we assume that the second player, instead of solving the aforementioned knapsack problem optimally, follows a simpler and more predictable algorithm. Note that there is a trade-off on the best-response algorithm that we fix for the weaker adversary. If it is too simplified, the strategy that we find for player 1 against it cannot perform well against player 2 who best-responds optimally; and if it is not simplified enough, there is no gain in considering the weaker adversary instead of player 2. To balance this, we allow the weak adversary to try all possible strategies on heavy battlefields, but only allow him to combine it with a greedy algorithm on light battlefields; this process is formalized in Algorithm~\ref{alg:weakresponse}.\footnote{For ease of exposition, we do not explicitly mention how to sort the battlefields when two have the same ratio in  Algorithm~\ref{alg:weakresponse}. This can be simply handled by assuming that the battlefield with a lower index is preferred by the weaker adversary.} As we show in the rest of this section, while being only $\epsilon$ away from the optimal best-response algorithm, it becomes possible to compute an optimal strategy of player 1 against the weak adversary who responds by Algorithm~\ref{alg:weakresponse} in polynomial time.

\begin{algorithm}
	\caption{Weaker adversary's best response algorithm.}
	\label{alg:weakresponse}
	\begin{algorithmic}[1]
		\Statex \textbf{Input:} a strategy $\b{x} = (x_1, \ldots, x_k)$ of player 1.
		\Statex \textbf{Output:} a strategy $\b{y} = (y_1, \ldots, y_k)$ of player 2.
		\State For every battlefield $i$, denote its ratio $r_i$ to be $w_i / x_i$.
		\State Let $\sigma(i)$ denote the index of the $i$th \underline{light} battlefield with the highest ratio.
		\State $\b{y} \gets (0, \ldots, 0)$
		\For {every possible strategy $\b{y}'$ of player 2 that puts non-zero troops only on heavy battlefields}
			\State $i \gets 1$
			\While {$m - \sum_i y'_i \geq x_{\sigma(i)}$ and $i \leq k$} \Comment{\textit{Play greedily on light battlefields based on their ratios.}}
				\State $y'_{\sigma(i)} \gets x_{\sigma(i)}$
				\State $i \gets i+1$
			\EndWhile
			\IIf {$\utilityb{\b{x}}{\b{y}'} \geq \utilityb{\b{x}}{\b{y}}$} update $\b{y} \gets \b{y}'$
		\EndFor
		\State \Return $\b{y}$
	\end{algorithmic}	
\end{algorithm}

\begin{lemma}\label{lem:weakisgood}
	For any strategy $\b{x}$ of player 1, let $\b{y}$ be the optimal best-response of player 2 and let $\tilde{\b{y}}$ be the strategy obtained by Algorithm~\ref{alg:weakresponse}. We have $\utilityb{\b{x}}{\tilde{\b{y}}} \geq \utilityb{\b{x}}{\b{y}} - \frac{\epsilon u}{2}$.
\end{lemma}
\begin{proof}
	Consider the iteration of the for loop in Algorithm~\ref{alg:weakresponse} where the chosen partial strategy $\b{y}'$ over the heavy battlefields is exactly equal to that of the optimal best-response $\b{y}$ of player 2. We claim that after this strategy is combined with the greedy algorithm over light battlefields to obtain a potential strategy $\b{p}$, it provides a payoff of at least $\utilityb{\b{x}}{\b{y}} - \epsilon u / 2$. Note that this is sufficient to prove the lemma since the returned solution $\b{y}$ guarantees a payoff against $\b{x}$ that is not less than that of $\b{p}$.
	
	Observe that we can safely ignore heavy battlefields since both $\b{y}$ and $\b{p}$ choose similar strategies over them. It only suffices to guarantee that the difference over light battlefields is at most $\epsilon u /2$. Let us denote by $m'$ the number of troops that are left for player 2 over light battlefields. The optimal best-response is equivalent to the solution of the following knapsack problem: The knapsack has capacity $m'$, and for each light battlefield $i$, there is an item with cost $x_i$ (since player 2 has to put at least $x_i$ troops to win it) and value $w_i$. Observe that the greedy approach of the weaker-adversary on light battlefields, is equivalent to a greedy algorithm for this knapsack problem where items are sorted based on their value per weight (i.e., $w_i / x_i$) and chosen greedily. This is in fact, the well-known greedy algorithm of knapsack that is known to guarantee an additive error of up to the maximum item value. Since here we only consider light battlefields and the maximum item value, or equivalently, the maximum battlefield weight is at most $\epsilon u / 2$, we lose an additive error of at most $\epsilon u / 2$ compared to optimal strategy $\b{y}$, which concludes the proof.
\end{proof}

\begin{corollary}\label{cor:weakeradversary}
	Any strategy $\b{x}$ of player 1 that obtains a payoff of $u$ against the weaker adversary, obtains a payoff of at least $(1-\epsilon/2)u$ against all possible strategies of player 2.
\end{corollary}
\begin{proof}
	Since the game is constant sum with $\sum_i w_i$ being sum of utilities, the weaker adversary gets a payoff of at most $\sum_i w_i - u$ against $\b{x}$. Therefore since by Lemma~\ref{lem:weakisgood} the optimal best response of player 2 obtains a payoff of at most $\sum_i w_i - u + \epsilon u /2$. This means that the first player gets a payoff of at least $u - \epsilon u / 2 = (1-\epsilon / 2) u$ against {\em any} strategy of the second player.
\end{proof}

By Corollary~\ref{cor:weakeradversary}, to obtain the desired $\maxmin{(1-\epsilon/2)u}{1}$ strategy, it suffices to guarantee a payoff of $u$ against the weaker adversary. Note that such strategy is guaranteed to exist since existence of a $\maxmin{u}{1}$ strategy against player 2 is guaranteed.

\paragraph{A wrong approach.} We start with a wrong approach in guaranteeing a payoff of $u$ against the weaker adversary and later show how to fix its shortcomings. Observe that the weaker adversary exhaustively searches through all of his possible strategies among heavy battlefields and plays greedily on light battlefields. Therefore, the first player needs to approach heavy and light battlefields differently. One way of doing this is to fix the number of troops that each of the players will spend on heavy and light battlefields, which is bounded by $O(nm)$, and solve two instances independently. On light battlefields, since player 2 responds greedily, it is not hard to optimize the first player's strategy.  On heavy battlefields, the number of pairs of valid strategies of the players are bounded by a polynomial by Observations~\ref{obs:puresetaheavypoly} and \ref{obs:puresetbheavyconstant}. Therefore, player 1 can find his optimal strategy among them by exhaustively checking all of his strategies against all possible responses of player 2.

The problem with this approach, is that it is not possible to solve the subproblems on heavy and light battlefields independently. More precisely, the number of troops that player 2 puts on heavy and light battlefields is a direct function of the strategy of player 1. This implies that it is not possible to fix the number of troops of player 2 over heavy and light battlefields a priori, without specifying the exact strategy of player 1.

\paragraph{The correct approach.} To resolve the aforementioned problem, we do not fix the budget of player 2 on light and heavy battlefields beforehand. Instead, we first only fix the strategy $\b{x}^h \in \puresetaheavy$ of player 1 on heavy battlefields. (By Observation~\ref{obs:puresetaheavypoly} there are only polynomially many such strategies, therefore it is possible to try them all.) Next, note that player 2, by Observation~\ref{obs:puresetbheavyconstant}, has only a constant number of responses to $\b{x}^h$ that we would have to consider (other responses simply guarantee us a payoff of $u$ on heavy battlefields alone). Roughly speaking, while finding the optimal strategy of player 1 on light battlefields, we would have to consider all these responses of player 2 on heavy battlefields. That is, each response $\b{y}^h \in \puresetbheavy$ of player 2 on heavy battlefields uniquely determines (1) what payoff the players get on heavy battlefields, (2) how many troops is left for player 2 to play on light battlefields. The strategy that we find for player 1 on light battlefields, has to perform well against all of these responses.

In a thought experiment, consider the optimal strategy $\b{x}^\ell$ of player 1 on light battlefields, given that his strategy on heavy battlefields is fixed to be $\b{x}^h$. For every strategy $\b{y}^h$ of player 2 on heavy battlefields, the greedy algorithm of player 2 on light battlefields first sorts light battlefields based on their ratios $w_i / x^\ell_i$ and greedily wins battlefields of higher ratio. This means that for every response of player 2 on heavy battlefields, there is a unique light battlefield of highest ratio that player 2 loses in. Let us denote by $b_i$ the index of this battlefield for the $i$th response of player 2 in $\puresetbheavy$. We also denote by $r_i := w_{b_i} / x^\ell_{b_i}$ the ratio of battlefield $b_i$ based on $\b{x}^\ell$. The takeaway, here, is that given the optimal strategy $\b{x}^\ell$ of player 2 on light battlefields, we can uniquely determine vectors $\b{b} = (b_1, b_2, \ldots, b_{|\puresetbheavy|})$ and $\b{r} = (r_1, r_2, \ldots, r_{|\puresetbheavy|})$. We have no way of knowing these vectors without knowing $\b{x}^\ell$, however, we find the right value of them by checking all possible vectors. To achieve this, the first step is to show that there are only polynomially many options that we need to try.

\begin{claim}\label{claim:polytriple}
	There are only polynomially many valid triplets $(\b{x}^h, \b{b}, \b{r})$.
\end{claim}
\begin{proof}
	We already know from Observation~\ref{obs:puresetaheavypoly} that there are only polynomially many possible choices for $\b{x}^h$. To conclude the proof, it suffices to show that for the choice of $\b{b}$ and $\b{r}$, we also have polynomially many possibilities. To see this, note that $\b{b}$ and $\b{r}$ only have $|\puresetbheavy|$ dimensions which is bounded by a constant by Observation~\ref{obs:puresetbheavyconstant}. Furthermore, each entry of $\b{b}$ is a number in $\{0, 1, \ldots, k\}$ and each entry of $\b{r}$ has at most $O(nk)$ values. This means there are only $k^{O(1)}$ many possibilities for $\b{b}$ and $(nk)^{O(1)}$ many possibilities for $\b{r}$. Both of these upper bounds are polynomial, which proves there are only polynomially many triplets $(\b{x}^h, \b{b}, \b{r})$.
\end{proof}

The next step is to find a strategy $\b{x}^\ell$ of player 1 on light battlefields that {\em satisfies} a given triplet $(\b{x^h}, \b{b}, \b{r})$. We start by formalizing what satisfying exactly means.

\begin{definition}\label{def:sat}
	We say a strategy $\b{x}^\ell$ of player 1 on light battlefields satisfies a triplet $(\b{x}^h, \b{b}, \b{r})$ if the following conditions hold.
	\begin{enumerate}[topsep=0pt,itemsep=-1ex,partopsep=1ex,parsep=1ex]
		\item The number of troops that are used in $\b{x}^\ell$ is at most $n - \sum_i x^h_i$.
		\item If player 2 chooses his $i$th strategy from $\puresetbheavy$ over heavy battlefields and plays greedily (according to Algorithm~\ref{alg:weakresponse}) on light battlefields, the highest ratio light battlefield in which he loses becomes $b_i$.
		\item The ratio $w_{b_i} / x^\ell_{b_i}$ is indeed equal to $r_i$.
	\end{enumerate}
\end{definition}

Note that it is not necessarily possible to satisfy any given triplet $(\b{x^h}, \b{b}, \b{r})$. However, it is guaranteed that the right choice of it, where $\b{x^h}$ corresponds to the optimal partial strategy of player 1 on heavy battlefields and $\b{b}$ and $\b{r}$ correspond to their actual value (as described before) for the optimal strategy of player 1, is satisfiable. Recall that our goal is to exhaustively try all possible triplets and find the optimal one. To achieve this, we need an oracle that confirms whether a given triplet is satisfiable. Lemma~\ref{lem:findstrata} provides this oracle via a dynamic program that further guarantees that the obtained strategy is maximin against the weak adversary given that player 1 is committed to satisfy $(\b{x}^h, \b{b}, \b{r})$.

\begin{lemma}\label{lem:findstrata}
	Given a triplet $\mathcal{T} = (\b{x}^h, \b{b}, \b{r})$, one can in polynomial time, either report that $\mathcal{T}$ is not satisfiable, or find a strategy $\b{x}^\ell$ that satisfies $\mathcal{T}$ while guaranteeing that the payoff of the response of the weaker adversary to strategy $(\b{x}^h, \b{x}^\ell)$ of player 1 that is obtained by Algorithm~\ref{alg:weakresponse} is minimized.
\end{lemma}
\begin{proof}
	For simplicity of notations, let us denote by $c := |\puresetbheavy|$ the total number of responses of player 2 over the heavy battlefields, and consequently, the size of vectors $\b{b}$ and $\b{r}$. Furthermore, given that player 1 plays strategy $\b{x}^h$ on heavy battlefields and that player 2 plays his $i$th response in $\puresetbheavy$ to $\b{x}^h$, we denote by $m_i$ the number of troops that are left for player 2 to play on light battlefields, and denote respectively by $g_{1}(i)$ and $g_2(i)$ the guaranteed payoff that players 1 and 2 get on heavy battlefields.
	
To satisfy condition 1 of Definition~\ref{def:sat}, we ensure that the strategy that we find for player 1 over the light battlefields only uses $n^\ell := n - \sum_i x^h_i$ troops. Satisfying the third condition is also straightforward: for any $i \in [c]$, it suffices for player 1 to put exactly $w_{b_i} / r_i$ troops on battlefield $b_i$. (We emphasize that we should not change the number of troops on these battlefields throughout the algorithm.) If during this process, we need to use more than $n^\ell$ troops or if for some $i$ and $j$ we have $b_i = b_j$ and $w_{b_i} / r_i \not = w_{b_j} / r_j$ we report that the given triplet is not satisfiable. The main difficulty is to ensure that the second condition of Definition~\ref{def:sat} also holds. For that, the only decision that we have to make is on the number of troops that we put on the remaining battlefields. For our final strategy $\b{x}^\ell$ over the light battlefields, assume that the battlefields are sorted decreasingly based on their ratio $w_i / x^\ell_i$ and let us denote by $\sigma(i)$ the index of the battlefield in position $i$. To convey the overall idea of how to satisfy this condition, let us first focus on the $i$th strategy of player 2 over the heavy battlefields. We need to make sure that in our strategy $\b{x}^\ell$ over light battlefields, the ratio of battlefields are such that if player 2 wins them greedily, he stops at battlefield $b_i$. More precisely, let $ \sigma(\gamma) = b_i$, we need to have\footnote{In case of a tie, we assume the battlefield with the lower index has a higher ratio.} 
\begin{equation}\label{eq:rank}
0 \leq m_i - \sum_{j=1}^{\gamma-1} x^\ell_{\sigma(j)} < x^\ell_{b_i}.
\end{equation}
Fix battlefield $b_i$, roughly speaking, in order to satisfy (\ref{eq:rank}), we need to decide which battlefields will have a higher ratio than battlefield $b_i$ so as to ensure that once player 2 wins them, he will have less than $x^\ell_{b_i}$ troops. Now recall that we need to satisfy this, simultaneously, for every entry $b_i$ of $\b{b}$. For this we use a dynamic program $D(j, n', \omega_1, \ldots, \omega_c, u_1, \ldots, u_c)$ with $j \in \{0, \ldots, k^\ell\}$, $n' \in \{0, , \ldots, n^\ell\}$, $u_i \in \{0, \ldots, \sum_i w_i\}$, and $\omega_i \in \{0, 1 , \ldots, m_i\}$. The value of $D(j, n', \omega_1, \ldots, \omega_c, u_1, \ldots, u_c)$ is either 0 or 1 and it is 1 iff it is possible to give a partial strategy $(x_1, \ldots, x_j)$ over the first $j$ light battlefields such that all the following conditions are satisfied.
	\begin{enumerate}
		\item We use exactly $n'$ troops over them, i.e., 
		\begin{equation}\label{eq:con1}
			\sum_{j' \in [j]} x_{j'} = n'.
		\end{equation}
		\item For any $i \in [c]$, we put at least $\omega_i$ troops in battlefields with ratio higher than $r_i$, i.e., 
		\begin{equation}\label{eq:con2}\sum_{j' \in [j] : \frac{w_{j'}}{x_{j'}} \geq r_i} x_{j'} \geq \omega_i, \qquad \forall i \in [c].	
		\end{equation}

		\item For any $i \in [c]$, sum of weights of battlefields with lower ratio than $r_i$ is at least $u_i - g_1(i)$, i.e., 
		\begin{equation}\label{eq:con3}\sum_{j' \in [j] : \frac{w_{j'}}{x_{j'}} < r_i} w_{j'} \geq u_i-g_1(i), \qquad \forall i \in [c].
		\end{equation}

	\end{enumerate}
%\begin{equation*}
%	D(j, n', \omega_1, \ldots, \omega_c, u_1, \ldots, u_c) = \begin{cases}
%	1, & \text{\parbox{10.2cm}{if it is possible to give a partial strategy $(x_1, \ldots, x_j)$ over the first $j$ light battlefields such that,
%	\begin{itemize}[topsep=0pt,itemsep=-1ex,partopsep=1ex,parsep=1ex]
%		\item we use exactly $n^\ell$ troops over them (i.e., $\sum_{j' \in [j]} x_i = n^\ell$),
%		\item for any $i \in [c]$, we put at least $\omega_i$ troops in battlefields with higher ratio than $r_i$ (i.e., $\sum_{j' \in [j] : w_{j'}/x_{j'} \geq r_i} x_{j'} \geq \omega_i$),
%		\item for any $i \in [c]$, sum of weights of battlefields with lower ratio than $r_i$ is at least $u_i - g_1(i)$ (i.e., $\sum_{j' \in [j] : w_{j'}/x_{j'} < r_i} w_{j'} \geq u_i-g_1(i)$),
%		\end{itemize} 
%		}}\vspace{0.2cm}\\
%	0, & \text{otherwise}.
%	\end{cases}
%\end{equation*}
%\begin{equation*}
%	D(j, n', \omega_1, \ldots, \omega_c, u_1, \ldots, u_c) = \begin{cases}
%	1, & \text{\parbox{10.2cm}{if it is possible to put $n'$ troops in the first $j$ light battlefields such that, for any $i \in [c]$, (1) we put at least $\omega_i$ troops in the first $j$ light battlefields with higher ratio than $r_i$; (2) sum of weights of battlefields with lower ratio than $r_i$  is at least $u_i-g_1(i)$.}}\vspace{0.2cm}\\
%	0, & \text{otherwise}.
%	\end{cases}
%\end{equation*}
It is easy to confirm, by definition, that we can satisfy $(\b{x}^h, \b{b}, \b{r})$ iff, $$D(k^\ell, n^\ell, m_1 - x^\ell_{b_1}-1, \ldots, m_c - x^\ell_{b_c}-1, 0, \ldots, 0) = 1.$$ To further maximize the payoff that player 1 gets against the weaker adversary as well as satisfying the given triplet it suffices to find the maximum value of $u$ where $$D(k^\ell, n^\ell, m_1 - x^\ell_{b_1} - 1, \ldots, m_c - x^\ell_{b_c} - 1, u, \ldots, u) = 1.$$

\paragraph{Base case.} We start with the base case, where $j=0$. Here, clearly, if $\omega_i > 0$ or $u_i > 0$ for some $i$, then the value of $D$ must be 0 as we have no way of satisfying (\ref{eq:con2}) or (\ref{eq:con3}). Moreover, if $n' > 0$, the value must again be 0, as (\ref{eq:con1}) requires that we need to spend exactly $n'$ troops over the first $j$ battlefields. Therefore, the only case where the value of the case where $j=0$ is one is where $n' = 0$, and $\omega_i = 0$ and $u_i = 0$ for all $i \in [c]$.

\paragraph{Updating the DP.} To update $D(j, n', \omega_1, \ldots, \omega_c, u_1, \ldots, u_c)$, we only have to decide on how many troops to put on the $j$th battlefield. The idea is to try all possibilities and check, recursively, whether any of these choices satisfies the requirements for thy dynamic value to be 1. Let us denote by $P = \{0, \ldots, n'\}$ the set of all possibilities for the number of troops that we can put on battlefield $j$. We update $D$ as follows:
\begin{equation*}
	D(j, n', \omega_1, \ldots, \omega_c, u_1, \ldots, u_c) = \max_{x \in P} D\Big(j-1, n'-x, \omega_1(x, j), \ldots, \omega_c(x, j), u_1(x, j), \ldots, u_c(x, j)\Big)
\end{equation*}
where for any $i \in [c]$, $\omega_i(x, j)$ and $u_i(x, j)$ are defined as
\begin{equation*}
	\omega_i(x, j) = \begin{cases}
	\omega_i, & \text{if } w_j/x < r_d, \\
	\omega_i - x, & \text{otherwise},
	\end{cases}
	\qquad\qquad \text{and,} \qquad\qquad
	u_i(x, j) = \begin{cases}
	u_i - w_{j}, & \text{if } w_j/x < r_d, \\
	u_i, & \text{otherwise}.
	\end{cases}
\end{equation*}
Recall that when $j = b_{i}$ for some $i \in [c]$, as previously mentioned, to prevent violation of condition 3 of Definition~\ref{def:sat}, we have to put exactly $w_{j} / r_i$ troops on battlefield $j$. If this is the case, and we cannot afford $w_j/r_i$ troops (i.e., if $w_j/r_i > n'$), we update $D$ to be 0. If $w_j/r_i \leq n'$, we make an exception for battlefield $j$ and overload the set $P$ to be $\{ w_{j}/r_i \}$ instead of $\{0, \ldots, n'\}$; the rest of the updating procedure would be the same.

\paragraph{Correctness.} Here, we argue that our updating procedure produces the right answer. We use induction on the value of $j$ and assume that $D(j', n', \omega_1, \ldots, \omega_c, u_1, \ldots, u_c)$ is correctly updated for $j' < j$. We argued why the base case, where $j=0$ is correctly updated. It only suffices to prove for that we correctly update $D(j, n', \omega_1, \ldots, \omega_c, u_1, \ldots, u_c)$. Assume for now that $D(j, n', \omega_1, \ldots, \omega_c, u_1, \ldots, u_c) = 1$. This implies, by definition, that there exists a strategy $(x_1, \ldots, x_j)$ over the first $j$ light battlefields that satisfies (\ref{eq:con1}), (\ref{eq:con2}), and (\ref{eq:con3}). Observe that $0 \leq x_j \leq n'$, and thus,  $x_j \in P$. Once putting $x_j$ troops in battlefield $j$, $\omega_i(x_j, j)$ denotes the number of troops that still has to be put on battlefields with higher ratio than $r_i$. Moreover, $u_i(x_j, j)$ denotes the requirement for the sum of weights of battlefields with lower ratio than $r_i$. Therefore, by (\ref{eq:con2}), and (\ref{eq:con3}) we need to have $D\big(j-1, n'-x_j, \omega_1(x_j, j), \ldots, \omega_c(x_j, j), u_1(x_j, j), \ldots, u_c(x_j, j)\big) = 1$ concluding this case. On the other hand, if $D(j, n', \omega_1, \ldots, \omega_c, u_1, \ldots, u_c) = 0$, no such strategy $(x_1, \ldots, x_j)$ can be found, therefore for all choices of $x_j$ in $P$, the requirements over the prior battlefields cannot be satisfied and $D\big(j-1, n'-x_j, \omega_1(x_j, j), \ldots, \omega_c(x_j, j), u_1(x_j, j), \ldots, u_c(x_j, j)\big) = 0$ for all choices of $x_j \in P$.

\paragraph{Wrap up.} To obtain the strategy $(\b{x}^h, \b{x}^\ell)$ of player 1 that satisfies the triplet $(\b{x}^h, \b{b}, \b{r})$ while providing the maximum guaranteed payoff against all possible strategies of the weaker adversary, we solve the aforementioned dynamic program and linear search over $[0, \sum_i w_i]$ to find the maximum value of $u$ where $$D(k^\ell, n^\ell, m_1 - x^\ell_{b_1} - 1, \ldots, m_c - x^\ell_{b_c} - 1, u, \ldots, u) = 1.$$
	While this only outputs the maximum payoff that can be guaranteed --- and not the actual strategy to provide it --- it is easy to construct the strategy using standard DP techniques. Roughly speaking, to achieve this, we need to slightly modify the DP to also store the actual partial strategy in case its value is 1.
\end{proof}

To conclude this section, for any choice of the triplet $\mathcal{T} = (\b{x}^h, \b{b}, \b{r})$ for which there are only polynomially many options by Claim~\ref{claim:polytriple}, we find the optimal strategy of player 1 among light battlefields that satisfies $\mathcal{T}$ if possible. Lemma~\ref{lem:findstrata} guarantees that the obtained strategy over the light battlefields is indeed the optimal solution if the choice of triplet $\mathcal{T}$ is right. Therefore, it suffices to compare the guaranteed payoff of all obtained solutions for different triplets and report the one that provides the maximum guaranteed payoff for player 1. This procedure gives us the optimal strategy of player 1 against the weaker adversary and therefore it guarantees a payoff of $u$ against him. By Corollary~\ref{cor:weakeradversary}, this gives a $\maxmin{(1-\epsilon/2)u}{1}$ strategy against player 2. Recall that this is achieved over the updated battlefield weights. However, by Corollary~\ref{cor:bfweightrounding}, any $\maxmin{(1-\epsilon/2)u}{1}$ strategy of player 1 on updated battlefields, gives a $\maxmin{(1-\epsilon)u}{1}$ on the original game instance and this proves the main theorem of this section.

\restatethm{\ref{thm:discretepure}}{\thmdiscretepure{}}

\subsection{The Case of 2-Strategies}
\label{sec:twopure}
In this section we generalize the results of the previous section to the case of \cmixedstrategies{2}. To achieve this we first design an algorithm that obtains a \maxmin{u/3}{p} \cmixedstrategy{2}. Then, we use this algorithm and adapt some of the techniques from Section~\ref{sec:one-pure} to give an algorithm that finds a \maxmin{(1-\epsilon)u}{p} \cmixedstrategy{2}. 

The proof of the following theorem which is the main result of this section comes later in the section.
\begin{theorem} \label{theorem:epsapprox2}
			Given that player 1 has a \maxmin{u}{p} \cmixedstrategy{2}, there exists a polynomial time algorithm that obtains a \maxmin{(1-\epsilon)u}{p} \cmixedstrategy{2} of player 1 for any arbitrarily small constant $\epsilon > 0$.
\end{theorem}

%
%\mdcomment{explain how you use \cmixedstrategy{2} , strategy and  $s=(\b{x}, \b{x'})$ }

Given that the first player randomizes over two pure strategies $\b{x}$ and  $\b{x'}$, the second player's best response is a strategy that maximizes his payoff. That is, player 2 seeks to find a strategy $\b{y} \in \puresetb$ such that $\min(\utilityb{\b{x}}{\b{y}}, \utilityb{\b{x'}}{\b{y}})$ is maximized.
In this case, we assume that player 1 cannot guarantee utility $u$ for himself, otherwise he would just play one pure strategy. Therefore, by Observation~\ref{obs:pmorethanhalf} and Observation \ref{obs:plessthanhalf} We can safely assume that in any \maxmin{u}{p} \cmixedstrategy{2} of the first player p=$\sfrac{1}{2}$ holds.

Although, there exist polynomial time algorithms to find the best response of player 2, we still need simpler algorithms to be able to use them and find an optimal \maxminc{(1-\epsilon)u}{p}{2} for player 1. To this end, we define a new opponent which is a weaker version of the second player and we call it the \textit{greedy opponent}. Instead of solving the best response problem optimally, the greedy opponent just takes a simple greedy approach. We prove that playing against this weaker opponent gives us a \maxminc{u/3}{p}{2}. To approach that we first give an alternative formulation of the best response problem in which the solution is represented by two binary vectors. We then relax the integrality condition of elements in the vectors and design a greedy algorithm that finds the best response of the second player in this case. This algorithm is the base for the greedy opponent's strategy. We prove that if the second player takes this approach he loses at most $2\cdot\wmax$ payoff compared to his best response strategy. Recall that $\wmax$ is the maximum weight of the battlefields.

To give an alternative formulation of the best response problem, we define a new problem in which any solution is represent  by two vectors of size $k$. Given strategies $\b{x}$ and $\b{x'}$ we first define \textit{cost} vectors $\b{c}$ and $\b{c'}$. For any battlefield $i$, if $x_i \leq x'_i$ then $c_i := x_i$ and $c'_i := x'_i - x_i$. Otherwise,  $c'_i := x'_i$ and $c_i := x_i - x'_i$. Roughly speaking, when $x_i \leq x'_i$ holds, $c_i$ is the number of troops that player 2 needs to put in battlefield $i$ to win strategy $\b{x}$ in this battlefield and $c'_i$ is the number of troops that he should add to win both strategies. The solution to this problem is two vectors $\b{h}$ and $\b{h'}$ that maximizes $\min(\b{w}\cdot\b{h}, \b{w}\cdot\b{h'})$ subject to the following conditions.
\begin{enumerate}
\item For any $i\in [k]$, elements $h_i$ and $h'_i$ are binary variables. 
\item If $x_i \leq x'_i$, then $h'_i \leq h_i$ holds. Otherwise,  $h_i \leq h'_i$ holds.  \label{condition 2}
\item Also, $\b{c} \cdot \b{h} + \b{c'} \cdot \b{h'} \leq m$.  \label{condition 3}
 \end{enumerate}
 Given vectors $\b{h}$ and $\b{h'}$, the solution of this problem, one can find a strategy $\b{y}$ of player 2 where $ u=\b{w}\cdot\b{h}$ and $u'=\b{w}\cdot\b{h}$. (Recall that $u$ and $u'$ respectively denote the utility that strategy $\b{y}$ gets against strategies $\b{x}$ and $\b{x'}$.) We claim that if for any $i\in [k]$, player 2 puts $h_i\cdot c_i + h'_i \cdot c'_i$ troops in battlefield $i$, then $u \geq \b{w}\cdot\b{h}$ and $u' \geq \b{w}\cdot\b{h}$ hold. 
 Note that by condition \ref{condition 3}, he has enough troops to play this strategy.  Without loss of generality, assume that $x_i \leq x'_i$. Therefore, by condition \ref{condition 2}, $h'_i \leq h_i$ holds. There are three possible cases. Either both $h'_i$ and $h_i$ are 0, both are equal to 1, or $h'_i = 0$ and $h_i=1$.  
If $h_i = h'_i = 0$, then $w_i \cdot h_i $ and $w'_i \cdot h'_i$ are both equal to zero. In the case of $h_i = 1$ and $h_i=0$, we have $h_i\cdot c_i + h'_i \cdot c'_i = x_i$. Therefore, in battlefield $i$, player 2 gets utility $w_i$ against strategy $x_i$. Also, if $h_i = h'_i = 1$, the equality $h_i\cdot c_i + h'_i \cdot c'_i = x'_i$ holds which means player 2 gets utility $w_i$ against both strategies in this battlefield.
 Thus, in all three cases, the utility that player 2 gets against strategies $\b{x}$ and $\b{x'}$ in battlefield $i$ is receptively at least $w_i \cdot h_i$ and at least $w'_i \cdot h'_i$. 
 
  Moreover, given any strategy $\b{y}$ of player 2, one can give two valid vectors \b{h} and \b{h'} such that $u = \b{w}\cdot\b{h}$ and $u' = \b{w}\cdot\b{h}$. Iff $y_i \geq x_i$, let $h_i := 1$. Also iff $y_i \geq x'_i$, we set $h'_i := 1$. Therefore,  $u = \b{w}\cdot\b{h}$ and $u' = \b{w}\cdot\b{h}$ hold. We also need to show that $\b{c\cdot h} + \b{c' \cdot h'} \leq m$. For this, it suffices to show $c_i\cdot h_i + c_i' \cdot h_i' \leq y_i$ for any $i\in [k]$. Without loss of generality, assume that $x_i \leq x'_i$. Thus, $c_i = x_i$ and $c'_i = x'_i - x_i$ hold. Note that we can assume that $y_i$ is either 0, $x_i$ or $x'_i$. In all these three cases,  $c_i\cdot h_i + c_i' \cdot h_i' \leq y_i$ holds; therefore, vectors \b{h} and \b{h'} satisfy the necessary conditions.

To sum up, for any solution of the defined problem which we denote by vectors $\b{h}$ and $\b{h'}$ there is a strategy of the opponent with payoff $\b{c} \cdot \b{h} + \b{c'} \cdot \b{h'}$. Also, for any \maxminc{u}{p}{2} of player 2,  there are two vectors $\b{h}$ and $\b{h'}$  that satisfy the necessary conditions of the problem and $\min(\b{w}\cdot\b{h}, \b{w}\cdot\b{h'}) = u$ holds for them. Therefore, this formulation is indeed a valid formulation of the best response problem. From now on, we may use these two formulations interchangeably.

\subsubsection{A $\sfrac{1}{3}$-Approximation}

Any strategy of the player 2 against a \cmixedstrategy{2} of the first player can be represented by two vectors $\b{h}$ and $\b{h'}$. If we relax the integrality constraint of the elements in these vectors, they can be fractional numbers between 0 and 1.  We call such a strategy a fractional strategy. In this section, we give an algorithm to find the best fractional strategy of the opponent. Using this algorithm we give an exact definition of the greedy opponent. Also, using some properties of the algorithm we prove that player 1 achieves a \maxmin{u/3}{p} strategy by finding his best strategy against the greedy opponent.

Given $s=(\b{x}, \b{x'})$ a \cmixedstrategy{2} of the first player, vectors  $\b{h}$ and $\b{h'}$ are a valid representation of a fractional response $\b{y}$ iff: \begin{enumerate}
\item For any $i\in [k]$, $h_i$ and $h'_i$ are real numbers between 0 and 1. 
\item If $x_i \leq x'_i$, then $h'_i \leq h_i$ holds. Otherwise,  $h_i \leq h'_i$ holds.  
\item Also, $\b{c} \cdot \b{h} + \b{c'} \cdot \b{h'} \leq m$.  
 \end{enumerate}
 
The utility that this strategy gets against $\b{x}$ and $\b{x'}$ is respectively $\utilityb{x}{y} = \b{w}\cdot\b{h}$ and $\utilityb{x'}{y}= \b{w}\cdot\b{h'}$. Let ($\b{h}$, $\b{h'}$) be such a strategy. We call an element $h_i$ (or $h'_i$) an \textit{available element} if it is possible to increase $h_i$ by a nonzero amount without changing $h'_i$. Formally element $h_i$ is available iff the following conditions hold: $h_i < 1$ and if $x'_i \leq x_i$, then $h_i \leq h'_i$ holds.  We also call $h_i$ and $h'_i$ \textit{jointly available} iff $h_i<1$ and $h'_i < 1$. In the other words, it is possible to increase both of them by a nonzero amount without violating the necessary conditions on $\b{h}$ and $\b{h'}$. Let \b{c} and \b{c'} be the cost vectors of a \cmixedstrategy{2} of the first player  denoted by $s$. We define two \textit{ratio vectors} of length $k$ for this strategy and denote them by $\b{r}$ and $\b{r'}$. For any battlefield $i\in [k]$, we set $r_i = \frac{c_i}{w_i}$ and $r'_i= \frac{c'_i}{w_i}$. Note that spending $\epsilon$ amount of troop on an available element, $h_i$, increases it by $\epsilon / c_i$ and  increases \utilityb{x}{y} by $\epsilon / r_i$. Roughly speaking, elements with smaller ratios are more valuable.
 
Given a \cmixedstrategy{2} of the first player in an instance of the Colonel Blotto, Algorithm~\ref{alg:2-fractional} find a best fractional response of the second player. This algorithm consists of several iterations. At the beginning, all the elements of vectors \b{h} and \b{h'} are $0$. In each iteration, the algorithm chooses two elements $h_i$ and $h'_j$ where they are independently or jointly available and they minimize $r_i + r'_j$. Then, it increases each one by a nonzero amount such that $\b{h} \cdot \b{w}$ and  $\b{h'} \cdot \b{w}$ increase equally. Note that to increase these elements the algorithm spends an amount of troop that is determined by the cost vectors. This algorithm stops if there is no troop left or it is not possible to increase the elements. Note that the output of this algorithm satisfies equation $\b{h} \cdot \b{w} = \b{h'} \cdot \b{w}$. In the other words, the strategy that this algorithm finds gets the same amount of utility against strategies \b{x} and \b{x'}. In Lemma~\ref{lem:equal-utility}, we prove that any best fractional response of the second player has this property. Also, we prove in Lemma~\ref{lem:3-approx-2-strategy} that this algorithm indeed finds a best fractional response.  

\begin{algorithm}[h]
	\caption{Best fractional strategy of the second player.}
	\label{alg:2-fractional}
	\begin{algorithmic}[1]
		\Statex \textbf{Input:} strategies $\b{x}$ and $\b{x'}$ of the first player and vector $\b{w}$ which is the battlefields' weight vector.
		\Statex \textbf{Output:} two vectors $\b{h}$ and $\b{h'}$ which denote a fractional strategy of player 2.
		\State Let $\b{c}, \b{c'}$ denote cost vectors of $\b{x}$ and $\b{x'}$.
		\State $\b{h}, \b{h'} \gets (0, \dots, 0)$
    		\While{$m' > 0$}
	\State $B_1 \gets \{i : i\in [k], h_i \text{ is available}\}$ and $a\gets \argmin_{i\in B_1}$ $r_i$   \label{line:4}

%		\State $b_1 :=  \argmax_{b\in B_1} \frac{w_i}{c_i} $
		\State $B'_1 \gets \{i : i\in [k], h'_i \text{ is available}\}$ and $b \gets \argmin_{i\in B'_1}$ $r'_i$ \label{line:5}
%		\State $b'_1 :=  \argmax_{b\in B'_1} \frac{w_i}{c'_i} $
		\State $B_2 \gets \{i : i \in [k], h_i \text{ and } h'_i \text{ are jointly available} \}$ and $d \gets \argmin_{i\in B_2}$ $ r_i+ r'_i$  \label{line:6}
%		\State $b_2 :=  \argmax_{b\in B_2} \frac{w_i}{c'_i + c_i} $

\Comment{\textit{In the case of tie, pick the smaller index}}
		\If{ $r_a + r'_b < r_d + r'_d$} 
		\State Maximize $t + t'$ subject to the following conditions: 
			\begin{itemize}[topsep=0pt,itemsep=-1ex,partopsep=1ex,parsep=1ex]\setlength{\itemindent}{0.5in} 
			 \item  After increasing $h_a$ by $t$ and $h'_b$ by $t'$, vectors $\b{h}$ and $\b{h'}$ are valid. 
			\item $t\cdot w_a = t'\cdot w_b$ 
			\item $ t \cdot c_a + t' \cdot c'_b \leq m' $ 
			\end{itemize}
		\State Increase $h_a$ by $t$ and $h'_b$ by $t'$.
		\State $m' \gets m' - (t \cdot c_a + t' \cdot c'_b)$
		\Else

		\State Maximize $t$ subject to the following conditions:
			\begin{itemize}[topsep=0pt,itemsep=-1ex,partopsep=1ex,parsep=1ex]\setlength{\itemindent}{0.5in}
			\item After increasing $h_d$ and $h'_d$ by $t$, vectors $\b{h}$ and $\b{h'}$ are valid.  
			\item $ t \cdot (c_d + c'_d) \leq m' $	
			\end{itemize}
		\State Increase both $h_d$ and $h'_d$ by $t$.
		\State $m' \gets  t \cdot (c_d + c'_d)$

		\EndIf
		\EndWhile
		\State \Return $\b{h}, \b{h'}$
	\end{algorithmic}	
\end{algorithm}

\begin{lemma}\label{lem:equal-utility}
If vectors $\b{h}$ and $\b{h'}$ represent a best fractional strategy of the player 2 against an arbitrary 2-strategy of the first player, then $\b{w}\cdot\b{h} = \b{w}\cdot\b{h'}$ holds.
\end{lemma}
\begin{proof}
Let $\b{y}$ be a best fractional strategy of the second player represented by $\b{h}$ and $\b{h'}$ where  $\b{w}\cdot\b{h} \neq \b{w}\cdot\b{h'}$. Without loss of generality, assume  $\b{w}\cdot\b{h} > \b{w}\cdot\b{h'}$. We show that it is possible to increase  $\min(\b{w}\cdot\b{h}, \b{w}\cdot\b{h'})$ by modifying \b{y}. Since $\b{w}\cdot\b{h} > \b{w}\cdot\b{h'}$ holds, there exists a battlefield $i\in [k]$ where $h_i > h'_i$. We can decrease $h_i$ by a nonzero amount and increase $h'_i$ using the extra amount of troop achieved from that. This modification of  $h_i$  and $h'_i$ increases   $\min(\b{w}\cdot\b{h}, \b{w}\cdot\b{h'})$ by a nonzero amount which is a contradiction with the assumption that \b{y} is a best fractional strategy.
\end{proof}

\begin{lemma} \label{lem:3-approx-2-strategy}
Algorithm~\ref{alg:2-fractional} gives a best fractional response of the second player.
\end{lemma}

\begin{proof}
Let strategy \b{y} denote the  output of Algorithm~\ref{alg:2-fractional}. 
 Assume $\b{y}$ is not a best response of the second player  and let strategy $\b{y'}$ represented by a pair of vectors $\b{\eta}$ and $\b{\eta'}$ be a best response of player 2 against the  \cmixedstrategy{2} ($\b{x}$, $\b{x'}$).   Let $t$ denote the first iteration of the algorithm after which there exists at least a battlefield $i\in [k]$ where $\eta_i < h_i  $ or $ \eta'_i < h'_i$ (Vectors \b{h} and \b{h'} are defined in Algorithm~\ref{alg:2-fractional}.) Let $i$ and $j$ denote the battlefields where $h_i$ and $h'_j$ are increased in the $t$-th iteration of the algorithm. (It is possible that $j=i$.) 
 Assume that among all the best strategies of the second player $\b{y'}$ is the strategy that minimizes $t$. Also among all those with minimum $t$, $\b{y'}$ is the one that minimizes $(h_i - \eta_i) +  (h'_j- \eta'_j)$. We show that given such a strategy we can modify it in a way that $t$ or $(h_i - \eta_i) +  (h'_i- \eta'_i)$ decreases which is a contradiction. Note that $\utilityb{\b{x}}{\b{y}} = \utilityb{\b{x'}}{\b{y}}$ holds since each iteration of the Algorithm~\ref{alg:2-fractional} increases them by the same amount. Therefore, strategy $\b{y'}$ should perform  better than $\b{y}$ against both $\b{x}$ and $\b{x'}$. More formally, $\utilityb{\b{x}}{\b{y}} < \utilityb{\b{x}}{\b{y'}}$ and   $\utilityb{\b{x'}}{\b{y}} < \utilityb{\b{x'}}{\b{y'}}$ hold. There are two possible cases for the relation between $\b{h}$ and $\b{h'}$ after iteration $t$. In both cases we prove that it is possible to modify $\b{y'}$ and lower $(h_i - \eta_i) +  (h'_i- \eta'_i)$ or $t$ while it is still a best strategy of the second player.

\paragraph{Case 1:} Both $h_i > \eta_i$ and $h'_j > \eta'_j$ hold. In this case, strategy \b{y} achieves more utility against \b{x} in battlefield $i$ and  against \b{x'} in battlefield $j$ than strategy \b{y'} does. Therefore, strategy \b{y'} compensates this by putting more troops in other battlefields. If there is a battlefield $i'\in [k]$, where $h_{i'} < \eta_{i'}$ and  $h'_{i'} < \eta'_{i'}$ then it is possible to modify strategy \b{y'} by decreasing $\eta_{i'}$ and  $\eta'_{i'}$ by a nonzero amount and increasing $\eta_i$ and  $\eta'_j$  such that the amount of troops used by strategy \b{y'} does not increase and its utility does not decrease. It is possible since Algorithm~\ref{alg:2-fractional} in each iteration,  among all the available elements chooses a pair that spending a unit of troop in them gives us the maximum amount of utility, and elements 
$h_{i'}$ and $h'_{i'}$ are both available at this iteration. If this does not hold then there are two battlefields $i', j'\in [k]$ where $h_{i'} < \eta_{i'}$ and  $h'_{j'} < \eta'_{j'}$. In this case, the only difference is that we need to make sure that both $h_{i'}$ and $h'_{j'}$ are available at this iteration which can be inferred from the fact that for both $i'$ and $j'$,  equations  $h'_{i'} \geq \eta'_{i'}$ and  $h_{j'} \geq \eta_{j'}$ hold.

\paragraph{Case 2:} Exactly one of $h_i > \eta_i$ and  $h'_j > \eta'_j$ holds. Without loss of generality, we assume $h_i > \eta_i$ holds. Note that by Lemma~\ref{lem:equal-utility} the amount of utility that strategy \b{y'} gets against strategies \b{x} and \b{x'} is equal. Also, since $h_i > \eta_i$ and  $h'_j \leq \eta'_j$ hold, there exists a battlefield $i'\in [k]$ where $\eta_{i'} - h_{i'} >  \eta'_{i'} - h'_{i'}$. This means that $h_{i'}$ is available at iteration $t$. However, in each iteration, Algorithm~\ref{alg:2-fractional}, chooses a pair of elements in \b{h} and \b{h'} where the amount of utility achieved per spending a unit of troop in their combination is maximized. This yields that it is possible to modify \b{y'} by decreasing $\eta_{i'}$ and increasing  $\eta_{i}$ in a way that $\b{\eta'}$ and $\b{\eta}$ are still a valid representation of strategy \b{y'} and \utilityb{\b{y}}{\b{x}} is not decreased.

In both cases we proved that it is possible to modify strategy \b{y} such that either $t$ decreases by 1 or $(h_i - \eta_i) +  (h'_j- \eta'_j)$  decreases while $t$ is not changed,  which is a contradiction with the assumptions that we had on strategy \b{y'}. Thus, strategy \b{y} the output of Algorithm~\ref{alg:2-fractional} is indeed the best fractional response of the second player
\end{proof}

To give a 2-approximation of the first player's best \cmixedstrategy{2}, we first define the greedy opponent. The exact definition is given in Definition~\ref{def:greedy opponent}. His response is to find a fractional response using Algorithm \ref{alg:2-fractional} and round it down to an integral one. In Observation~\ref{lem:two-different} we prove that the greedy opponent's response differs from the fractional response in at most two battlefields. Note that if $\wmax > u/3$ holds, the first player simply achieves \maxmin{u/3}{p} \cmixedstrategy{2} by putting all his troops on the battlefield with wight $\wmax$. For the simplicity throughout the paper we assume $\wmax \leq u/3$ holds. This yields that the best \cmixedstrategy{2}  of the first player against such an  opponent is a \maxmin{u/3}{p} \cmixedstrategy{2}. Thus, in the rest of this section we focus on finding the best strategy of player 1  against the greedy opponent.

\begin{definition}[Greedy Opponent] \label{def:greedy opponent}
Greedy opponent is a weaker version of the second player. The response of this opponent against a \cmixedstrategy{2} of the first player, $\b{s}=(\b{x}, \b{x'})$, is as follows. 
 Using Algorithm~\ref{alg:2-fractional}, he first finds a best fractional response of the second player against $\b{s}$ , which is denoted by vectors $\b{h}$ and $\b{h'}$. Since it is a fractional response some of the elements in these vectors are fractional. The greedy opponent rounds down the elements of these vectors and plays according to the rounded vectors. 
	
\end{definition}

\begin{observation}\label{lem:two-different}
Let strategy $\b{s}$  be the fractional response of player 2 against a given \cmixedstrategy{2} of the first player that is obtained from Algorithm~\ref{alg:2-fractional}. Also, let $\b{s'}$  denote the response of the greedy opponent against the same strategy of the first player. Strategies $\b{s}$  and $\b{s'}$  put different amount of troops in at most two battlefields. 
\end{observation}

\begin{proof}
	We prove in Lemma~\ref{lem:theverygoodlem} that there are at most two battlefields $i, j\in [k]$ where $h_i, h_j, h'_i$ or $h'_j$ are fractional. Since the greedy opponent rounds down these vectors, his strategy differs from the fractional strategy in at most two battlefields.
\end{proof}

In the following lemma we present three import properties of Algorithm~\ref{alg:2-fractional} which we use later to prove the main theorem of this section.
  
 \begin{lemma} \label{lem:theverygoodlem}
 Consider $a, b$ and $c$ that are defined in lines \ref{line:4}, \ref{line:5} and \ref{line:6} of Algorithm~\ref{alg:2-fractional}.
 \begin{enumerate}[topsep=0pt,itemsep=-1ex,partopsep=1ex,parsep=1ex]
 \item At the beginning of each iteration of Algorithm~\ref{alg:2-fractional}, any fractional element of $\b{h}$ and $\b{h'}$ is in $\{h_a, h'_b, h_c, h'_c\}$.
 \item Also, at the beginning of each iteration at most one of $h_a, h'_b, h_c$ or $h'_c$ is fractional.
 \item After the last iteration of the algorithm there are at most two indices $i, j\in [k]$ where $h_i, h_j, h'_i$ or $h_j$ is fractional.
 \end{enumerate}
 
 \end{lemma}
 
 \begin{proof}
 We first prove that at the beginning of the algorithm, all the elements of \b{h} and \b{h'} other than $h_a, h'_b, h_c$, and $h'_c$ are integral. Let $r$, $r'$ respectively denote the minimum ratios among the available elements of vectors \b{h} and \b{h'} in an arbitrary iteration of the algorithm. It is easy to see that neither $r$ nor $r'$ decreases throughout the algorithm. Thus, any element that is increased before this iteration is either unavailable or is in \{$h_a, h'_b, h_c$, $h'_c$\}. Also note that any element that has value less than $1$ is either available or jointly available. Therefore, any element that is not in \{$h_a, h'_b, h_c$, $h'_c$\} is equal to either $0$ or $1$.
 
Moreover, we show that at the beginning of each iteration of Algorithm~\ref{alg:2-fractional}, at most one of $h_a, h'_b, h_c$ and  $h'_c$ is fractional. We use proof by induction. We prove that if it is not the last iteration of the algorithm, and at the beginning of that at most one element in \{$h_a, h'_b, h_c$, $h'_c$\} is fractional, after this iteration the same holds. It suffices since we know that any fractional element is in \{$h_a, h'_b, h_c$, $h'_c$\}. Note that since at most one element in both vectors is fractional this fractional element is independently available. Therefore, it is either $h_a$ or $h'_b$. If $r_a + r'_b \leq r_c + r_c $ then the algorithm increases $h_a$ and $h'_b$ until one of them is not available anymore (note that it is not the last iteration) which means it is equal to $1$. Otherwise, the algorithm increases both $h_c$ and $h'_c$ by the same amount until one of them becomes equal to $1$. Therefore, in both cases at most one element remains fractional.

Now, consider the last iteration of the algorithm. We proved that at the beginning, at most one element in \{$h_a, h'_b, h_c$, $h'_c$\} is fractional. Also, the fractional element is either $h_a$ or $h'_b$. If at this iteration, the algorithm changes $h_a$ and  $h'_b$, then these are the only ones that can be fractional. Also, if the algorithm changes $h_c$ and $h'_c$, then only $h_c$, $h'_c$ and one of  $h_a$ or $h'_b$ can be fractional. In both cases there are at most two battlefields $i$ and $j$ where $h_i, h_j, h'_i$ or $h'_j$ are fractional.
 \end{proof}

 For any \cmixedstrategy{2} of the first player we define a signature which is determined by Algorithm~\ref{alg:2-fractional} and  the state in which this algorithm terminates.

\begin{definition}[Signature of a 2-strategy] \label{def:sign}
	For any strategy $s = (\b{x}, \b{x'})$ find the best fractional response of the algorithm using Algorithm~\ref{alg:2-fractional} and denote it by vectors $\b{h}$ and $\b{h'}$. consider sets $B_1$, $B'_1$, and $B_2$ in the last iteration of the algorithm. Define \begin{itemize}[topsep=3pt,itemsep=-1ex,partopsep=1ex,parsep=1ex]
 	\item $a := \argmin_{i\in B_1} r_i $, 
 	\item  $b :=  \argmin_{i\in B'_1} r'_i$, and 
 	\item $c:=\argmin_{i\in B_2} r_i + r'_i$.
 \end{itemize}
 Also, let $\mu$ be the total number of troops that are spent on $h_a$, $h'_b$, $h_c$ and $h'_c$. For instance, if $a$, $b$ and $c$ are three different battlefields, $\mu := c_a \cdot h_a + c'_b \cdot h'_b +  c_c \cdot h_c + c'_c \cdot h'_c$. Moreover, denote by $\dot{u_1}$ and $\dot{u_2}$ the utility that the greedy opponent achieves against strategies \b{x} and \b{x'} in these battlefields.
	We define the signature of strategy $s$ to be $(a, b, c, \mu , x_{a,b,c}, x_{a,b,c}', \dot{u_1}, \dot{u_2} )$ and we denote it by $\sign{s}$. Note that we represent $(x_a, x_b, x_c)$ by $x_{a,b,c}$. 
\end{definition}

\begin{lemma}\label{lem:sign-respose}
Let $\sign{}$ be the signature of the greedy opponent's response against a 2-strategy of the first player ($\b{x}, \b{x'}$). There exists a function that determines the response of the greedy opponent in any battlefield $i \in [k]-\{a,b,c\}$ given $x_i$,  $x'_i$ and $\sign{}$. 
\end{lemma}

\begin{proof}
Let $\sign{}$ be $(a, b, c, \mu , x_{a,b,c}, x_{a,b,c}', \dot{u_1}, \dot{u_2} )$. Also without loss of generality, we assume that $x_i \leq x'_i$: therefore,  $c_i = x_i$ and $c'_i = x'_i - x_i$ hold. We claim that Algorithm~\ref{alg:partial-strategy} finds $h_i$ and $h'_i$ given $x_i$, $x'_i$, and signature $\sign{}$. To prove that this algorithm is correct, we need to show that if $h_i$ or $h_j$ is set to $1$ by this algorithm, any greedy response to a strategy with this signature also sets this to $1$. We also need to prove that any element that is set to $0$ is also $0$ any greedy response to a strategy with this signature. 

Note that the minimum ratio of the available (or jointly available) elements does not decrease after each iteration. Also, note that in Algorithm~\ref{alg:2-fractional}, the tie breaking rule while finding the available element with the minimum ratio is the index of elements. Thus, the definition of the signature directly results that any element that is set to $1$ by this algorithm is indeed correct. We just need to show the second part. By Algorithm~\ref{alg:2-fractional}, for battlefield $i$, if $r_i + r'_i > r_c + r'_c$ holds or  $r_i + r'_i = r_c + r'_c$ and $i > c$ hold, $h'_i = 0$.  Also, if in addition to the mentioned condition, $r_i > r_a$ holds and  $r_i = r_a$ and $i > a$ hold then, $h_i = 0$.  The only remaining case is when $r_i < r_a$ and  $r'_i \geq r_b'$ holds. In this case  the algorithm sets $h'_i = 0$ even if $r_i + r'_i \leq r_c + r'_c$. We prove by contradiction that $h'_i = 0$ is correct in this situation. Assume that there exists a response in which $h'_i = 1$. Let $\gamma$ and $\gamma'$ respectively denote the minimum ratios among the ratio of available elements in vectors $\b{h}$ and $\b{h'}$ at the iteration that $h'_i$ first changes. The only way that $h'_i$ changes is as a jointly available element. For this to be possible, conditions  $r_i + r'_i \leq \gamma +  \gamma' $ and  $r_i + \gamma' >  \gamma +  \gamma'$ must hold. This is a contradiction with the fact that minimum ratio of the available  elements does not decrease throughout the algorithm. Thus, Algorithm~\ref{alg:partial-strategy} correctly finds the best response of the greedy opponent in battlefield $i$. 
\end{proof}

\begin{algorithm}
	\caption{Greedy opponent's response in battlefield $i$ given $x_i$, $x'_i$ and $\sign{} = (a, b, c, \mu , x_{a,b,c}, x_{a,b,c}', \dot{u_1}, \dot{u_2} )$}
	\label{alg:partial-strategy}
	\begin{algorithmic}[1]
	\Statex Without loss of generality, this algorithm assumes that $x_i\leq x'_i$. Also $r_i$ and $r'_i$ are $i$-th elements of ratio vectors that can be computed using $x_i, x'_i$ and $w_i$.
	\State $h_i, h'_i \gets 0$
			\If{$r_i < r_a$ \b{or} ($r_i = r_a,  i < a$) }
				\State $h_i \gets 1$
				\If{$r'_i < r_b'$ \b{or} ($r'_i = r_b', i < b $)}
					\State $h'_i \gets 1$
				\EndIf
			\ElsIf{$r_i + r'_i < r_c + r_c'$ \b{or} ($r_i + r'_i = r_c + r_c', i < c$)}
				\State $h_i, h'_i \gets 1$
			\EndIf
		
		\State \Return $h_i$ and $h'_i$
	\end{algorithmic}	
\end{algorithm}

%\begin{algorithm}
%	\caption{\maxminc{u/2}{p}{2} strategy of player 1}
%	\label{alg:2-integral}
%	\begin{algorithmic}[1]
%		\Statex \textbf{Input:} vector $\b{w}$ which is the battlefields' weight vector.
%		\Statex \textbf{Output:} $\b{x}$, a \maxminc{u/2}{p}{2} strategy of the first player.
%		\State \mdcomment{Algorithm goes here}
%		\State \Return $\b{x}$
%	\end{algorithmic}	
%\end{algorithm}

\begin{lemma} \label{lem:poly-sign}
Let $S$ be the set of 2-strategies of the first player in an instance of Discrete Colonel Blotto  and let $\sign{S}$ denote the set of signatures of the strategies in $S$. Size of the set $\sign{S}$ is polynomial and there exists a polynomial time algorithm to find it.
\end{lemma}
\begin{proof}
Let $(a, b, c, \mu , x_{a,b,c}, x_{a,b,c}', \dot{u_1}, \dot{u_2} )$ denote the signature of an arbitrary \cmixedstrategy{2} $s=(\b{x}, \b{x'})$. Note that all possible values of  $a, b, c, x_{a,b,c} $ and $  x_{a,b,c}'$ is bounded by a polynomial in $n$ and $k$. Also, by Lemma~\ref{lem:theverygoodlem}, the amounts of troop that the second player spends in the rest of the battlefields is integral. Therefore, $\mu$ is polynomial in $m$. Thus, we just need to prove that given an assignment to these variables there are polynomially many possible cases for values of $\dot{u_1}$ and $\dot{u_2}$. Let $u_1$ and $u_2$ denote the amount of troops that the the greedy opponent gets in the other $k-3$ battlefields. By Lemma~\ref{lem:equal-utility}, in a valid response of this opponent,  $u_1 + \dot{u_1} = \dot{u_2} + u_2$ holds. Since $u_1 - u_2$ has polynomially many possible values, the same holds for $\dot{u_1} - \dot{u_2}$. We claim  that there is a polynomial time algorithm that finds $\dot{u_1}$ and $\dot{u_2}$ given value of $\dot{u_1} - \dot{u_2}$. By Lemma~\ref{lem:theverygoodlem}, before the last iteration of Algorithm~\ref{alg:2-fractional}, at most one of $h_a, h'_b, h_c$ and $h'_c$ is fractional and all the other elements are integral. We also mention in the proof that this fractional one is either $h_a$ or $h'_b$. This element is responsible for the difference between $\dot{u_1}$ and $\dot{u_2}$; therefore, the value of that is uniquely determined given $\dot{u_1}-\dot{u_2}$. After fixing the value of this fractional element, we find the exact values of $\dot{u_1}$ and $\dot{u_2}$ by simulating the last iteration of the algorithm.
\end{proof}
\begin{theorem} \label{theorem:3-approx}
Given that there exists a \maxminc{u}{p}{2} of the first player, there is an algorithm that finds a \maxminc{u/3}{p}{2} of this player in polynomial time.
\end{theorem}

\begin{proof} By Observation~\ref{lem:two-different}, the strategy of the greedy opponent differs from the fractional best response in at most two battlefields. Also, the best integral strategy gets utility at most equal to the fractional one. Therefore, this weaker opponent loses at most $2\cdot\wmax$ compared to his best response strategy. Note that if losing at most $2\cdot\wmax$ does not guarantee a  \maxminc{u/3}{p}{2} for player 1, he can put all his troops in the battlefield with the maximum weight and just win that one. This yields that the best strategy of the first player against the greedy opponent guarantees a \maxminc{u/3}{p}{2} for him. 
	
	The problem that we need to solve now is how to find the best strategy of player 1 against the greedy opponent. We cannot afford to go over all the possible strategies and search for the best one. However, playing against the greedy opponent gives us the ability to narrow down our search space and find the best response by solving polynomially many dynamic programs. For any signature $\sign{}$, let $S_{\sign{}}$ denote the set of all the first player's strategies with signature $\sign{}$.
		To find the best strategy of player 1 against the greedy opponent, we just need to exhaustively searches through all the possible signatures and for any \sign{} find the best strategy of player 1 in set $S_{\sign{}}$. Note that by  Lemma~\ref{lem:poly-sign} there are polynomially many valid signatures. We give a dynamic program that given a signature $\sign{} = (a, b, c, \mu , x_{a,b,c}, x_{a,b,c}', \dot{u_1}, \dot{u_2} )$, finds the best strategy in set $S_{\sign{}}$. 	By Lemma~\ref{lem:sign-respose}, there exists a function that determines the response of the greedy opponent in any battlefield $i$ given $x_i$,  $x'_i$, and $\sign{}$. Let $f(i, x_i, x'_i, \sign{})$ denote this function.  The dynamic program that we use is $D(j, n_1, n_2, \omega, u_1, u_2)$ with $j \in \{0, \ldots, k-3\}$, $n_1, n_2 \in \{0, , \ldots, n\}$, $\omega \in \{0, 1 , \ldots, m-\mu\}$ and $u_1 , u_2 \in \{0, \ldots, \sum_i w_i\}$. Note that $j$ is in set $\{0, \ldots, k-3\}$ since before the dynamic program we remove three battlefields $a$, $b$ and $c$. The value of $D(j, n_1, n_2, \omega, u_1, u_2)$ is either 0 or 1. It is 1 iff it is possible to give two partial strategies $\b{x} =(x_1, \ldots, x_j)$ and $\b{x'}=(x'_1, \ldots, x'_j)$ over the first $j$  battlefields such that all the following conditions are satisfied. Let $\b{h}$ and $\b{h'}$ denote the partial response of the greedy opponent to \cmixedstrategy{2} $(\b{x}, \b{x'})$ where for any $i$, $h_i$ and $h'_i$ are determined by $f(i,x_i,x'_i,\sign{})$.		\begin{itemize}
			\item Condition 1: $\Sigma_{i\in[j]} x_i = n_1$ and $\Sigma_{i\in[j]} x'_i = n_2$. This is a condition on the amount of troops that the first player uses in the first $j$ battlefields.
			\item Condition 2: $\b{w}\cdot \b{h} = u_1 $ and $\b{w}\cdot \b{h'} = u_2 $. This condition guarantees that the utility that the greedy opponent gets against strategies \b{x} and \b{x'} in the first $j$ battlefield is respectively $u_1$ and $u_2$. 
			\item Condition 3: $\b{c}\cdot \b{h} + \b{c'} \cdot \b{h'} = \omega $, where $\b{c'}$ and $\b{c}$ are partial cost vectors of $\b{x}$ and $\b{x'}$. This condition guarantees that the greedy opponent uses $\omega$ troops in the first $j$ battlefields.
		\end{itemize}	
			
\paragraph{Base case.} We start with the case where $j=0$. Here, clearly, $D(0, n_1, n_2, \omega, u_1, u_2)= 1$ iff $n_1, n_2, \omega, u_1$ and $u_2$ are all 0.
			
\paragraph{Updating the DP.} To update $D(j, n_1, n_2, \omega, u_1, u_2)$, we only have to decide on how many troops to put on the $j$-th battlefield in strategies $\b{x}$ and $\b{x}'$. We try all the possibilities for values of variables $x_j$ and $x'_j$ and check, recursively, whether any of these choices satisfies the requirements for the dynamic value to be 1. Let $\xi$ and $\xi'$
denote an assignment to $x_j$ and $x'_j$. Also, let $\omega'$ denote the amount of troops that the greedy opponent puts in battlefield $j$ if $x_j = \xi $ and $x'_j = \xi'$. Moreover, $u'_1$ and $u'_2$ respectively denote the utility that he gets against strategies $\b{x}$ and $\b{x}'$ in this battlefield. Note that $u'_1$, $u'_2$ and $\omega'$ can be achieved using function $f$.
 $D(j, n_1, n_2, \omega, u_1, u_2) =1$ iff there exist at least a pair of $\xi \in \{0, \dots, n_1 \}$ and $\xi' \in \{0, \dots, n_2 \}$ where $$ D(j-1, n_1-\xi, n_2-\xi', \omega-\omega', u_1-u'_1, u_2-u'_2)=1.$$

After finding the value of $D$ for all the valid inputs, we need to identify the strategies that are in set $S_{\sign{}}$ to be able to find the best one. Roughly speaking, one may think that to find the best strategy of player 1 in set $S_{\sign{}}$,  it is enough if using the dynamic data we find the strategy that minimizes the utility that the greedy opponent gets in $k-3$ remaining battlefields. However, it is not true since using this method we may end up finding a strategy that is apparently very good for the first player but does not have the same signature as $\sign{}$.  This may happen since in the dynamic we do not consider the fact that  the greedy opponent finds a strategy in which he gets the same utility against both strategies $\b{x}$ and $\b{x'}$. To avoid that, before finding the best strategy using the dynamic data we need to filter out the strategies that are not in $S_{\sign{}}$. Recall that signature $\sign{}$ is $(a, b, c, \mu , x_{a,b,c}, x_{a,b,c}', \dot{u_1}, \dot{u_2} )$. For any pair of $u_1$ and $u_2$ that $D(k-3, n_1 , n_2, m-\mu, u_1, u_2) = 1$ there exists a \cmixedstrategy{2} $\b{s} \in S_{\sign{}}$ against which the greedy opponent gets utility $u_1 + \dot{u_1}$  iff $u_1 + \dot{u_1} = u_2 + \dot{u_2}$ since any response of the greedy opponent to a \cmixedstrategy{2} $\b{s}=(\b{x}, \b{x'})$ achieves the same utility against both strategies $\b{x}$ and $\b{x'}$. Therefore, to find the best \cmixedstrategy{2} in $S_{\sign{}}$, using the dynamic data we find a pair of $u_1$ and $u_2$ such that  $D(j, n_1, n_2, \omega, u_1, u_2) =1$ and $u_1 + \dot{u_1} = u_2 + \dot{u_2}$ subject to minimizing $u_1 + \dot{u_1}$. This means that there exists a 2-strategy $\b{s}$ of the first player against which the greedy opponent achieves $u_1 + \dot{u_1}$ utility.

By Lemma~\ref{lem:poly-sign}, there are polynomially many different valid signatures and there is a polynomial time algorithm to find them. Therefore, one can find the best strategy of the first player against the greedy opponent by going over all the possible signatures and finding the best \cmixedstrategy{2} using the described dynamic program. Note that having a best strategy against the greedy opponent gives a \maxminc{u/3}{p}{2} of player~1.
\end{proof}

\subsubsection{A $(1-\epsilon)$-Approximation}
 In Section~\ref{sec:one-pure} we give a $(1-\epsilon)$-Approximation for the case of one pure strategy. Here, we use the same idea and adapt it for our purpose.  Recall that, we first define an updated wight vector by rounding down the weight of battlefields to a power of $(1+\epsilon)$. Then, we partition them to two sets of heavy and light battlefields with threshold $\tau = \epsilon u/4$. Corollary~\ref{cor:bfweightrounding} implies that finding a \maxmin{(1-\epsilon/2)u}{1} \cmixedstrategy{2} after this modification gives us a $\maxmin{(1-\epsilon)u}{1}$ \cmixedstrategy{2}. Therefore, we only focus on the instance with the updated weights. The overall idea is to first define a weaker opponent such that finding the best strategy of the first player against him gives us a \maxmin{(1-\epsilon/2)u}{1} \cmixedstrategy{2}. Then, give a polynomial time algorithm that finds this best strategy. The weaker adversary in the previous section is an opponent whose response against a strategy of the first player is to go over all his pure strategies on the heavy battlefields and for each one, play greedily on the light battlefields. The weaker adversary that we define against \cmixedstrategies{2} is almost similar to the one in the previous section. The only difference is on the greedy algorithm that he uses on the light battlefields which is to respond as a greedy opponent (with a minor technical modification that we explain later).  Note that the greedy opponent loses  at most $2\cdot\wmax$ compared to his best response; thus, playing against this weaker opponent gives a $\maxmin{(1-\epsilon/2)u}{1}$ \cmixedstrategy{2} of the first player. 
  
 To be able to provide a best response algorithm, we first give a new representations for the strategies of both players on heavy battlefields so that we can give proper limits on the number of strategies that they have. By Observation~\ref{obs:constantheavybfw}, the number of distinct heavy battlefields in bounded by a constant. For the case of \cmixedstrategies{1}, we represent first player's strategies by the number of troops that he puts on each battlefield weight.  However, the exact same representation does not work for   \cmixedstrategies{2} of the first player, but we are still able to represent them by vectors of  polynomial length such that the number of valid representations is polynomial as well. Note that given a \cmixedstrategy{2} $s=(\b{x}, \b{x'})$ of the first player, for each battlefield $i\in [k]$, either $x_i\leq x'_i$ or $x_i > x'_i$ holds. Based on this we partition the battlefields to two types. A battlefield $i$ is of type~1 if $x_i\leq x'_i$ otherwise it is of type~2. Lemma~\ref{lem:uniquetypes} proves that there is an optimal  \cmixedstrategy{2} $s=(\b{x}, \b{x'})$ of the first player where both strategies \b{x} and \b{x'} put roughly the same number of troops in battlefields of the same type and the same weight.

\begin{lemma}\label{lem:uniquetypes}
	If player 1 has a \maxmin{u}{p} \cmixedstrategy{2} , he also has a \maxmin{u}{p} \cmixedstrategy{2} $s=(\b{x}, \b{x'})$ where for any two battlefields $i$ and $j$ that conditions $w_i = w_j$, $i \leq j$, and $x_i \leq x'_i$ hold we have $0 \leq x_i - x_j \leq 1$ and $0 \leq x'_i - x'_j \leq 1$.
\end{lemma}

\begin{proof}
Let $s=(\b{x}, \b{x'})$ be a \maxmin{u}{P} \cmixedstrategy{2} of player 1. Assume battlefields $i$ and $j$, where $i< j$, are of the same type and the same weight, but conditions $|x_i - x_j| \leq 1$ and $|x'_i - x'_j| \leq 1$ do not hold for them. Let $t = x_i+x_j$ and $t' = x'_i+x'_j$ respectively denote the amount of troops that strategies $\b{x'}$ and $\b{x'}$ put on these battlefields. There are five possible cases for the utility that the second player gets in these battlefields. We show that it is possible to redistribute $t$ and $t'$ on battlefields $i$ and $j$ to satisfy the mentioned conditions without lowering the number of troops that the opponents needs to spend in each case. 
Let $u$ and $u'$ respectively denote the utility that the second player gets against strategies \b{x} and \b{x'} in these battlefields. All the possible cases for value of $u$ and $u'$ are as follows.
\begin{enumerate}
\item  Case 1: $u = w$ and $u' = w$. The second player achives this  by spending $\min(x'_i, x'_j)$ troops which is maximized when $x'_i = \lceil t'/2\rceil$, $x'_j = \lfloor t'/2\rfloor$.
\item  Case 2: $u = 2w$ and $u' = w$. The opponent spends $\min(x'_i+ x_j, x'_j + x_i)$ troops to get this utility. This is maximized when $x'_i = \lceil t'/2\rceil$, $x'_j = \lfloor t'/2\rfloor$,   $x_i = \lceil t/2\rceil$ and $x_j = \lfloor t/2\rfloor$.
\item  Case 3: ($u = 0$ and $u' = 0$), 	 ($u = w$ and $u' = w$) or ($u = 2w$ and $u' = 2w$). These cases are independent of the distribution. The opponent spends respectively $0$, $t$ and $t+t'$ troops to achieve theses results.
\end{enumerate}
Therefore, if a given strategy does not satisfy the desired conditions, we modify it by setting  $x'_i = \lceil t'/2\rceil$, $x'_j = \lfloor t'/2\rfloor$,   $x_i = \lceil t/2\rceil$ and $x_j = \lfloor t/2\rfloor$. As we showed, after this modification there is no possible amount of utility that the second player can get with spending fewer number of troops in battlefields $i$ and $j$. Also this modifidied strategy satisfies $0 \leq x_i - x_j \leq 1$ and $0 \leq x'_i - x'_j \leq 1$.
\end{proof}

We represent any partial \cmixedstrategy{2} of the first player on the heavy battlefields by a vector of length $\poly{(1/\epsilon)}$. Each entry of this vector is associated with a heavy weight. The $i$-th entry consists of a tuple which determines how many battlefields of the $i$-th heavy weight are from type 1, and how many troops each of \b{x} and \b{x'} spend on the battlefields of each type with this weight. Let us denote by $\puresetaheavy{}$ the set of all such vectors. Note that each entry of a vector in this representation has polynomially many possible values. The following is a corollary of this fact.
\begin{corollary}
\label{cor:heavypoly}
	$|\puresetaheavy{}| \leq n^{\poly(1/\epsilon)}$ where $\puresetaheavy{}$ is the set of \cmixedstrategies{2} of the first player on the heavy battlefields. 
\end{corollary}

Also, using the same argument as we had for the case of \cmixedstrategies{1}, one can see that, here as well, it is possible to represent  pure strategies of the second player on heavy battlefields by vectors of constant size. Let $\puresetbheavy{}$ denote the set of all such vectors. It is easy to see that bservation~\ref{obs:constantheavybfw} holds for this case as well, and $|\puresetbheavy{}|$ is bounded by a constant.

As we mentioned, the weaker adversary acts as a greedy opponent on the light battlefields; therefore, he uses Algorithm~\ref{alg:2-fractional} to find a best fractional response on these battlefields. However, there is a technical detail that we should consider here. Let $s= (\b{x}, \b{x'})$ denote the strategy of player 1 and let $\b{y}$ be a partial response of the weaker adversary on heavy battlefields. By Lemma~\ref{lem:equal-utility} any best fractional response gets the same utility against both the strategies of the first player. Note that it is possible that $\utilityb{\b{x}}{\b{y}} =\utilityb{\b{x'}}{\b{y}}$ does not hold.  Thus, the amount of utility that a best fractional response  gets against \b{x} and \b{x'} on the light battlefields differs by $\utilityb{\b{x}}{\b{y}} -\utilityb{\b{x'}}{\b{y}}$. Here, the issue is that in Algorithm~\ref{alg:2-fractional} we assume that at the beginning the amounts of utility that the greedy opponent has against strategies  $\b{x}$ and $\b{x'}$ are both equal to 0. Therefore, at each iteration the algorithm increases them by the same amount. Hence, the same algorithm cannot be used here, but with a slight modification we can use it for the case of $\utilityb{\b{x}}{\b{y}} \neq \utilityb{\b{x'}}{\b{y}}$ as well. Denote by $u$ and $u'$ the amount of utility that the fractional response gets against strategies $\b{x}$ and $\b{x'}$ until an arbitrary iteration of Algorithm~\ref{alg:2-fractional}. If we use it to complete the partial strategy $\b{y}$ on the light battlefields, at the beginning $u = \utilityb{\b{x}}{\b{y}}$ and $u' = \utilityb{\b{x'}}{\b{y}}$. Without loss of generality, assume $u' \leq u $. The modification is as follows. Since we want $u$ and $u'$ to be equal at the end, before we start increasing both of them, we greedily increase the available elements of vector $\b{h}$ (to increases $u$) until $u = u'$ holds. This is a correct approach since in a best fractional response the overall weight of the light battlefields in which $h_i > h'_i$ is at least $u- u'$. Also, any such battlefield is available at the beginning of Algorithm~\ref{alg:2-fractional}. One can easily verify that this modification does not affect the properties that we have proved for the greedy opponent's response. 

Also, recall that signature of any \cmixedstrategies{2} of the first player, defined in Definition~\ref{def:sign}, is unique since the response of the greedy opponent against a \cmixedstrategy{2} is unique as well. However, the signature of a \cmixedstrategy{2} on the light battlefields also depends on the strategy of the opponent on the heavy battlefields. Roughly speaking, a \cmixedstrategy{2}  of the first player  $s = (\b{x}, \b{x'})$, for any $\b{y^h} \in \puresetbheavy{}$ faces a different type of greedy opponent on the light battlefields since the remaining amount of troops and $\utilityb{\b{x}}{\b{y^h}} - \utilityb{\b{x'}}{\b{y^h}}$ changes the strategy of the greedy opponent on the light battlefields. Therefore, in the definition of signature of a \cmixedstrategy{2} on the light battlefields we should also include the type of the opponent which is determined by a strategy of the second player on the heavy battlefields. Given any such strategy, $\b{y^h}$, we denote signature of \cmixedstrategy{2} $\b{s}$ by \signtwo{\b{s}}{\b{y^h}}. Now, we are ready to prove the main theorem of this section which is as follows.

%Let $s^{h}= (\b{x}, \b{x'})\in \puresetaheavy{}$ be a \cmixedstrategy{2} of the first player  on heavy battlefields. Let $S$ denotes the set of all the \cmixedstrategies{2} of the first player that plays strategy $s$ on heavy battlefields. For $s$ we finds a partial \cmixedstrategy{2} on the light battlefields $s^{l}$ such that the combination of these two partial strategies denoted by $(s^{h}, s^{l})$ looses at most \mdcomment{how much?} compared to the best strategy in $S$. For this, we use a dynamic program almost the same as the one that we used in the previous section to find a \maxmin{(1-\epsilon)u}{1} strategy. Note that $|\puresetbheavy{}|$ is bounded by a constant. 

\restatethm{\ref{theorem:epsapprox2}}{Given that player 1 has a \maxmin{u}{p} \cmixedstrategy{2}, there exists a polynomial time algorithm that obtains a \maxmin{(1-\epsilon)u}{p} \cmixedstrategy{2} of player 1 for any arbitrarily small constant $\epsilon > 0$.}
			
\begin{proof}
Let  $s^{h}= (\b{x}, \b{x'})\in \puresetaheavy{}$ be a \cmixedstrategy{2} of the first player on heavy battlefields and let $S(s^{h})$ denote the set of all the \cmixedstrategies{2} of the first player that plays strategy $\b{s^{h}}$ on heavy battlefields. 
We find a partial \cmixedstrategy{2} on the light battlefields $\b{s^{l}}$ such that the combination of these  
\cmixedstrategies{2}, denoted by $\b{s} = (\b{s^{h}}, \b{s^{l}})$, loses at most $\epsilon u/2$ compared to the best strategy in $S(s^{h})$. For any $i\in |\puresetbheavy{}|$, strategy $\b{y^h_i}$ denotes the $i$-th strategy in \puresetbheavy{}. For any $i\in |\puresetbheavy{}|$, we fix a signature \sign{i} and find a \cmixedstrategy{2} $\b{s^{l}}$ such that for any $i$, $\signtwo{(\b{s^h, s^{l}})}{\b{y^h_i}} = \sign{i}$ is satisfied. Also, among all such strategies, we pick the one that minimizes the utility of the weaker adversary against $(\b{s^h, s^{l}})$. Recall that to find a best strategy of the first player against a single greedy opponent we design a dynamic program in the proof of Theorem~\ref{theorem:3-approx}. The main difference here is that we want a strategy that is the best against multiple greedy opponents; therefore, we need to run these dynamic programs in parallel and update them simultaneously. For that, we use an idea that is presented in the previous section. In Section~\ref{sec:one-pure}, we design a dynamic program that finds the best strategy of the first player against multiple weaker adversaries at the same time. Combining this idea with the dynamic program presented in the proof of Theorem~\ref{theorem:3-approx} gives us a polynomial time algorithm to find a best strategy of the first player against the weaker adversary.

Note that |\puresetbheavy{}| is bounded by a constant and by Lemma~\ref{lem:poly-sign} the number of possible signatures is polynomial; thus, all the possible cases that we consider to find the best strategy of the first player against the weaker adversary in $S(s^{h})$ is polynomial as well. Also, size of the set $\puresetaheavy{}$ is polynomial. Hence, we can find the best \cmixedstrategy{2} of the first player against the weaker adversary in polynomial time.
As we mentioned such a strategy is a \maxmin{(1-\epsilon)u}{p} \cmixedstrategy{2} strategy of the first player.
\end{proof}

\subsection{Generalization to the Case of $c$-Strategies for $c > 2$}\label{sec:discg}
In this section the goal is to provide an algorithm that finds a \maxmin{(1-\epsilon)u}{p} \cmixedstrategy{c}
 of the first player given that the existence of a  \maxmin{u}{p} \cmixedstrategy{c} is guaranteed. The result of this section is stated in the following theorem more formally.
 
\begin{theorem}\label{theorem:discretemulti}
Given that there exists a \maxmin{u}{p} \cmixedstrategy{c} of player 1, there is a polynomial time algorithm that finds a  \maxmin{(1-\epsilon)u}{p} \cmixedstrategy{c} of this player.
\end{theorem}

Note that in the case of \maxmin{u}{P} \cmixedstrategy{2} we show that in an optimal strategy the first player plays both strategies with the same probability. It is easy to see that for the case of $c>2$ this does not hold.
However, we show that to find an optimal strategy it suffice to just consider a constant set of probability assignment to the strategies.  By Theorem~\ref{thm:cstrategyprobabilities} given that a \maxmin{u}{p} \cmixedstrategy{c} is guaranteed, there exists an algorithm to construct a set $P_c$ of $O(1)$ profiles in time $O(1)$. Recall that for a mixed strategy $\b{x}$, its profile, denoted by \profile{\b{x}}, is a multisite of probabilities associated to the pure strategies in the support of $\b{x}$. Since $|P_c|$ is bounded by a constant we can solve the problem for all profiles $P_c$ and for each one find the best strategy of player 1 that has this profile. In the rest of this section we assume that a probability assignment to the strategies is given. Let $\b{s} = (\b{x^1},  \dots \b{x^c})$ be a \cmixedstrategy{c} of the first player. For any $i\in [c]$ denote by $p^i$ the probability with which the first player plays strategy $\b{x^i}$.

Similar to the case of \maxmin{(1-\epsilon)u}{p} \cmixedstrategy{2}, we start by decomposing the battlefields to sets of heavy and light battlefields. the threshold for that is $\tau = \epsilon u/2c$. We also round down the wights as mentioned in the previous sections and just solve the problem for rounded wights.  Then, we  give proper limits on the number of possible \cmixedstrategies{c} of the first player  and the strategies of the second player on heavy battlefields to be able to design a polynomial time dynamic program that finds this a $(1-\epsilon)$-approximate \cmixedstrategy{c}. Let $\b{s} = (\b{x^1},  \dots, \b{x^c})$ be a \cmixedstrategy{c} of the first player. We say two heavy battlefields $i$ and $j$ are from the same type if $w_i = w_j$ and there exists a permutation  $(a_1, a_2, \dots, a_c)$  of numbers 1 to $c$ where $x^{a_1}_i \leq \dots \leq x^{a_c}_i$ and $x^{a_1}_j \leq \dots \leq x^{a_c}_j$ hold. We claim if a \maxmin{u}{p} \cmixedstrategy{c} of the first player exists, there also exists a \maxmin{u}{p} \cmixedstrategy{c} in which the number of troops that the first player puts in the battlefields of the same type in each strategy is almost the same. The formal claim is stated in the following lemma.

\begin{lemma}\label{lem:multiuniquetypes}
If player 1 has a \maxmin{u}{p} \cmixedstrategy{c}, he also has a \maxmin{u}{p} \cmixedstrategy{c} $(\b{x^1}, \dots \b{x^c})$ where for any two battlefields of the same type $i$ and $j$ and any $a\in [c]$ we have $0 \leq x^{a}_i - x^{a}_j \leq 1$.
\end{lemma}  

\begin{proof}
Assume a best  \maxmin{u}{p} \cmixedstrategy{c} $\b{s}$ is given that does not satisfy condition $0 \leq x^{a}_i - x^{a}_j \leq 1$ for some $a\in [c]$. For any strategy $\b{x^a}$, let $t^a$ denote the number of troops that this strategy spends in battlefields $i$ and $j$. We modify $\b{s}$ by redistributing these troops. For any $a\in [c]$, set $x_j^{a} = \lfloor t^a/2 \rfloor$ and $x_i^a = \lceil t^a/2 \rceil$. 
We prove that this modification does not lower the amount of troops that the opponent needs to spend in these two battlefields to get any possible amounts of utility from strategies $\b{x^1} \dots \b{x^c}$. Let $w$ denote the weight of battlefields $i$ and $j$. Also we say a strategy $x^a$ is smaller than $x^b$ in these battlefields if $a$ comes before $b$ in the given permutation (that defines their type).
Utility of the second player in these battlefields has the following form: 
there are two strategies $\b{x^a}$ and $\b{x^b}$ where $\b{x^a} < \b{x^b}$ and the second player gets utilities $2\cdot w$ against strategies smaller than or equal to $\b{x^a}$ and  gets utility $w$ against the ones that are between $\b{x^a}$ and $\b{x^b}$. To get this utility the second player spends $\min(x^a_i+ x^b_j , x^a_j + x^b_i)$ number of troops. One can easily verify that our modification does not decrease this. Therefore, any given \maxmin{u}{p} \cmixedstrategy{c} can be transformed to a \cmixedstrategy{c} that gets the same utility and satisfies the mentioned conditions. 
\end{proof}

Note that there are $\poly(1/\epsilon)$ different weights of heavy battlefields. Also, the number of permutation of numbers $1$ to $c$ is $c!$. Therefore, there are $\poly(1/\epsilon)$ types of heavy battlefields. This means that we can represent any \cmixedstrategy{c} of the first player on the heavy battlefields by a vector of length $\poly(1/\epsilon)$. Each element of this vector, for any $i\in [c]$, contains a variable that shows how many troops strategy $\b{x^i}$ puts in battlefields of this type. Let $\puresetaheavy{}$ denote the set of all such partial $\cmixedstrategies{c}$ in this representation. Since the total number of troops is $n$, each variable has $n$ possible values. The following is a corollary of this.

\begin{corollary}\label{obs:puresetaheavypoly}
	$|\puresetaheavy{}| \leq n^{\poly(1/\epsilon)}$.
\end{corollary}

When the goal of the first player is to achieve utility $u$ with probability $p$, he loses if the second player gets utility at least $u$ with probability more than $1-p$. We call a set of strategies in ${x^1, \dots x^c}$ a \textit{losing set} if the sum of probability of the strategies in it is more than $(1-p)$. If the second player gets utility at least $u$ against all the strategies in at least one losing set, the first player loses in the sense that he can not get utility $u$ with probability $p$. Therefore, a \cmixedstrategy{c} of the first player is a \maxmin{u}{p} \cmixedstrategy{c} iff there does not exist a pair of a losing set $L$ and a pure strategy of the second player $\b{y}$ where $\b{y}$ gets utility at least $u$ against all the strategies in $L$. Let $L_c$ be the set of all the losing sets. For any $L\in L_c$ we define a new type of opponent whose goal is to get utility at least $u$ against all the pure strategies in set $L$ and we denote it by $P_L$. Also, recall that to prevent player 1 from achieving a payoff of $u$, any opponent of type $P_L$ can only lose in at most $2c/\epsilon$ heavy battlefields against  any strategy in $L$. We represent any pure strategy of this opponent by a vector of size $\poly(1/\epsilon)$.  Entries of this vector represent the number of battlefields of any type that the second player wins against any strategy of the first player in set $L$. Therefore, each entry contains $c\cdot c!$  numbers where each one  is in $[0, 2c/\epsilon]$. Let $|\puresetbheavy{}(L)|$ be the set of all partial pure strategies of the opponent of type $P_L$ on heavy battlefields. We denote any pure strategy of the second player on heavy battlefields by a pair of $(L, \b{y})$ where $L$ determines the type of the opponent and $\b{y}$ is his partial strategy on the heavy battlefields. Let $|\puresetbheavy{}|$ denotes the set of all such pairs.

\begin{corollary}\label{obs:puresetbheavyconstant}
	$|\puresetbheavy{}| \leq O(1)$.
\end{corollary}

Similar to the Section~\ref{sec:one-pure} and Section~\ref{sec:twopure} we define a weaker adversary such that his best strategy loses at most $u/2c$ compared to the best strategy of the second player. We also design a dynamic program that given a partial \cmixedstrategy{c} of player 1 and a partial strategy of the weaker adversary on the heavy battlefields, finds the best \cmixedstrategy{c} of player 1 against this opponent on the light battlefields. Without loss of generality we assume that the utility of the opponent is the minimum utility that he gets against all the strategies of the first player. The reason is that we have different types of opponents and for each one we fix the strategies that he is playing against. One can verify that combining these with the methods of the pervious sections gives us an algorithm to find a \maxmin{(1-\epsilon)u}{p} \cmixedstrategy{c} of the first player. 

Recall that in the case of \cmixedstrategies{2}, we had an alternative representation of the second player's response which is a pair of binary vectors of length $k$. We define a similar representation for the response of the second player against \cmixedstrategies{c}. Note that  the main representation of his response is by a vector of length $[k]$ in which the $i$-th entry is the number of troops that the second player puts  in the $i$-th battlefield. An alternative representation is to represent any response of player 2 against a \cmixedstrategy{c} $\b{s} = (\b{x^1}, \dots, \b{x^c})$ of player 1 by $c$ vectors of length $k$ which we denote by $(\b{h^1}, \dots, \b{h^c})$. For any $i\in [c]$, vector $\b{h^i}$ determines the strategy of player 2 against strategy $\b{x^i}$. If $h^i_b = 1$ holds for a battlefield $b\in [k]$, the second player wins strategy $\b{x^i}$ in this battlefield. Note that if for a pair of $i, j \in [c]$, $x^i_b < x^j_b$ holds, then $h^j_b = 1$ yields $h^i_b = 1$ which means $h^j_b \leq h^i_b$. We also define $c$ cost vectors $(\b{c^1}, \dots, \b{c^c})$ for the strategy of first player. These cost vectors, are to transform strategies of player 2 between the two representations.  For any battlefield $b$,  let $\pi^b$  be a permutation  of numbers $1$ to $c$ where for any two consecutive elements $i, j$ in that, $x^i_b \leq x^j_b$. We set $c^j_b := x^j_b- x^i_b$. Roughly speaking, assuming that the second player wins strategy $\b{x^i}$ in battlefield $b$, entry $c^j_b$ is the number of troops that he needs to add to this battlefield to win strategy $\b{x^i}$ as well. Also, for $l:=\pi^b_1$  we set  $c^{l}_b:= x^{l}_b$. A set of vectors is a fractional solution to the best response problem iff:
\begin{enumerate}
\item For any pair of $i, j \in [c]$, $x^i_b < x^j_b$  yields $h^j_b \leq h^i_b$.
\item Entries of the vectors are fractional numbers between $0$ and $1$.
\item The amount of troops used by the second player is at most $m$. In the other words, $\Sigma_{i\in [c]} \b{c^i}\cdot \b{h^i} \leq m$. 
\end{enumerate} Such a response is a best response iff $\min_{i\in [c] } \b{w} \cdot \b{h^i}$ is maximized. We define the strategy of the weaker adversary as follows: he searches through all his pure strategies on the heavy battlefields and for each one, finds his best fractional response on the light battlefields. Then, he rounds down the fractional vectors and plays according to the rounded vectors on the light battlefields. Note that any integral response gets utility at most equal to the best fractional strategy. We prove that there exists a best fractional response that the number of battlefields in which at least one of the vectors $\b{h^1}, \dots, \b{h^c}$ is fractional is at most $c$. Thus, the weaker adversary loses at most $\epsilon u/2$ compared to the best integral strategy of the second player.

\begin{lemma}
For any response of the second player \b{y}, let $B_{\b{y}}$ denote the set of battlefields where for any $i\in B_{\b{y}}$ there exists at least one $j\in [c]$ where $h^j_i$ is fractional. There exists a best fractional response of player 2 for which $|B_{\b{y}}|\leq c$.
\end{lemma}

\begin{proof}
  Let \b{y} be a best fractional strategy with minimum $|B_{\b{y}}|$. For any $i\in B_{\b{y}}$ define a vector $\b{v^i}$. For any $j\in [c]$, entry $v^i_j:=1$ if $h^i_j$ is fractional. Otherwise it is 0. By the optimality of the solution, all such vectors are independent. Thus, the number of such vectors is bounded by the dimension which is $c$ here.	
\end{proof}

Let $\b{s} = (\b{x^1}, \dots, \b{x^c})$  be a strategy of player 1 and let $\b{y}$ be the response of the weaker adversary to that. We define the signature of $\b{s}$ to be the set of $c$ battlefields that have at least a fractional element in $\b{y}$ and the number of troops that each player puts in them. We claim that knowing the signature of a strategy and the number of troops that different strategies of player 1 put in a given battlefield, one can uniquely determine the strategy of player 2 in that battlefield. Define the ratio of a subset of strategies in $\b{s}$ to be the minimum amount of troops that one need to add to the fractional solution to increase the utility of the second player against these strategies by a very small fixed amount denoted by $\delta$. The overall idea is that given the signature of a strategy, we can find these ratios for all the subsets and similar to what we do in Algorithm~\ref{alg:partial-strategy} for the case of \cmixedstrategies{2}, these ratios and the index of the fractional battlefields are enough to determine the strategy of player 2 in a given battlefield. Having this function and combining it with the ideas of the previous section (the dynamic program designed in the proof of \ref{theorem:3-approx}) gives us a dynamic program to find the best strategy of player 1 against the weaker adversary.

\section{Extension to Maximin Strategies}\label{sec:exp}
\newcommand{\opt}[0]{\ensuremath{\mathsf{OPT}}}

In this section, we show that our results carry over to the case where our goal is to maximize the guaranteed expected utility. Recall that for the case of \maxmin{u}{p} strategies, we proved in Section~\ref{sec:probdist} that regardless of the game structure, it suffices to only consider a constant number of probability assignments (profiles) to the pure strategies. We used this to first fix the profile and then solve the game by finding the actual pure strategies. Unfortunately, this is not the case when the objective is to maximize the expected utility. However, we show that it is possible to consider only a polynomial number of profiles while ensuring that the found solution among them is a $(1-\epsilon)$-approximation of the actual maximin strategy.

Throughout this section we denote by $\opt$ the guaranteed expected payoff of the optimal maximin \cmixedstrategy{c}. Our goal is to construct a \cmixedstrategy{c} in polynomial time that guarantees an expected utility of at least $(1-\epsilon)\opt$ against any strategy of the opponent for any given constant $\epsilon > 0$. We call this a $(1-\epsilon)$-approximate maximin strategy. We start with the following claim.

\begin{claim}
	Either $\opt = 0$ or $\opt > 1/c$.
\end{claim}
\begin{proof}
	Assume that $\opt > 0$. This means that there exists a set $S$ of $c$ pure strategies, where against any strategy of player 2, at least one of the strategies in $S$ obtains a non-zero utility. Since the battlefield weights are integers, against any strategy of player 2, at least one strategy in $S$ obtains a payoff of at least 1. Now, by playing each of these strategies with probability $1/c$, we guarantee an expected utility of at least $1/c$ against any strategy of the opponent. Hence $\opt > 1/c$.
\end{proof}

Let us denote by $w := \sum_{i\in[k]} w_i$ the sum of all battlefield weights. Our next claim gives a lower bound for the probabilities assigned to the strategies in the support.

\begin{claim}
	For any given $\epsilon > 0$, there exists a $(1-\epsilon)$-approximate maximin strategy for any instance of (continuous or discrete) Colonel Blotto where every strategy in the support is played with probability at least $\frac{\epsilon \opt}{cw}$.
\end{claim}
\begin{proof}
	Consider an optimal maximin strategy. If no strategy in its support is played with probability less than $\frac{\epsilon \opt}{cw}$, we are done. Otherwise, set the probability of all such strategies in the support to be 0 (i.e., remove them from the support). Now consider a strategy of player 2. Each of the removed strategies gets a utility of at most $w$ against this strategy since $w$ is sum of battlefield weights. On the other hand, there are at most $c$ such strategies. Therefore, the overall cost for the expected utility is
	\begin{equation}
		\frac{\epsilon \opt}{cw} \cdot c \cdot w = \epsilon \opt.
	\end{equation}
	This implies that the remaining strategies in the support obtain an expected utility of at least $(1-\epsilon)\opt$ concluding the proof.
\end{proof}

Assuming that $\opt > 0$ (otherwise a single pure strategy without any troops is the solution), by combining the two claims above we get the following observation.
\begin{observation}\label{obs:minprob}
	We can assume w.l.o.g., that the minimum probability is $\Omega(1/w)$.
\end{observation}

We use this observation to consider only $O(\log w)$ probabilities for each strategy, leading to a polynomial number of profiles that have to be considered.

\begin{lemma}\label{lem:polyprobexp}
	For any constant $\epsilon > 0$, and for any instance of continuous or discrete Colonel Blotto, there are only polynomially many profiles among which an $(1-\epsilon)$-approximate maximin is guaranteed to exist. 
\end{lemma}
\begin{proof}
	Suppose that the probabilities are all in set $P = \{p_0, (1+\epsilon)p_0, (1+\epsilon)^2p_0, \ldots, 1\}$ where $p_0$ is the lower bound for minimum probability. We showed in Observation~\ref{obs:minprob} that it suffices to have $p_0 = \Omega(1/w)$, therefore $|P| \leq O(\log w)$ since $\epsilon$ is assumed to be constant. Note that $O(\log w)$ is polynomial in the input size, thus, even if we try $O(\log w)$ possibilities for $c$ strategies, we have to try polynomially many possibilities. It remains to prove that $P$ provides a $(1-\epsilon)$-approximate maximin. Consider an optimal maximin strategy. Round down the probability of each of its strategies to be in set $P$. Clearly, the updated probability of each strategy is more than a $(1-\epsilon)$ fraction of its original probability. Therefore, against every strategy of player 2, the updated strategy with probabilities in $P$ obtains a payoff of at least $(1-\epsilon)\opt$ concluding the proof. 
\end{proof}

We use Lemma~\ref{lem:polyprobexp} to first fix the probabilities that are assigned to the strategies in the support and then construct them. We further need to fix the value of $\opt$ a priori. This can be done via a binary search so long as by having probabilities $p_1, \ldots, p_c$ and the value of $\opt$, we have an oracle that decides whether it is feasible to construct strategies $\b{x}^1, \ldots, \b{x}^c$ that guarantee an expected payoff of at least $(1-\epsilon)\opt$ with these probabilities or not. Therefore it suffices for the continuous and discrete variants of Colonel Blotto to provide this oracle. This is our goal in the next two sections.
%%%%%

\subsection{Continuous Colonel Blotto}\label{sec:expcont}
Given probabilities $p_1, \ldots, p_c$ and the optimal maximin value $\opt$, our goal in this section is to construct $c$ strategies $\b{x}^1, \ldots, \b{x}^c$ for the continuous variant of Colonel Blotto that guarantee a payoff of at least $(1-\epsilon)\opt$ in expectation, against any strategy of player 2 (or report that this is infeasible). As in Section~\ref{sec:continuous}, we start by formulating the original problem as a (non-linear) program.
\begin{equation}\label{prog:expcg1}
\begin{array}{ll@{}ll}
\text{find}  & \b{x}^1, \ldots, \b{x}^c &  &\\
\text{subject to}& x^j_i \geq 0  & &\forall i, j: i\in[k], j\in[c] \\
&                  \sum_{i\in[k]} x^j_{i} \leq n         & &\forall j \in [c]\\
&	\sum_{j\in[c]} p_j \cdot \utilitya{\b{x}^j}{\b{y}} \geq \opt & & \forall \b{y} \in \puresetb{}
\end{array}\end{equation}
We need to better understand the last constraint of the formulation above to be able to solve it in polynomial time. For this, similar to the case of \maxmin{u}{p} strategies, we give an appropriately adapted definition of {\em critical tuples} and combine it with configurations that were introduced in Section~\ref{sec:contcg}.

\begin{definition}[Critical tuples]\label{def:expcritical}
	Consider a tuple $\b{W} = (W_1, \ldots, W_k)$ where each $W_i$ is a subset of $[c]$. We call $\b{W}$ a {\em critical tuple} if and only if we have $
		\sum_{i, j: j \in W_i } p_j w_i < \opt.$
\end{definition}

Recall from Section~\ref{sec:contcg} that a configuration \b{G} is a vector of $k$ matrices $G_1, \ldots, G_k$ which we call partial configurations, where for any $i \in [k]$, and for any $j_1, j_2 \in [c]$, the value of $G_i(j_1, j_2)$ is $\leqc$ if $x^{j_1}_i \leq x^{j_2}_i$ and it is $\geqc$ otherwise.  Furthermore, for configuration $\b{G}$ and critical tuple $\b{W}$, define $z_i(\b{G}, \b{W}) := \argmax_{j: j \not\in W_i}{x^j_i}$. Note that it is crucial that $z_i(\b{G}, \b{W})$ is solely a function of $\b{G}$ and $\b{W}$ (and not the actual strategies $\b{x}^1, \ldots, \b{x}^c$) so long as strategies $\b{x}^1, \ldots, \b{x}^c$ comply with \b{G}.
\begin{equation}\label{lp:expcg2}
\begin{array}{ll@{}ll}
\text{find}  & \b{x}^1, \ldots, \b{x}^c &  &\\
\text{subject to}& x^j_i \geq 0  & &\forall i, j: i\in[k], j\in[c] \\
&                  \sum_{i\in[k]} x^j_{i} \leq n         & &\forall j \in [c]\\
&  \text{ensure that }\b{x}^1, \ldots, \b{x}^c \text{ comply with } \b{G}\\
&	\sum_{i \in [k]} x^{z_i(\b{G}, \b{W})}_i > m & & \text{for every critical tuple $\b{W}=(W_1, \ldots, W_k)$}
\end{array}\end{equation}
This is essentially the same LP as LP~\ref{lp:cg2} of Section~\ref{sec:contcg} except that we use a different notion of critical tuples.

\begin{observation}\label{obs:expequivalencelpprog}
	Suppose that the configuration \b{G} of the optimal solution we seek to find is fixed. Then, the last constraint of LP~\ref{lp:expcg2} is equivalent to the last constraint of Program~\ref{prog:expcg1}.
\end{observation}
\begin{proof}
	Suppose at first that the last constraint of LP~\ref{lp:expcg2} is violated. We show that given that the solution should comply with \b{G}, Program~\ref{prog:expcg1} is infeasible. To do so, consider the critical tuple $\b{W}=(W_1, \ldots, W_k)$ for which we have $\sum_{i \in [k]} x^{z_i(\b{G}, \b{W})}_i \leq m$. Consider the strategy $\b{y}$ of player 2 where $y_i = x^{z_i(\b{G}, \b{W})}_i$. Clearly this is a feasible strategy for player 2 since it requires only $\sum_{i \in [k]} x^{z_i(\b{G}, \b{W})}_i$ troops which is assumed to be no more than $m$. By Definition~\ref{def:expcritical}, the expected utility of any strategy of player 1 that complies with \b{G} is less than \opt{} against $\b{y}$, meaning that Program~\ref{prog:expcg1} is infeasible.
	
	Now suppose that LP~\ref{lp:expcg2} has a feasible solution $\b{x}^1, \ldots, \b{x}^c$. We show that this is also a valid solution for Program~\ref{prog:expcg1}. Assume for the sake of contradiction that this is not true. That is, player 2 has a strategy $\b{y}$ that prevents strategy $\b{x}^1, \ldots, \b{x}^c$ to obtain an expected payoff of $\opt$. For any $i\in[k]$, define $W_i := \{ j: x^j_i > y_i \}$. Since the expected utility of player 1 by playing this strategy is less than $\opt$ against $\b{y}$, we have 
	$\sum_{i, j: j \in W_i } p_j w_i < \opt.$ which means $(W_1, \ldots, W_k)$ is indeed a critical tuple. Consider the constraint of LP~\ref{lp:expcg2} corresponding to this critical tuple, we need to have $\sum_{i \in [k]} x^{z_i(\b{G}, \b{W})}_i \leq m$ and therefore this constraint must be violated; contradicting the fact that $\b{x}^1, \ldots, \b{x}^c$ is a feasible solution of LP~\ref{lp:expcg2}.	
\end{proof}

\begin{lemma}
	LP~\ref{lp:expcg2} can be solved in polynomial time using the ellipsoid method.
\end{lemma}
\begin{proof}
	We only need to give a separating oracle for the last constraint of LP~\ref{lp:expcg2} since there are only polynomially many constraints of other types. Observation~\ref{obs:expequivalencelpprog} shows that this constraint is essentially equivalent to the best response of player 2 which is used in Program~\ref{prog:expcg1}. That is, if we can solve the best response of player 2 in polynomial time, we will be able to check whether any constraint of LP~\ref{lp:expcg2} is violated. We show that in fact, the best-response of player 2 can be solved in polynomial time. To do so, given $c$ strategies $\b{x}^1, \ldots, \b{x}^c$ along with their probabilities $p_1, \ldots, p_c$, we seek to find a strategy $\b{y}$ of player 2 that maximizes his expected utility. For this, one can use a simple knapsack-like dynamic program $D(i, m')$ which essentially represents what expected payoff can be obtained from the first $i$ battlefields given that player 2 uses only $m'$ troops among them. One can easily confirm that this dynamic program can be updated by considering all possibilities of the number of troops over the $i$th battlefield and recurse over the prior battlefields. This gives a polynomial time algorithm for the best response of player 2, and therefore, a polynomial time algorithm for the separating oracle of LP~\ref{lp:expcg2}.
\end{proof}

The lemma above shows that if we are given the actual configuration \b{G}, we can solve the problem in polynomial time. In Sections~\ref{sec:contc2} and \ref{sec:contcg}, we showed how it suffices to only consider a polynomial number of configurations when all battlefields have the same weight (i.e., the uniform variant of the game). The same argument applies to the expected case since the players are still indifferent to the battlefields of the same weight. The generalization to the case of different battlefield weights follows from essentially the same approach described in Section~\ref{sec:contcg}. It suffices to consider the $\delta$-uniform variant of the game for $\delta$ being a relatively smaller error threshold than $\epsilon$. Then on each bucket of battlefields of the same weight, we only count the number of battlefields of each partial configuration instead of specifying the exact partial configuration of each battlefield. Using similar techniques as in Sections~\ref{sec:contc2} and \ref{sec:contcg}, it can be shown that it suffices to consider only a constant number of possibilities per bucket and therefore a polynomial number of configurations in total. This concludes the continuous variant of the game when the objective is to compute an approximate maximin strategy.

\begin{theorem}\label{thm:expcont}
	For any $\epsilon > 0$, and any constant $c$, there exists a polynomial time algorithm to obtain a $(1-\epsilon)$-approximate maximin \cmixedstrategy{c} for player 1 in the continuous Colonel Blotto game.
\end{theorem}

\subsection{Discrete Colonel Blotto}\label{sec:expdisc}

In this section we solve the same problem solved above for the discrete variant of Colonel Blotto. That is, given probabilities $p_1, \ldots, p_c$ and the optimal maximin value $\opt$, our goal is to construct $c$ strategies $\b{x}^1, \ldots, \b{x}^c$ for the discrete variant of Colonel Blotto that guarantee a payoff of at least $(1-\epsilon)\opt$ in expectation, against any strategy of player 2 (or report that this is infeasible).

We show that a similar approach to that of Section~\ref{sec:integral}, with minor changes, can be applied to this case. Recall that the main idea that we used in Section~\ref{sec:integral} was to partition the battlefields into two disjoint subsets of heavy and light battlefields. Then roughly speaking, the idea was to perform an exhaustive search over the heavy battlefields and solve the problem over the light battlefields against a number of weaker adversaries each corresponding to a response of player 2 over the heavy battlefields. Fix $\delta$ to be a relatively smaller error threshold than $\epsilon$. We say battlefield $i$ is heavy if and only if $w_i \geq \delta \opt$. We first show that w.l.o.g. we can assume that the number of heavy battlefields is at most $O(1)$ or otherwise a simple strategy obtains an expected utility of at least \opt{}.

\begin{claim}
	If the number of heavy battlefields is more than $2c^2/(\epsilon\delta)$, there exists an algorithm to find a \cmixedstrategy{c} minimax strategy providing a utility of at least $\opt$ in polynomial time if $m < (1-\epsilon)nc$.
\end{claim}
\begin{proof}
	If the number of heavy battlefields is more than $2c^2/(\epsilon\delta)$, partition them into $c$ disjoint subsets of size $2c/(\epsilon\delta)$. Consider the \cmixedstrategy{c} of player 1 that chooses one of these subsets uniformly at random (i.e., with probability $1/c$) and distributes his troops among its battlefields almost uniformly (i.e., with pairwise difference of at most 1). Observe that even if in one of these strategies, player 1, wins $c/\delta$ heavy battlefields, the expected utility that he gets would be more than \opt, since
	\begin{equation*}
		\frac{1}{c} \cdot \frac{c}{\delta} \cdot \delta \opt = \opt.
	\end{equation*}
	Therefore for player 2 to prevent player 1 from getting an expected utility of \opt{}, he has to win at least $(1-\epsilon)\sfrac{c}{\epsilon\delta}$ battlefields of each strategy which is not feasible for him since he needs to have $m \geq (1-\epsilon)nc$.
\end{proof}

Now, since we bound the number of heavy battlefields by a constant, we have our desired property that the number of relevant responses of player 2 over the heavy battlefields is bounded by a constant. 

\begin{observation}
	Given a strategy of player 1 over the heavy battlefields, there exists a set $S^h_2$ of a constant number of responses of player 2 over heavy battlefields among which an optimal response is guaranteed to exist.
\end{observation}
\begin{proof}
	We can assume that the number of troops that player 2 puts on the $i$th heavy battlefield is equal to the number of troops that player 1 puts in this battlefield in one of his $c$ pure strategies. Therefore on each battlefield player 2 has $c+1$ options. Combined with the fact that the number of heavy battlefields is $O(1)$, this means there are only $(c+1)^{O(1)} \in O(1)$ relevant responses for player 2 over heavy battlefields.
\end{proof}

The observation above implies that the techniques of Section~\ref{sec:discg} where we perform an exhaustive search over the heavy battlefields and solve a dynamic program with as many dimensions as the number of responses of player 2 over the heavy battlefields is essentially feasible. It only remains to define the weaker adversary over the light battlefields. This turns out to be much simpler than the case of finding a \maxmin{u}{p} strategy.

\paragraph{The weaker adversary.} Given strategies $\b{x}^1, \ldots, \b{x}^c$ with probabilities $p_1, \ldots, p_c$, for any $i \in [k]$ and any $j \in [c]$ define $u(i, j)$ to be the expected payoff that player 2 gets by putting exactly $x^j_i$ troops in battlefield $i$. Moreover, we define the ratio $r(i, j)$ to be $u(i, j)/x^j_i$. The first action of the weaker adversary is to choose $i \in [k]$, $j\in[c]$ with maximum ratio $r(i, j)$ and put exactly $x^j_i$ troops in battlefield $i$. Next, for any strategy $j'$, we decrease $x^{j'}_i$ by $x^j_i$. Intuitively, this updates the additional number of troops that the weaker adversary has to put in this battlefield to win strategies with higher number of troops. Now over the updated strategies, the weaker adversary again computes the ratios, picks the one with the higher value and update the strategies accordingly. This continues until the weaker adversary spends all of his $m$ troops.

\begin{lemma}
	By optimizing against the weaker adversary's greedy best response, one can guarantee an expected payoff of $\opt - w_{\max}$.
\end{lemma}

The lemma above shows that if we optimize our strategy against the weaker adversary over the light battlefields we obtain our desired $(1-\epsilon)\opt$ expected utility so long as we allow the adversary to play any arbitrary strategy over the heavy battlefields. We capture the last iteration of algorithm above by {\em signatures} similar to Definition~\ref{def:sign} of Section~\ref{sec:discg}. Then by fixing strategy of player 1 over the heavy battlefields, and the signature of the weaker adversary for each of his constant relevant responses, we construct the optimal strategy of player 1 over the light battlefields in polynomial time using a dynamic program as in Section~\ref{sec:discg}.

\begin{theorem}\label{thm:expdisc}
	For any $\epsilon > 0$, and any constant $c$, there exists a polynomial time algorithm to obtain a $(1-\epsilon)$-approximate maximin \cmixedstrategy{c} for player 1 in the discrete Colonel Blotto game if $n \geq (1+\epsilon) m/c$.
\end{theorem}

Note that for the case where $n < m/c$, the second player can simply put $\max\{ x^1_i, \ldots, x^c_i\}$ troops on battlefield $i$, ensuring that we get expected utility 0.

\section{Further Complexity Results}\label{sec:complexity}
For the purpose of studying its complexity, let us define {\sc Colonel Blotto} as the following computational problem:  Given a description of the discrete Colonel Blotto game --- that is, the integer number of available troops for both players, and integer weights for the $k$ battlefields --- what are the maxmin strategies of the two players in the game whose utility is the probability of winning more than a threshold, say half, of the total weight?   Since the maxmin strategy is an exponential object, we only require the probability with which the last strategy (a specific allocation) is played\footnote{It would be more natural to ask the value of the game; we require instead a component of the maxmin strategy for a technical reason:  in our reduction below for the generalized problem, computing the value is easy: the game is symmetric, and the value is always zero.}.  It is clear that this problem can be solved by the ellipsoid algorithm in $2^{n^{O(1)}}$ (exponential) time, where $n$ is the size of the input.
%, and there is a PTAS for it \cite{ahmadinejad2016duels, behnezhad2017faster}.
We conjecture that the problem is exponential time-complete.

We have been unable to prove this conjecture;  but as a promising start and consolation prize, we can show exponential time-completeness for the following generalization of the problem:  In {\sc General Colonel Blotto}\footnote{According to the Wikipedia, ``General Colonel'' is a extant rank in certain armies.} we are given a function $u_1:{\cal S}_1\times {\cal S}_2$ to the integers; that is, for each allocation of troops, $u_1$ computes the utility of Player 1 (the utility of Player 2 is, as always, its negation).  The function $u_1$ is of course given as a Boolean circuit, $U$, since its explicit form is exponential.  Thus, the circuit $U$ is the input of the problem.
%A special case is {\sc Separable General Colonel Blotto}, in which the function $u_1$ is of the form $u_1({\cal S}_1\times {\cal S}_2) = \sum_{i=1}^k u_1^i(a_i,b_i)$, where $a_i$ is the number of troops allocated by ${\cal S}_1$ to battlefield $i$ and similarly for $b_i$.

We can show the following:
\begin{theorem}\label{thm:complexity}
{\sc General Colonel Blotto} is exponential time-complete.
%(b) {\sc Separable General Colonel Blotto}  is exponential time-complete under polynomial space reductions.
\end{theorem}
\begin{proof}
We start with a problem we call {\sc Succinct Circuit Value}:  You are given a Boolean circuit with $2^n$ gates {\em implicitly} through another circuit $C$ with $n$ inputs.  For each input $i\in[2^n]$, $C$ outputs $2n+3$ bits interpreted as a triple $(t(i),j(i),k(i))$, where $t(i)\in\{0,1,\lor,\land,\lnot\}$ is the type of the gate, and $j(i),k(i)<i$ are the gates that are inputs of gate $i$.  If $t(i)\in\{0,1\}$ then $j(i)=k(i)=0$, and if $t(i)=\lnot, k(i)=0$.  The question asked is, does the output gate $2^n-1$ evaluate to 1?  It follows immediately from the techniques in \cite{papadimitriou1986note} that  {\sc Succinct Circuit Value} is exponential time-complete.  Fo technical reasons, we require that not all inputs are zero (say, $t(0)=1$), a restriction that obviously maintains complexity.

We reduce this problem to another we call {\sc Succinct Linear Inequalities}:  We are given a circuit $C$ which, in input $i,j$ gives the integer entry $A_{ij}$ of an $M\times N$ matrix $A$ --- $j = N+1$ it returns the value $b_i$ of an $M$-vector $b$.  The question is, does the system $Ax=b, x\geq 0$ have a solution?  We claim that this problem is also exponential time-complete, by a simple reduction from {\sc Succinct Circuit Value}, emulating the well known reduction between the non-succinct versions (see for example the textbook \cite{dasgupta2006algorithms}, page 222).  In proof, we know from \cite{papadimitriou1986note} that it suffices for such reductions between succinct problems to work that the corresponding reduction between the non-succinct problems is of a special kind called {\em projection}, and the vast majority of known reductions can be easily rendered as projections.  Again for technical reasons, we modify slightly the construction by adding redundant constraints to the $0,1,\lnot$ gates to make sure that $N=2^n$ and $M=3\cdot 2^n-3$

Finally, we reduce {\sc Succinct Linear Inequalities} to {\sc General Colonel Blotto}.  For this part we follow the surprisingly recent reduction \cite{adler2013equivalence} from linear programming to zero-sum games (a 60 year old reduction due to Danzig was known to be incomplete, but it was far too much technical work to fix it...).  Adler's reduction starts from a system $Ax=b, x\geq 0$ and produces a skew-symmetric payoff matrix 
\[
P=
\left[
\begin{array}{c c c c c}
0 & 0 & A & e & -b \\
0 & 0 & -e^TA & 1 & -e^Tb \\
-A^T & A^Te & 0 & 0 & 1 \\
-e^T & -1 & 0 & 0 & 1 \\
b^T & -b^Te & 0 & -1 & 0 \\ 
\end{array}
\right].
\]
Actually, Adler starts by adding one extra row to $A$ to ensure that $Ax=0$ has no nonzero, nonnegative solutions; however, in our case this is guaranteed by requiring that the original circuit has at least one nonzero input.  He shows that the system $Ax=b, x\geq 0$ has a solution iff the last component of the (symmetric) maxmin strategy of this zero-sum game is nonzero.  

What remains is to label the $L=M+N+3 =4^n$ rows and columns of this matrix by allocations of troops by the two general colonels, and define the utility $u_1$.  The number of battlefields is $4n$, and we define the set of {\em feasible} allocations to be of the form $C(S)$ where $S$ is any subset of the first half of the battlefields.  $C(S)$ assigns one troop to each battlefield in $S$, and each battlefield $j$ such that $j-2n\not\in S$.  That is, troop assignments in the first and the second half of the battlefields complement each other.  Finally, we define the utility function $u_1$:  Given two allocations $A,B$, $u_1(A,B)$ is defined as follows:
\begin{itemize}
\item If both $A$ and $B$ are infeasible, $u_1(A,B)=0$;
\item if $A$ is feasible and $B$ is infeasible, $u_1(A,B)=-c$, where $c$ is larger than any payoff; this way, player 1 is disincentivized from using $A$;
\item similarly, if $B$ is feasible and $A$ is infeasible, $u_1(A,B)=c$;
\item finally, if both $A$ and $B$ are feasible, say $A=C(S)$ and $B=C(T)$ for subsets $S,T$ of $[2n]$, $u_1(A,B)=P_{ij}$, the $(i,j)$-th entry of the payoff matrix $P$ constructed in Adler's reduction, where the binary representation of $i$ is the set $S$ followed by the set $[2n]-S$, and similarly for $j$ and $T$.
\end{itemize}
It is clear that infeasible allocations are dominated, and can thus be eliminated from the game.  The feasible strategies are in one-to-one correspondence with the rows and columns of the matrix $P$, and thus the maxmin of the  {\sc General Colonel Blotto} game is the same as the maxmin of $P$.  Therefore, the last component of the maxmin strategy is nonzero if and only if the original circuit has value one, and the reduction is complete.
\end{proof}

%\input{cg2}

%\clearpage

\bibliographystyle{plain}
\bibliography{references}
	
\end{document}